\documentclass[10pt,twocolumn,notitlepage,aps,pra,showpacs,superscriptaddress,longbibliography,nofootinbib]{revtex4-1} 
\usepackage{epsfig,amssymb,amsmath,mathrsfs,bm,amsthm,graphicx,xcolor}%float,
\usepackage{algorithm}
\usepackage{algpseudocode}
\definecolor{dullmagenta}{rgb}{0.4,0,0.4} 
\usepackage[colorlinks=true,linkcolor=blue,citecolor=red%,urlcolor=ForestGreen
,urlcolor=blue,pdfpagelabels]{hyperref}
\usepackage[normalem]{ulem}

\usepackage{tikz}
\usetikzlibrary{shapes.geometric,shapes.misc,fit,calc}

\makeatletter
\tikzset{
  @arc through/.style 2 args={
    to path={
      \pgfextra
        \pgfextract@process\pgf@tostart{\tikz@scan@one@point\pgfutil@firstofone(\tikztostart)\relax}%
        \pgfextract@process\pgf@tothrough{\tikz@scan@one@point\pgfutil@firstofone#1}%
        \pgfextract@process\pgf@totarget{\tikz@scan@one@point\pgfutil@firstofone(\tikztotarget)\relax}%
        \pgfextract@process\pgf@topointMidA{\pgfpointlineattime{.5}{\pgf@tostart}{\pgf@tothrough}}%
        \pgfextract@process\pgf@topointMidB{\pgfpointlineattime{.5}{\pgf@totarget}{\pgf@tothrough}}%
        \pgfextract@process\pgf@tocenter{%
          \pgfpointintersectionoflines{\pgf@topointMidA}
            {\pgfmathrotatepointaround{\pgf@tothrough}{\pgf@topointMidA}{90}}
            {\pgf@topointMidB}{\pgfmathrotatepointaround{\pgf@tothrough}{\pgf@topointMidB}{90}}}%
        \pgfcoordinate{arc through center}{\pgf@tocenter}%
        \pgfpointdiff{\pgf@tocenter}{\pgf@tostart}%
        \pgfmathveclen@{\pgfmath@tonumber\pgf@x}{\pgfmath@tonumber\pgf@y}%
        \edef\pgf@toradius{\pgfmathresult pt}
        \pgfmathanglebetweenpoints{\pgf@tocenter}{\pgf@tostart}%
        \let\pgf@tostartangle\pgfmathresult
        \pgfmathanglebetweenpoints{\pgf@tocenter}{\pgf@totarget}%
        \let\pgf@toendangle\pgfmathresult
        \ifdim\pgf@tostartangle pt>\pgf@toendangle pt\relax
          \pgfmathsetmacro\pgf@tostartangle{\pgf@tostartangle-360}%
        \fi
        #2%
          \pgfmathsetmacro\pgf@toendangle{\pgf@toendangle-360}%
        \fi
      \endpgfextra
      arc [radius=+\pgf@toradius, start angle=\pgf@tostartangle, end angle=\pgf@toendangle] \tikztonodes
    }},
  arc through ccw/.style={@arc through={#1}{\iffalse}},
  arc through cw/.style={@arc through={#1}{\iftrue}},
}
\makeatother

\usepackage{units}
\usepackage{subcaption} 
\usepackage{ulem}
\newcommand\mpwS[1]{{\let\helpcmd\sout\parhelp#1\par\relax\relax} }
\newcommand{\Z}{\mathbb{Z}}
\newcommand{\R}{\mathbb{R}}
\newcommand{\C}{\mathbb{C}}

\newcommand{\ph}{{\phantom{.\!\!}}}%{{\rm P}}

\newcommand{\settt}[1]{\mathtt{#1}}%[1]{\mathsf{#1}}
\newcommand{\grp}[1]{\mathsf{#1}}
\newcommand{\spc}[1]{\mathcal{#1}}

\newcommand{\Cspace}{{\rm C}}

\def\d{{\rm d}}

\def \vec#1{{\bm #1}}

\newcommand{\Lin}{\mathsf{Lin}}

\def\>{\rangle}
\def\<{\langle}
\def\kk{\>\!\>}
\def\bb{\<\!\<}

\newcommand{\wc}{{\rm wc}}
\newcommand{\ent}{{\rm ent}}
\newcommand{\cov}{{\rm cov}}

\newcommand{\map}[1]{\mathcal{#1}}
\newcommand{\Tr}{\operatorname{Tr}}
\newcommand{\id}{{I}} % Mischa added this: This is idnetity matrix. Alternatively, can use \mathds{1} with \usepackage{dsfont}. Alternativerly use \mathbbm{1} with \usepackage{bbm}. % Mischa added this line

\usepackage{braket}

\newcommand{\St}{{\mathsf{St}}}

\floatname{algorithm}{Protocol}

\newcommand{\app}{\text{Appendix}}

\newcommand{\Mspace}{\vspace{0.18cm}} % Comtomise the space between paragraphs in the text

\newtheorem{theo}{Theorem}

\newtheorem{lemma}{Lemma}
\newtheorem{prop}{Proposition}

\newcommand{\red}[1]{{\color{red} #1}}
\def\qed{$\blacksquare$ \newline}

\usepackage{lipsum}

\begin{document} 
\title{Optimal Universal Quantum Error Correction via Bounded Reference Frames} 

\author{Yuxiang Yang} \affiliation{Institute for Theoretical Physics, ETH Z\"urich, Switzerland}
%\orcid{0000-0002-0531-8929}\email{yangyu@phys.ethz.ch} 
\affiliation{QICI Quantum Information and Computation Initiative, Department of Computer Science, The University of Hong Kong, Pokfulam Road, Hong Kong}

\author{Yin Mo}
\affiliation{QICI Quantum Information and Computation Initiative, Department of Computer Science, The University of Hong Kong, Pokfulam Road, Hong Kong}
%\affiliation{The University of Hong Kong Shenzhen Institute of Research and Innovation, 5/F, Key Laboratory Platform Building, No.6, Yuexing 2nd Rd., Nanshan, Shenzhen 518057, China}

\author{Joseph M. Renes} \affiliation{Institute for Theoretical Physics, ETH Z\"urich, Switzerland}
%\orcid{0000-0003-2302-8025}

\author{Giulio Chiribella}
\affiliation{QICI Quantum Information and Computation Initiative, Department of Computer Science, The University of Hong Kong, Pokfulam Road, Hong Kong}
\affiliation{Department of Computer Science, Parks Road, Oxford, OX1 3QD, UK}
\affiliation{Perimeter Institute for Theoretical Physics, Waterloo, Ontario N2L 2Y5, Canada}
\affiliation{The University of Hong Kong Shenzhen Institute of Research and Innovation, 5/F, Key Laboratory Platform Building, No.6, Yuexing 2nd Rd., Nanshan, Shenzhen 518057, China}

\author{Mischa P. Woods} \affiliation{Institute for Theoretical Physics, ETH Z\"urich, Switzerland}
\email{mischa@phys.ethz.ch}%\email{mischa.woods@gmail.com}

\begin{abstract}  
Error correcting codes with a universal set of transversal gates are a desideratum for quantum computing. Such codes, however, are ruled out by the Eastin-Knill theorem.
Moreover, the theorem also rules out codes which are covariant with respect to the action of transversal unitary operations forming continuous symmetries. 
In this work, starting from an arbitrary code, we construct \emph{approximate} codes which are covariant with respect to the entire group of local unitary gates in dimension $d$ $(<\infty)$, using quantum reference frames. We show that our codes are capable of efficiently correcting different types of erasure errors. When only a small fraction of the $n$ qudits upon which the code is built are erased, our covariant code has an error that scales as $1/n^2$, which is reminiscent of the Heisenberg limit of quantum metrology. When every qudit has a chance of being erased, our covariant code has an error that scales as $1/n$. We show that the error scaling is optimal in both cases. Our approach has implications for fault-tolerant quantum computing, reference frame error correction, and the AdS-CFT duality.
\end{abstract}

\maketitle
 
\section{Introduction}\label{sec:introduction}
Reliable universal quantum computation requires fault-tolerant error correction~\cite{Campbell2017}, 
the 
%In order to build a universal quantum computer, fault-tolerance must be achieved . This is the 
ability to correct errors with gates that are themselves noisy. %in the implementation before they propagate too far in the computation. 
%are allowed to propagate too far through the computation; beyond the point at which they can be corrected. 
Achieving such quantum error correcting codes (QECCs) is a notoriously challenging task, due to fundamental limitations such as quantum no cloning.\Mspace

One of the earliest proposals to achieve fault-tolerance  for universal quantum computation was the idea of implementing all the logical gates `transversally'~\cite{Shor1996,Gottesman2006}, which is the following idea: Given a qudit logical space and an $n$-qudit physical space, find an encoder $\map{E}_\cov$ which maps all logical gates $V\in \grp{SU}(d)$, the group of unitaries in $d$ dimensions, to a tensor product of physical gates:
\begin{align}\label{covariance}
\map{E}_\cov\circ\map{V}_{\rm L}=\left(\map{V}_1\otimes\cdots\otimes\map{V}_n\right)\circ\map{E}_\cov
\end{align}
where $\map{V}_{\rm L}(\cdot)=V(\cdot)V^\dag$ is the unitary channel corresponding to logical gate $V_{\rm L}$, and $\map{V}_m$ $(m=1,\ldots,n)$ is either the unitary channel $V(\cdot)V^\dag$ or the identity channel, on the $m$th physical qudit.\footnote{One may feel that Eq.~\eqref{covariance} should be relaxed to hold only for a finite, universal gate set. However, observe that since the transversality condition is preserved under gate composition, and the set of gates generated by a universal gate set is dense in $\grp{SU}(d)$, these two conditions are effectively equivalent.}
Therefore, in such proposals, all gates required to achieve universal  computation  would be realised at the physical level by applying a tensor product of local gates. This structure is naturally suited to error correction, since an error in one of the physical qudits does not easily propagate to other physical qudits | keeping errors local so that it can be efficiently corrected afterwards.\Mspace

However, as shown by Easting and Knill, there cannot exist a code $\map{E}$ satisfying Eq.~\eqref{covariance}, with a finite dimensional code space that \emph{perfectly} corrects local errors~\cite{eastin2009restrictions}. Moreover, their result actually holds more generally for any set of qudit gates which form a continuous symmetry, that is to say, any set of gates $\mathcal{V}$ which act on qudits and form a Lie subgroup of $\grp{SU}(d)$~\cite{eastin2009restrictions, woods2020continuousgroups,faist2019continuous}.\Mspace

The code space is  simply the image, in the physical space, of the encoding map, while a code that perfectly corrects is simply one for which the decoder $\map{D}_\cov$ can correct any error $\map{C}$ and still decode perfectly any logical state $\rho_{\rm L}$:
\begin{align}
\map{D}_\cov\circ\map{C}\circ\map{E}_\cov\circ \rho_{\rm L}=\rho_{\rm L}.
\end{align}
The Eastin-Knill theorem does not rule out the ability to correct local errors for a non-universal set of gates, however. A case in point is the set of Clifford gates, which can be implemented transversally but lack one crucial gate in order to form a universal set. Supplementing this set with the missing gate is the idea behind one of the frontrunner proposals for universal quantum computation, using so-called `magic states'~\cite{bravyi2005universal}. Several other schemes to relax the transversality condition also exist~\cite{knill1996threshold,bombin2007topological,paetznick2013universal,jochym2014using,bombin2015gauge,yoder2016universal,Browneaay4929}.\Mspace

In this work, we demonstrate a different approach to circumvent the Eastin-Knill theorem. In particular, all gates can be applied transversally under various local error models in our scheme.
The key ingredients are quantum reference frames and randomness. We now introduce the former before explaining their relevance to our scheme. \Mspace

In physics, all observations are reported relative to a reference frame. While the reference frames have traditionally been \emph{treated according to the laws of classical physics}, the usage of quantum states to encode reference frame information (a Cartesian coordinate system, for instance) in quantum superpositions has been shown to be advantageous in problems involving reference frame alignment and overcoming super-selection rules \cite{peres2001transmission,bagan2001aligning,chiribella2004efficient,bartlett2007reference,PhysRevA.69.052326,2013Marvian,PhysRevA.78.022304}.
Quantum reference frames also play a crucial role in demystifying a number of controversies and paradoxes \cite{Bartlett2006,Angelo_2011,Angelo_2012,2011.01951}, in addition to unifying various different paradigms \cite{1912.00033}.  There has been a renaissance very recently, due to generalisation of quantum reference frame transformations to a ``superposition of coordinate transformations" which consistently describe the physics without appealing to an external, absolute reference frame \cite{Giacomini2019,1809.05093,PhysRevLett.123.090404,Vanrietvelde2020changeof,delahamette2021perspectiveneutral,hoehn2021quantum,delaHamette2020quantumreference}.\Mspace

Quantum reference frames in the context of QECCs was first explored in Ref.~\cite{hayden2017error}, where classical idealised reference frames\footnote{An idealised reference frame, also know as a ``perfect'' reference frame, is one whose orientation can be deterministically obtained via measurement which is contrary to a quantum reference frame whose orientation is subject to quantum uncertainty. See \cite{bartlett2007reference} for definition and \cite{woods2020continuousgroups} for further insight.} were employed. The hypothetical setup allows for exact decoding, and does not violate the Eastin-Knill theorem due to the use of an infinite dimensional code space.
While impractical due to infinite dimensionality, the approach in Ref.~\cite{hayden2017error} paved the way for another approach to circumvent the Eastin-Knill theorem: using finite dimensional reference frames and decoders which only recover approximately. 
This route was followed in Ref.~\cite{woods2020continuousgroups}, where finite quantum reference frames have been used to allow for a single Abelian family of transversal gates. This construction, while useful, does not allow for universal quantum computation with transversal gates.

Here we further develop this approach and design a new family of quantum reference frames to achieve QECC constructions with all logical gates being transversal. Quantum reference frames have an inherent quantum uncertainty which unavoidably leads to  a small error in the decoding. Consequently, our quantum error correcting code is only approximate. On a theoretical level, this is essential to circumvent the Eastin-Knill theorem while achieving our objective of a universal quantum gate set. On a practical level, this error can be made smaller than any chosen tolerance by increasing the size of the quantum reference frame.\Mspace

 The physical space necessary for our implementation of the QECC consists of two parts, the \emph{computational space} and the \emph{reference frame space}. Roughly speaking, the computational space is where the logical information is stored and the sequence of logical gates needed for the computation are applied on various copies. The reference frames on the other hand, play the role of recording information about which gate was applied | analogously to how gyroscopes record a Cartesian coordinate system. See Fig. \ref{fig:setup pic}.\Mspace

 At an abstract level, our encoder works by choosing a quantum reference frame which can record the ``coordinates'' of any logical state. This is to say, states which are distinguishable by measurements under transformations of $\grp{SU}(d)$ | the special group of unitary transformations in $d$ dimensions. To implement any logical gate under this encoding, we can simply apply the gate transversally on the computational space in tandem with updating our reference frames with a corresponding transformation.\Mspace

However, before this paradigm can be implemented, one needs to perform the encoding | this is where the randomness comes into play. We pick a logical gate at random, apply it transversally to the computational system and a copy to the reference frame. The reference frame information is retrieved during the decoding stage via measurement.\Mspace

\begin{figure}[]
	\includegraphics[width=1\linewidth]{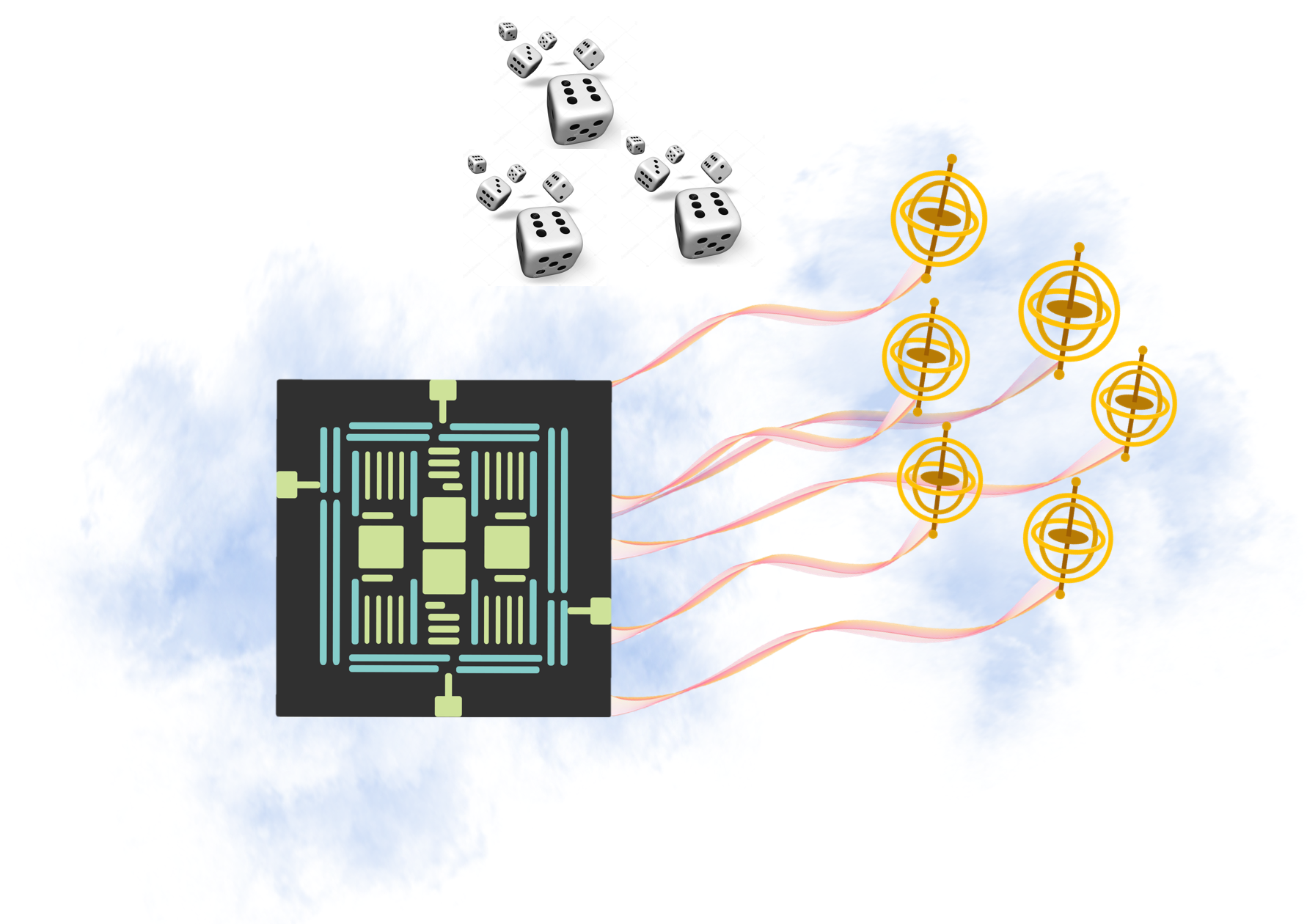}%[width=1.3\linewidth]{Computer_References_imageV1.pdf}
	\caption{\textbf{Error correction via reference frames}. The physical space of the computer is over $n$ qudits which are divided into two subspaces: the computational space (black and green) and a reference frame space (yellow and brown). The former is where the gates needed for the algorithm are applied transversally, while the latter is where the reference frames are located. The two systems are only classically correlated at any given moment in our protocols. \label{fig:setup pic}}
\end{figure}

In the regime of error correction via finite quantum reference frames, the error of approximate recovery is dependent on the size of the reference frames, leading to a trade-off relation.
It is thus of prime importance to identify the optimal scaling of the approximation error with respect to the size of the reference frame, so that the cost of making the recovery errors admissible can be determined. 
Here we analyse the  error scaling of our QECC construction in two different commonly-studied error models and show that our approach is optimal in both cases. Besides the analytical analysis, we also conduct numerical simulation that further justifies our analysis in various practical error models.

In a nutshell, our contribution is to show that the use of quantum reference frames allows for approximate QECC where all gates necessary for universal quantum computation can be applied transversally, and to identify the optimal tradeoff between the approximation error and the size of the reference frames.
%We demonstrate this for relatively ``strong'' error models. 
This was the initial idea of the pioneers of quantum error correction~\cite{Shor1996,Gottesman2006}, who, in light of the Eastin-Knill theorem, invented distinct methods alternative to finding a transversal gate set for quantum error correction. Our work thus offers a new test bed for quantum error correction whose full potential can now be further explored.  An immediate question is whether the encoder and the decoder in our construction can be implemented fault-tolerantly. As discussed in more details later, there are several paths that potentially lead to this ultimate purpose, though the concrete approach is still unknown. This is thus an important open question left for investigation in future research.

\Mspace

\section{Results}
\noindent {\bf Overview.} We start by introducing the relevant definitions to define our encoder, decoder, and characterisation of the decoding error. We then consider two errors models: the \emph{i.i.d.\null{} error model} and the \emph{weak error model} and fully characterise the performance of our QECC under these different scenarios. Details appear in the appendices.\Mspace

\noindent {\bf Construction of Covariant Encoder and Decoder.} A QECC is characterised by an encoder-decoder pair $(\map{E},\map{D})$. The encoder $\map{E}$ is a quantum channel mapping a logical qudit, whose dimension we denote by $d$, to a state of a physical register consisting of several qudits. The decoder $\map{D}$ corrects possible errors in the physical register and maps the state back to a logical qudit.
Since our protocol uses reference frames and randomness to convert an arbitrary encoder $\map{E}$ and decoder $\map{D}$ to a covariant encoder-decoder pair,  we refer to $\map{E}$ and $\map{D}$ as the \emph{subroutine} encoder and decoder. The $n$-qudit physical register is divided into two parts, the \emph{computational register} \Cspace, consisting in $n_\Cspace$ qudits and a \emph{reference frame register} {\rm R}, consisting in the remaining $n_{\rm R}=n-n_\Cspace$ qudits. The subroutine decoder, $\map{D}$, maps to the computational register only.\Mspace

In order to be resilient to errors, the reference frame register may consist of $s_{\rm R}\ge1$ identical blocks, each in the quantum state $|\psi\>$. Here we choose $|\psi\>$ to be an entangled state of a bipartite system $\spc{H}\otimes\spc{H}'$\,\, ($\spc{H}'\simeq\spc{H}$), as they
offer superior performance \cite{acin2001optimal,chiribella2004efficient,chiribella2005optimal}. Each part of the bipartite system further consists of several qudits. This reference frame is a finite dimensional quantum state $|\psi\>$. Rotating every qudit in the part of $|\psi\>$ on $\spc{H}$ by the same (unknown) unitary $U$ in $\grp{SU}(d)$ (while keeping the part on $\spc{H}'$ unchanged) yields a state that, when measured appropriately, allows one to deduce $U$ with high accuracy. The better the quality of $|\psi\>$, the smaller the uncertainty in $U$ is.\Mspace

\begin{figure} 
	\begin{subfigure}[b]{0.5\textwidth}
		\includegraphics[width=1\linewidth]{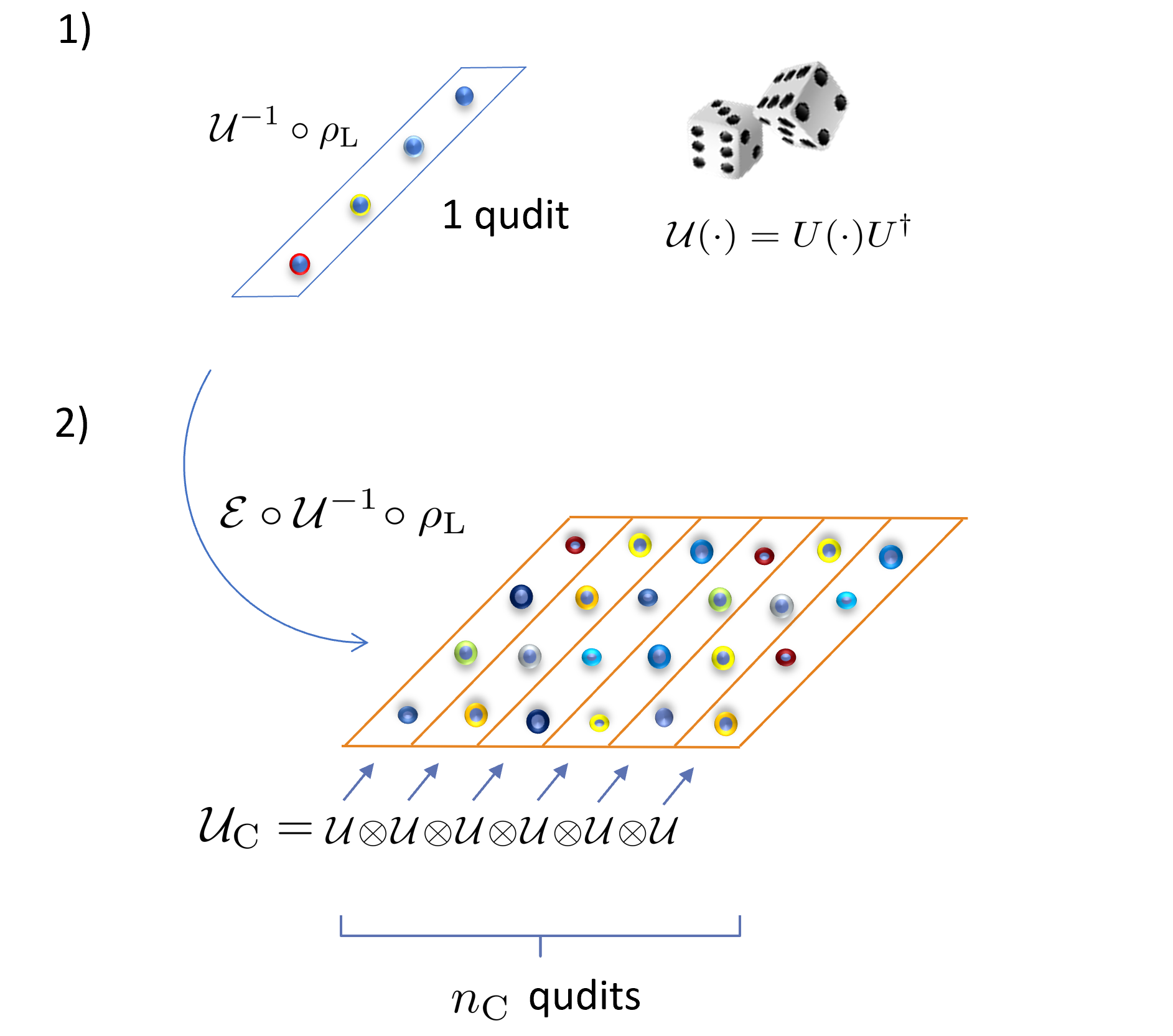}
		%\caption{}
		%\label{fig:Ng1} 
	\end{subfigure}
	
	\begin{subfigure}[b]{0.48\textwidth}
		\includegraphics[width=1\linewidth]{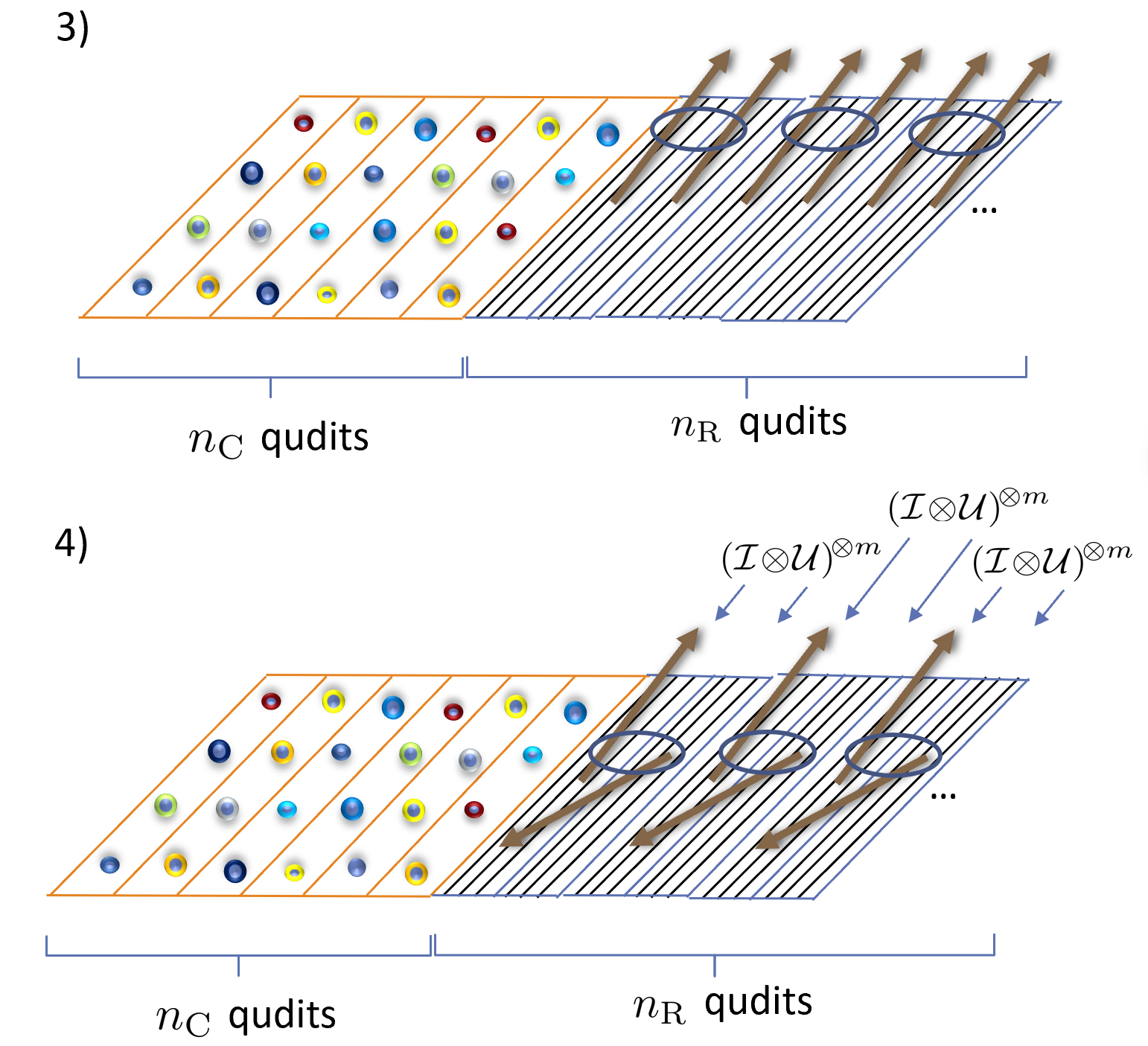}
		%\caption{}
		%\label{fig:Ng2}
	\end{subfigure} 

	\caption[l]{\textbf{Covariant encoder construction}. %
	(The following is an in-principle method to implement our QECC. For practical implementations, see section \emph{Towards fault-tolerant QECCs}.) Given an arbitrary channel $\map{E}$ encoding one logical qudit into $n_{\rm C}$ qudits, and a corresponding decoder $\map{D}$, our covariant encoding protocol runs as follows: 1) Choose a unitary $U\in\grp{SU}(d)$ at random (in practice, this can be achieved by sampling from a large enough finite set of gates) and apply its inverse unitary channel $\map{U}^{-1}_{\rm L}$ to the logical state $\rho_{\rm L}$ one wishes to encode.  2) Perform $\map{U}_\Cspace\circ\map{E}$, where $\map{U}_{\Cspace}$ is the unitary channel corresponding to the representation of $U$ in the computational space, and $\map{U}_\Cspace$ are local gates which will be transversal for the new code. 3) Such an operation alone does not result in a useful code, as the identification of the particular $U$ is needed to later correct errors and decode. Thus, to store $U$ efficiently, we append $n_{\rm R}$ quantum reference frames and 4), apply it to a quantum reference frame, resulting in the unitary's identity being encoded into  the reference frame's state. The outcome of steps 1) to 4) is a new code with a covariant encoder $\map{E}_\cov$ which outputs an encoded state on the original computational and reference frame registers, \Cspace{} and {\rm R}.}\label{eq:fig encoder}
\end{figure}

Our covariant encoder is defined as
\begin{align}\label{E-cov}
\map{E}_\cov(\cdot)=\int\d U~\map{U}_\Cspace\circ\map{E}\circ\map{U}_{\rm L}^{-1}(\cdot)\otimes\map{U}_{\rm R}(\Psi)\,,
\end{align} 
where $\Psi:=|\Psi\>\<\Psi|$ with $|\Psi\>=|\psi\>^{\otimes s_{\rm R}}$ is the initial state of the reference frame register and $\d U$ is the Haar measure on $\grp{SU}(d)$. The unitary channels in Eq. \eqref{E-cov} are defined by $\map{U}_\Cspace:=\map{U}^{\otimes n_\Cspace}$ and $\map{U}_{\rm L}:=\map{U}$ the actions of an element $U\in\grp{SU}(d)$ on the computational register and the logical register, and by $\map{U}_{\rm R}:=(\map{U}\otimes\map{I})^{\otimes n_{\rm R}/2}$ its action on the reference frame register (as one part of the reference frame state is rotated and the other part remains unchanged), where $\map{U}(\cdot):=U(\cdot)U^\dag$. See Fig. \ref{eq:fig encoder} for an example of how to construct the encoder.\Mspace

The covariant decoder consists of first measuring the reference frame with a POVM $\{M_{\hat{U}}\d\hat{U}\}$, which yields an estimate $\hat U$, and then applying  $\hat{\map{U}}_\Cspace^{-1}$. Here $\d\hat{U}$ is again the Haar measure. We then perform the subroutine decoder $\map{D}$ of the original code to correct errors, and finally we redo $\hat{\map{U}}_{\rm L}$ to obtain the final, decoded output. This gives rise to the channel 
\begin{align}\label{D-cov}
\map{D}_\cov(\cdot)= \int \d\hat{U}~\left(\hat{\map{U}}_{\rm L}\circ\map{D}\circ\hat{\map{U}}_\Cspace^{-1}\otimes\map{M}_{\hat{U}}\right) (\cdot)\,,
\end{align}
where the map $\map{M}_{\hat{U}}(\cdot):=\Tr[(\cdot)M_{\hat{U}}]$ outputs the probability density function of obtaining measurement outcome estimate $\hat U$. \Mspace 

One may be concerned about the practicality of constructing Eqs.\ \eqref{E-cov} and \eqref{D-cov}, since sampling from the Haar measure is essentially to randomly apply all unitary gates, which might be computationally hard. In fact, we show later that, using the technique of unitary $t$-designs \cite{PhysRevLett.98.130502,Dahlsten_2007,PhysRevA.80.012304,CMPHarrow,Diniz2011,PhysRevA.78.062329,CMPBrandao,PhysRevX.7.021006}, we only require to sample from a discrete set of unitary gates in a computationally efficient manner.\Mspace

Errors arise in this procedure both from any noise that occurs between encoder-decoder, as well as in the measurement of the reference frame itself. 
%The above procedure implements the desired computation transversally, at the price of introducing a decoding error. 
Nevertheless, we show that the overall error vanishes with the number of qudits upon which we construct the covariant code.\Mspace

As promised, the encoder in Eq.~\eqref{E-cov} satisfies the transversality condition, Eq.~\eqref{covariance}, thanks to the well-known invariance properties of the Haar measure. In our particular case, it takes the form:
\begin{align}
\map{E}_\cov\circ\map{V}_{\rm L}=(\map{V}_\Cspace\otimes\map{V}_{\rm R})\circ\map{E}_\cov
\end{align}
for any $V\in\grp{SU}(d)$, where $\map{V}_\Cspace:=\map{V}^{\otimes n_\Cspace}$ and $\map{V}_{\rm L}:=\map{V}$, $\map{V}_{\rm R}:=(\map{V}\otimes\map{I})^{\otimes n_{\rm R}/2}$ with $\map{V}:=V(\cdot) V^\dag$, and $\map{I}$ the identity channel. Notice that since our decoder is also covariant, i.e.\null{} $
\map{D}_\cov\circ(\map{V}_\Cspace\otimes\map{V}_{\rm R})=\map{V}_{\rm L}\circ\map{D}_\cov\,$, our entire code is covariant: $\map{D}_\cov\circ\map{E}_\cov\circ\map{V}_{\rm L}=\map{V}_{\rm L}\circ\map{D}_\cov\circ\map{E}_\cov\,$.

\medskip

%This enables an alternative realization of a logical gate $V$, which may have the decomposition $V=V_N\cdots V_2V_1$: Instead of doing it in a single round of encoding and decoding, one can run $N$ rounds of encoding and decoding, each with $(\map{V}_i)_{\rm C}\otimes(\map{V}_i)_{\rm R}$ in between (with $i=1,\dots,N$).

\noindent {\bf Reference Frame structure.}
There is a trade-off: the optimal pair of the initial reference frame state $\Psi$ and the POVM $\{M_{\hat{U}}\d\hat{U}\}$, which yields the smallest decoding errors when no noise is present, tends not to be very resilient to noise. Conversely, those which are highly resilient to noise, tend to lead to substantial decoding errors. As such, the optimal reference frame choice depends on the noise model. We will only provide the general structure considered here, and specialise later as specific noise models are introduced. \Mspace

We divide the $n_{\rm R}$ reference frame register qudits into $s_{\rm R}$ blocks of identical systems, i.e., $\spc{H}_{\rm R}=\spc{H}_{{\rm R},1}\otimes\spc{H}_{{\rm R},2}\otimes\cdots\otimes\spc{H}_{{\rm R},s_{\rm R}}$ and construct identical states on each of them, $\Psi=\psi^{\otimes s_{\rm R}}$. Each block consists in $2m$ qudits each, such that $2m\, s_{\rm R} = n_{\rm R}$. We assume that the error occurring on each of the reference frames is detectable, i.e., it acts as follows for all $k$: it takes the state on $\spc{H}_{{\rm R},k}$ to an error space $\spc{H}_{{\rm err},k}$ that is orthogonal to $\spc{H}_{{\rm R},k}$. This is the case, for instance, for erasure errors. %See section \ref{sec:discussion 2} for a discussion of errors which are not of this form.
\Mspace
 
The measurement performed on the reference frame during the decoding stage, is a two step process: We first perform binary projective measurements onto $\spc{H}_{{\rm R},k}$ and $\spc{H}_{{\rm err},k}$ for all blocks $k=1,\ldots,s_{\rm R}$, followed by discarding those blocks whose state is in the error space. We then measure the remaining blocks to obtain an estimate $\hat U$, which is used to decode the information in the computational space [see Eq.~(\ref{D-cov})].
\Mspace

\noindent {\bf Decoding Error and Noise Structure.} There are two contributing factors to errors in our covariant scheme: noise occurring between encoding and decoding, and errors in the decoding due to the finiteness of the reference frames. Here we show that our protocol is able to recover from both of these. Since these errors hinder our ability to decode, we refer to them as \emph{decoding errors} and quantify it via the diamond norm:
\begin{align}\label{error-diamond 2}
\epsilon_\cov\!:=\max_{V\in\grp{SU}(d)}\frac12 \left\|\map{D}_\cov\circ\map{C}\circ(\map{V}_\Cspace\otimes\map{V}_{\rm R})\circ\map{E}_\cov\,-\,\map{V}_{\rm L}\right\|_\diamond\!,
\end{align}
where $\map{C}$ is the noise channel being considered, and the subscript in $\epsilon_\cov$ is to remind the reader that the error is originating from making an encoder-decoder pair covariant. We assume it to be a statistical mixture over channels on the computation and reference frame registers, i.e.\null{} 
\begin{align}\label{eq:error channel structure}
\map{C}= \sum_j p_j\, \map{C}_{j,\Cspace}\otimes\map{C}_{j,{\rm R}},
\end{align}
where $\map{C}_{j,\Cspace/{\rm R}}$ are channels characterising the error on the computational/reference frame register, and $\{p_j\}$ is a probability distribution.
%and for the subroutine QECC to correct errors $\map{C}_{j,\Cspace}$ up to a  worst-case error of $\epsilon_{j,{\rm code}}=\frac12 \left\|\map{D}\circ\map{C}_{j,\Cspace}\circ\map{E}\,-\,\map{I}_{\rm L}\right\|_\diamond$ \MW{[while it is the same thing, should we not write here the maximisation over the unitaries as in Eq. \eqref{error-diamond 2} for completeness?]}. 
We also need to assume that the errors are covariant, namely $\map{C}\circ (\map{V}_\Cspace\otimes\map{V}_{\rm R})=(\map{V}_\Cspace\otimes\map{V}_{\rm R})\circ\map{C}$ for all $V\in\grp{SU}(d)$. For example, erasure errors are of this form.\Mspace

In \app{}~\ref{sec-general-qec}, we show that the total decoding error satisfies a bound of the form $\epsilon_{\cov}\leq \sum_j p_j C_j$, where $C_j$ only depends on either the computational error $\map{C}_{j,\Cspace}$ or the reference frame error $\map{C}_{j,{\rm R}}$; whichever error is dominant.\Mspace

\noindent {\bf Error bounds for the Weak Error Model.} Our first result concerns the case in which relatively few errors occur. Specifically, the error is that at most $n_{\rm e}$ qudits among the $n$ qudit systems composing the computational register and the reference frame register, are randomly lost:
\begin{align}\label{weak-error-model 2}
\map{C}=\sum_{\settt{s}} p_{\settt{s}}\left(\map{C}_{\rm e}\right)_{\settt{s}},
\end{align}
where the summation runs over all subsets $\settt{s}\subset\{1,2,\ldots,n\}$ of cardinality at most $n_{\rm e}$, and $\left(\map{C}_{\rm e}\right)_{\settt{s}}$ denotes the erasure of qudits whose labels are in the set $\settt{s}$, and $\{p_{\settt{s}}\}$ is a probability distribution.\Mspace

Any quantum error-correcting code over qudits of distance at least $k+1$ can perfectly correct $k$ erasures~\cite{Gottesman2006}. For instance, the polynomial codes of Aharonov and Ben-Or are $[[2k+1,1,k+1]]_d$ stabilizer codes with this property~\cite{aharonov_fault-tolerant_1997}. In these cases, we employ one of the perfect codes as the non-covariant subroutine $(\map{E},\map{D})$ of our code. Since we only need to encode one logical qudit, the perfect code requires only $n_\Cspace=O(1)$ computational qudits.

This is to say: the subroutine code $(\map{E},\map{D})$ we use requires only $O(1)$ computational qudits and corrects perfectly up to $n_{\rm e}$ erasure errors on the computational register.  Next we arrange the reference frame to perform well against this type of noise. To this purpose, we set $s_{\rm R}= n_{\rm e}+1$, so that the reference frames has $n_{\rm e}+1$ blocks. Crucially, these conditions imply that at least one of the reference frame states will survive the erasure, and we can measure it to obtain the embedded rotation. 
% each consisting in $2m$ qudits.

On each block we construct a reference frame state. Since the noise is weak, we can utilise as much entanglement as possible to enhance the accuracy of the reference frame decoding. To this purpose, we design a new family of highly coupled quantum reference frame states, which turn out to be a highly nontrivial generalisation of the sine-shape states \cite{buvzek1999optimal}, and thus we name them \emph{generalised sine states}. Using the generalised sine states as reference frames for our error correction protocol, we can recover the information with high accuracy. We derive an upper bound on the recovery error $\epsilon_{\cov}$ in our protocol for the weak error model, defined by Eq.\ (\ref{weak-error-model 2}):
\begin{align}
\epsilon_{\cov}\le \frac{81\pi^2d^4(d-1)^2(n_{\rm e}+1)^2(n_\Cspace+d-1)^2}{2n_{\rm R}^2}+O(n_{\rm R}^{-3}). \label{thm-weak-err 2}
\end{align}
All details can be found in \app{}~\ref{sec-weak-error}.
Since $n_\Cspace$ can be chosen to be $O(1)$ and $n_{\rm R}=n-n_\Cspace$ is close to $n$, our protocol achieves the Heisenberg limit $1/n^2$ with respect to the total number of qudit systems. \Mspace

It is most challenging to correct the errors in the weak error model when the locations of the erasures are most uncertain, that is, when $p_0=0$, $|\settt{s}|= n_{\rm e}$ and the $p_{\settt{s}}$ are constant, in Eq. \eqref{weak-error-model 2}. In this worse case scenario, we can lower bound the decoding error:

\begin{align}\label{eq:lower bound wea error model}
\epsilon_{\cov}\ge \frac{1}{16n^2(1+1/n_{\rm e})}.
\end{align} 
What is more, we show that \emph{any} encoder-decoder pair which satisfies the transversality condition Eq. \eqref{covariance}, encoding one qudit into $n$ physical qudits which is subject to the weak error model (\ref{weak-error-model 2}) with maximum uncertainty as described above, has a decoding error which is lower bounded by \eqref{eq:lower bound wea error model}. As shown in  \app{}~\ref{sec-weak-error}, the derivation of Eq. \eqref{eq:lower bound wea error model} comes from a strengthening of a bound found in~\cite{kubica2020using} (see also \cite{zhou2021new} for a more recent result with similar bounds).
Since the bound \eqref{eq:lower bound wea error model} matches the performance of our protocol in scaling, we conclude that the optimal error scaling of covariant codes is $1/n^2$.\Mspace %hhh\MW{[Perhaps comment on how this relates to Faist et al's work]}\MW{Move to discussions}

\noindent {\bf Error bounds for the i.i.d.\null{} Error Model.} We now turn our attention to types of erasure errors which are stronger than the one considered in the previous section and more realistic.
The error affects each qudit independently, erasing the $j^\text{th}$ qudit with a probability $p_{{\rm e}}$:
\begin{align}\label{strong-error-model 2}
\map{C}=\bigotimes_{j=1}^{n}\left((1-p_{{\rm e}})\map{I}+p_{{\rm e}}\map{C}_{\rm e}\right)\qquad    p_{{\rm e}}\in\left(0,\frac12\right),
\end{align}
where $\map{C}_{\rm e}$ denotes the single-qudit erasure channel. This error model can be written in the form Eq.~\eqref{eq:error channel structure}, and is thus compatible with our prior assumptions. Since the erasure channel is degradable, if $p_{\rm e}\ge1/2$ the information leaked to the environment would not be retrievable. Note that, in general, the error model on each qudit does not have to be identical, and each qudit $j$ can have distinct probability $p_{{\rm e},j}$ of being erased. In that case, however, we can simply set $p_{\rm e}$ to be the worst case over $\{p_{{\rm e},j}\}$ and consider this more stringent model instead.\Mspace

In contrast to the weak error model case, there does not exist any code, which we can use for our subroutine code, that perfectly corrects the errors at hand. Instead, there exist pretty good codes that correct the error unless too many qudits are erased.
The quantum capacity of the erasure channel, with erasure probability $p_{\rm e}$, has been determined to be $1-2p_{\rm e}$ \cite{bennett1997capacities}. We can choose $n_\Cspace$, the number of computational qudits, to grow with $n$.
When $n$ is large, since the number of qudits we want to encode is only one qudit and is much smaller than that allowed by the capacity \big[which is $(1-2p_{\rm e})n_\Cspace$\big], the error probability would vanish exponentially in $n_\Cspace$.\Mspace

We choose the number of computational qudits to scale sub-linearly in the total number of qudits, that is, $n_\Cspace= n^\gamma$, where $\gamma\in(0,1)$ does not depend on $n$ and can be chosen to be very small. We use any subroutine code $(\map{E},\map{D})$ that encodes one qudit into $n^\gamma$ computational qudits, with the property that it has a decoding error 
$O\left(e^{-x_{d}\cdot n^\gamma}\right)$
for some $x_{d}>0$ that may depend on $d$. By a random coding argument, one can show that there exists a stabilizer code satisfying our requirement (see, e.g., \cite{gottesman1997stabilizer}), although its explicit form is not given. Recently, progress in error correcting codes also showed that quantum polar codes~\cite{renes_efficient_2012,renes_polar_2014} and Reed-Muller codes~\cite{kumar_reed-muller_2016} have the desired property. Notice that the requirement of the $O\left(e^{-x_{d}\cdot n^\gamma}\right)$ scaling is chosen for convenience of analysing the error, and it can be further relaxed in practice.\Mspace

Meanwhile, the model is now too noisy for the highly coupled reference frame state used in the weak error model to be effective. Instead, we prepare the reference frame register to have only two qudits per block, namely $2\, s_{\rm R}=n_{\rm R}$, with the reference frame on each two qudits block being the maximally entangled state. While like in the weak error model, this choice means that one erasure error only destroys one reference frame block, now there are far more blocks, which reduces the chance of them all being erased. The price to pay for this more error resilient setup, is less precise reference frames, due the smaller dimension of each block.\Mspace

As long as there are still order $n$ non erased blocks left we can achieve high performance, which happens with very high probability. Indeed, the overall error, Eq. \eqref{strong-error-model 2}, can be recast into the form:
\begin{align}\label{strong-error-model 3}
\map{C}=\sum_{k=0}^{n}p_{\rm e}^{k}(1-p_{\rm e})^{n-k}\sum_{\settt{s}:|\settt{s}|=k}\left(\map{C}_{\rm e}\right)_{\settt{s}}.
\end{align}
 Then the  number of surviving blocks follows a binomial distribution with mean $n(1-p_{\rm e})$, which is sufficient for a high accuracy decoding. This reasoning leads us to the following upper bound on the decoding error under the i.i.d.\null{} error model (\ref{strong-error-model 3}): 
For any $\alpha>0$, there exists $n_{\alpha}>0$ such that
\begin{align}\label{coverr-iid}
\epsilon_{\cov}\le \left(\frac{36(d^2-d+32)d^{\frac{d^2-d+2}{2}}}{(1-2p_{\rm e})^2\prod_{j=1}^{d}(j-1)!}\right)\left(\frac{1}{(1-2p_{\rm e})n}\right)^{1-\alpha}
\end{align}
for all $n\ge n_\alpha$. Details of the proof can be found in \app{}~\ref{sec-strong-error}.
\Mspace

The error of our protocol scales  almost as $1/n$, instead of $1/n^2$ in the previous case. This is a result of the (stronger) i.i.d.\null{} noise. In fact,  the error scaling of our protocol is still optimal for this error model.
Analogously to the weak error model, we can prove that the $1/n$ scaling is also optimal for the i.i.d.\null{} error model. In particular, we prove that all covariant QECC satisfy the following bound under the i.i.d.\null{} error model:
\begin{align}\label{eq:lower bound trong error model}
\epsilon_{\cov}\ge \frac{p_{\rm e}}{64n(1-p_{\rm e})}.
\end{align} 
In the above analysis, we focused on erasure error models, as they allow us to derive analytical bounds on our protocol's performance. Nevertheless, our protocol also performs well under other commonly encountered error types, as we show in the following.

\medskip

\noindent {\bf Numerical simulation for generic error models.} In addition to the above bounds, we also conduct numerical experiments to evaluate the performance of our protocol. As a working example, we use the ``5-qubit code" \cite{laflamme1996perfect} (5 computational qubits; one logical qubit) as the subroutine encoding and decoding pair. We consider not only erasure errors but also other common error types like dephasing errors and depolarising errors.

We first consider the weak erasure error model and use a group of identical generalised sine states as the reference frame register. 
We assume that at least one reference frame state survives the erasure error, which consists of $2m$ qubits. (Recall that the number of qubits is always even since we use entangled reference frame states.) 
%%\YY{consistency! change $N, N_{total}\to n,n_{\rm R}...$ }\MW{[As far as I recall, we defined $n_{\rm R}$ to be the number of reference frame qubits before error occurred, then than after.]} \YY{Update the text and figure with $n_{\rm R}/(n_{\rm e}+1)$ as the variable.}
As shown in Fig.~\ref{fig_nume_erasure}, the decoding error (in the leading order) is $\frac{269.2}{(2m)^2}$, matching the prediction of our theory.

\iffalse
\begin{table}[h]
	\begin{center}
		\begin{tabular}{|c|c|c|c|c|}
			\hline
			\,\, N \,\, & \,\, 10 \,\, & \,\, 20 \,\, & \,\, 40 \,\, & \,\, 60 \,\, \\
			\hline
			Erasure noise & 1.97e-1 & 3.71e-2 & 8.39e-3 & 3.65e-3 \\
			\hline
			\,\, N \,\, & \,\, 80 \,\, & \,\, 100 \,\, & \,\, 120 \,\, & \,\, 150 \,\, \\
			\hline
			Erasure noise & 2.27e-3 & 1.26e-3 & 1.01e-3 & 6.20e-4 \\
			\hline
		\end{tabular}
	\end{center}
	\caption{\textbf{Erasure noise with ``5-qubit code"}. The error of the our covariant code using ``5-qubit code" as the subroutine encoding is given for some $N$ from $10$ to $150$. After fitting, the entanglement error is about 
	$\frac{19.42}{n_{\rm R}^2}$.}
	\label{table_nume_erasure}
\end{table}
\fi

\begin{figure}[]
	\includegraphics[width=0.9\linewidth]{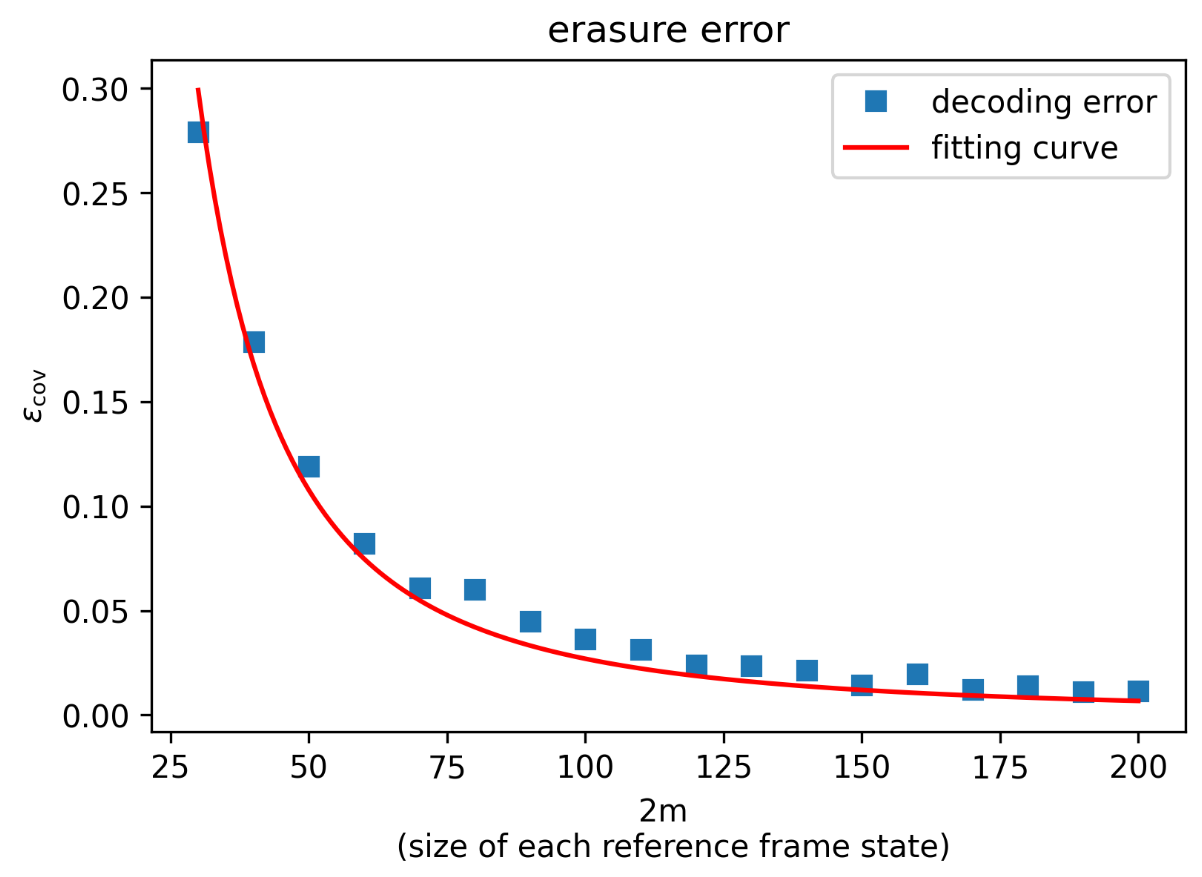}%[width=1.3\linewidth]{Computer_References_imageV1.pdf}
	\caption{\textbf{Weak erasure error with ``5-qubit code"}. Under the weak erasure model, the decoding error of our protocol (with the 5-qubit code \cite{bennett1997capacities} as the subroutine non-covariant code and the generalized sine states as the reference frame) is plotted as a function of $2m$, the number of qubits contained in each block of the reference frame register. The blue dots correspond to the evaluated values of the decoding error and the red line corresponds to the fitting curve	$\frac{269.2}{(2m)^2}$. \label{fig_nume_erasure}}
\end{figure}
 
We now consider errors beyond the erasure models. 
We first consider the case where all reference frame qubits are subject to i.i.d.\null{} depolarising/dephasing error with probability $p=0.2$.
In this case, our reference frame register consists of $n_{\rm R}/2$ maximally entangled states. Each of the reference frame states are measured individually for an estimate of the reference frame parameters. We assume at most one out of the five computational qubits is subject to error (corresponding to an error rate of $0.2$), so that the logical qubit can always be recovered. 

Notice that, unlike the erasure error, depolarising and dephasing errors in the reference frame register are not flagged (i.e., we do not know whether a reference frame state is subject to an error or not). In order to mitigate the effect of erroneous reference frames, we adopt a majority vote algorithm that abandons estimates in case they are likely to be faulty. With this method, for the i.i.d.\null{} depolarizing/dephasing model, we are able to achieve decoding errors (in the leading order) $\frac{33.30}{n}$ for the depolarising error and $\frac{32.43}{n}$ for the dephasing error, with $n(:=n_{\rm R}+5)$ being the total number of physical qubits employed in our construction (see Fig.~\ref{iid}). % and Fig.~\ref{fig_nume_dephasing}.)  %\YY{Add the plot for dephasing as a subfigure.} 
The $1/n$ error scaling, observed for both the i.i.d.\null{} depolarising error and the i.i.d.\null{} dephasing error, coincides with the lower bound (\ref{coverr-iid}) which we proved for the i.i.d.\null{} erasure error model. It suggests that the error scaling of our protocol is consistent in generic error models. 
We have also tested our protocol under variants of the above error models: We consider models where one fifth reference frame qubits go through the completely depolarising/dephasing error, and the $1/n$ error scaling has also been observed (see Table. \ref{table_nume_de} for a summary).

These results show that the application of our protocol is not limited to erasure errors. Instead, it has good potential to work well under generic error types.
Our code is available at \cite{moyin2021github}, and more details can be found in Appendix \ref{sec-numerical}.  

\begin{table}[h]
	\begin{center}
		\begin{tabular}{|c|c|c|c|c|}
			\hline
			\,\, $n$ (total qubits)\,\, & \,\, 75 \,\, & \,\, 135 \,\, & \,\, 205 \,\, & \,\, 305 \,\, \\
			\hline
			i.i.d.\null{} dephasing ($p=0.20$) & 4.28e-1 & 2.58e-1 & 1.71e-1 & 1.18e-1 \\
			\hline
			i.i.d.\null{} depolarising ($p=0.20$) & 4.29e-1 & 2.62e-1 & 1.77e-1 & 1.27e-1 \\
			\hline
			20\% qubits dephased & 3.88e-1 & 2.38e-1 & 1.74e-1 & 1.20e-1 \\
			\hline
			20\% qubits depolarised & 4.60e-1 & 3.15e-1 & 2.18e-1 & 1.53e-1 \\
			\hline
			\,\, $n$ (total qubits)\,\, & \,\, 405 \,\, & \,\, 505 \,\, & \,\, 605 \,\, & \,\, 705 \,\, \\
			\hline
			i.i.d.\null{} dephasing ($p=0.20$) & 8.79e-2 & 7.19e-2 & 6.02e-2 & 5.09e-2 \\
			\hline
			i.i.d.\null{} depolarising ($p=0.20$) & 8.76e-2 & 7.32e-2 & 5.83e-2 & 5.30e-2 \\
			\hline
			20\% qubits  dephased & 8.65e-2 & 7.03e-2 & 5.94e-2 & 5.34e-2 \\
			\hline
			20\% qubits depolarised & 1.21e-1 & 9.17e-2 & 7.52e-2 & 6.37e-2 \\
			\hline
		\end{tabular}
	\end{center}
	\caption{\textbf{Performance of our construction under dephasing and depolarising errors}. For variants of the i.i.d.\null{} dephasing/depolarising error model, the decoding errors of our construction (using ``5-qubit code" as the subroutine non-covariant code) as functions of the code size are numerically evaluated. In the variant models, a randomly chosen subset of the qubits are completely depolarised/dephased, while others remain unaffected by noise. Numerical simulation shows that in models considered here the decoding errors vanish with the $1/n$ scaling.}
	\label{table_nume_de}
\end{table}

\begin{figure} 
	\begin{subfigure}[b]{0.44\textwidth}
		\includegraphics[width=1\linewidth]{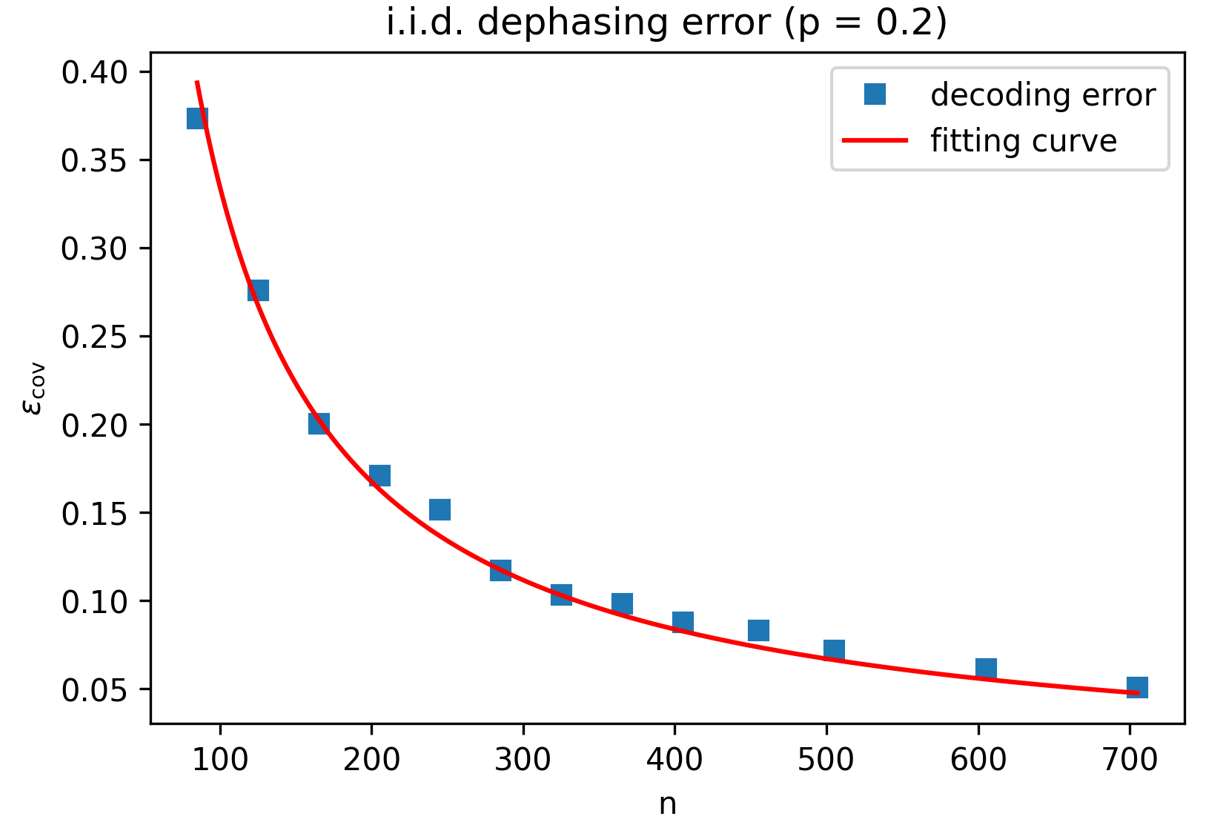}
		%\caption{}
		%\label{fig:Ng1} 
	\end{subfigure}
	
	\begin{subfigure}[b]{0.44\textwidth}
		\includegraphics[width=1\linewidth]{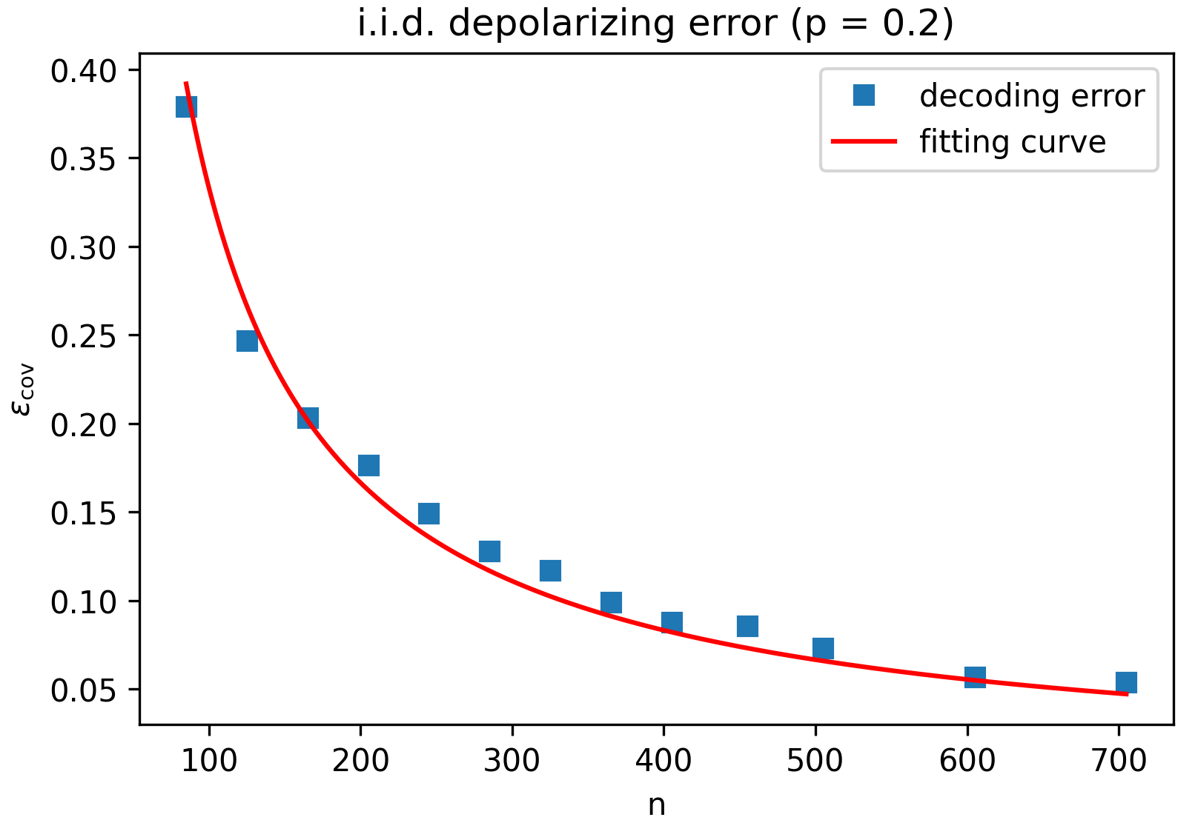}
		%\caption{}
		%\label{fig:Ng2}
	\end{subfigure} 
	
	\caption[l]{\textbf{i.i.d.\null{} dephasing/depolarizing with $p=0.2$ and ``5-qubit code"}. The decoding error for i.i.d.\null{} dephasing/depolarizing error of the our covariant code using ``5-qubit code" as the subroutine encoding and $n_{\rm R} /2$ maximally entangled states as the reference frame is plotted as a function of $n$, the number of total physical qubits. The blue dots correspond to the evaluated values of the decoding error and the red lines correspond to the fitting curves: $\frac{32.43}{n}$ for the dephasing error and $\frac{33.30}{n}$ for the depolarizing error.}\label{iid}
\end{figure}

\iffalse
\begin{figure}[]
	\includegraphics[width=0.7\linewidth]{iid_depolarizing.png}%[width=1.3\linewidth]{Computer_References_imageV1.pdf}
	\caption{\textbf{i.i.d depolarizing with $p=0.2$ and ``5-qubit code"}. The error of the our covariant code using ``5-qubit code" as the subroutine encoding is given for $n$ between $80$ and $700$. ($n$ total physical qubits.) The blue dots correspond to the evaluated values of the decoding error and the red line corresponds to the fitting curve $\frac{33.30}{n}$. \label{fig_nume_de}}
\end{figure}

\begin{figure}[]
	\includegraphics[width=0.7\linewidth]{iid_dephasing.png}%[width=1.3\linewidth]{Computer_References_imageV1.pdf}
	\caption{\textbf{i.i.d dephasing with $p=0.2$ and ``5-qubit code"}. The error of the our covariant code using ``5-qubit code" as the subroutine encoding is given for $n$ between $80$ and $700$. ($n$ total physical qubits.) The blue dots correspond to the evaluated values of the decoding error and the red line corresponds to the fitting curve $\frac{32.43}{n}$. \label{fig_nume_dephasing}}
\end{figure}
\fi

\noindent {\bf Resource requirements and implementation.} We now analyse some of the memory and time resources needed to implement our QECC which are required beyond those of the subroutine code.\Mspace

Our reference frames are constructed on $n_{\rm R}$ qudits constituting a $d^{n_{\rm R}}$ dimensional Hilbert space. However, the reference frame states only have support on a space $d_{\rm R}:=\dim\left(\mathsf{Span}\{\map{U}_{\rm R}(\Psi)\}_{U\in\grp{SU}(d)}\right),$ which we prove to be upper bounded by
\begin{align}
d_{\rm R} &\le \left(\frac{n_{\rm R}}{2}+1\right)^{(d^2-1)}.
\end{align}
Therefore, we can compress the reference frame state to a system of much smaller dimension, which reduces exponentially the cost of quantum memory. The cost can be further reduced at the price of a small recovery error using the compression protocols in Refs.\ \cite{yang2016efficient,yang2016optimal,yang2018compression}.\Mspace

We now turn our attention to the implementation of the covariant encoding and decoding: 
At first sight, our encoding and decoding may appear to necessitate sampling from the Haar measure on $\grp{SU}(d)$ in order to be implemented [see Eqs.~\eqref{E-cov} and \eqref{D-cov}]. From a computational complexity standpoint, this would be a problem, since to implement a Haar-random unitary one needs an exponential number of two-qubit gates and random bits \cite{Knill1995Haar}. Fortunately, our encoding and decoding do not require sampling from the Haar measure, but rather only from a distribution which agrees with it upto the $t^\textup{th}$ moment, for appropriate $t$. Such equivalent distributions are known as unitary $t$-designs~\cite{PhysRevLett.98.130502,Dahlsten_2007,PhysRevA.80.012304,CMPHarrow,Diniz2011,PhysRevA.78.062329,CMPBrandao,PhysRevX.7.021006}, and are more efficient to implement. \Mspace

Specifically, let $U$ be a unitary representation on $\C^{d\times d}$ and $P_{t,t}(U)$ be a matrix whose entries are polynomials of order $t$ in the coefficients of $U$ and of order $t$ in the coefficients of $U^*$, and let $\mathbb{E}_{U\sim \nu}[f(U)]$ be the expectation value of a function $f$ according to measure $\nu$. We then say that $ \mathbb{E}_{U\sim \nu_\textup{Haar}}[P_{t,t}(U)]$, where $\nu_\textup{Haar}$ denotes the Haar measure, \emph{admits a unitary $t$-design}.\Mspace

Our result is that the encoder $\map{E}_\cov(X_{\rm L})$ and decoder $\map{D}_\cov(X_\Cspace)$ admit unitary $(n_\Cspace+n_{\rm R}/2+1)$-designs for all $X_{\rm L}\in \Lin(\mathcal{H}_{\rm L})$, $X_\Cspace\in \Lin(\mathcal{H}_\Cspace^{\otimes n})$.
%\end{prop}
See \app{}~\ref{sec-dim-prob} for proof.
Ref.\ \cite{CMPBrandao} devises and quantifies a method to approximate unitary $t$-designs: To approximate our encoder and decoder up to an error $\epsilon_\cov$ in diamond norm, it is enough to use $k$ two-qubit gates, randomly drawn from the Haar measure on $\grp{SU}(4)$. Here
\begin{align}
k=& 510,000 \, \log_2(d) (n_\Cspace+n_{\rm R}/2+1) ^{9.1} \\
&\times \Big(  2 \log_2(d) (n_\Cspace+n_{\rm R}/2+1)+1 +\log_2(1/\epsilon_\cov) \Big)\notag,
\end{align}
which scales polynomially in both $n_\Cspace$ and $n_{\rm R}$ [recall lower bounds on $\epsilon_{\cov}$, in the weak and i.i.d.\null{} error models, Eqs. \eqref{eq:lower bound wea error model}, \eqref{eq:lower bound trong error model}].\Mspace
 
Since sampling a polynomial number of times from $\grp{SU}(4)$ can be performed efficiently,  both the encoder and the decoder  can be efficiently implemented in $n_\Cspace$ and $n_{\rm R}$ so long as the reference frame state $\Psi$ and measurement $\map{M}_{\hat{U}}$ can be efficiently constructed. \Mspace
%Since one needs to apply transversal gates on the encoded logical state to implement any computation, it is also important that under the application of an arbitrary number of transversal gates, the errors in Eqs. \eqref{approx code 2}, \eqref{approx decode 2} do not grow \cite{PhysRevX.7.021006}.

\noindent{\bf Efficient error syndrome extraction.} In quantum error correction, we often need to measure error syndromes after applying quantum gates. It might appear that a full decoding is needed to achieve this in our QECC, which, despite the efficiency analysis in the previous part, could have non-negligible effects on the performance of the computation. However, this is not true for the erasure error models we considered. Instead, one can keep track of the error in a very simple way and correct them all together at a suitable time.

To see this, first notice that erasure errors are ``flagged'', meaning that the qudit is mapped from its original Hilbert space $\spc{H}$ to an orthogonal Hilbert space $\spc{H}_{\rm err}$. Then, the error syndrome can be obtained by a binary projective measurement consisting of projectors on $\spc{H}$ and $\spc{H}_{\rm err}$. This measurement commutes with the (covariant) implementation of logical operations, and thus to implement a sequence of logical gates we can adapt the following procedure:
\begin{enumerate}
  \item Apply $\map{E}_{\rm cov}$.
  \item Create a (classical) bit string $\settt{s}$, initiated as a null vector, to record erasure errors.
  \item For $i=1,\dots,k$:
  \begin{enumerate}
    \item Apply $\map{V}_i$,
    \item measure error syndromes, and
    \item add the measured error locations to $\settt{s}$.
    \item If $|\settt{s}|$ exceeds the maximal number of tolerable local errors, do $\map{D}_{\rm cov}$ followed by $\map{E}_{\rm cov}$.
  \end{enumerate}
  \item Apply $\map{D}_{\rm cov}$.
\end{enumerate}
Therefore, a full physical decoding is done only if there is too many local errors (which is determined by the subroutine code) rather than after each logical gate, making our QECC much more efficient than it appears. We remark that how this feature generalises to other error models remains an open question to be investigated in the future.
\Mspace

\noindent{\bf Towards fault-tolerant QECCs.}
To achieve practical fault tolerance, one must show that each step of the QEC procedure can be performed fault-tolerantly. For our QECC construction, the logical gate implementation is transversal and thus fault-tolerant thanks to covariance. It remains to be shown that the encoder and the decoder are both: a) computationally efficient to implement, and b) can be implemented fault tolerantly. 
For a), we have provided covariant encoders and decoders which can be implemented via a polynomially in $n$ (total number of physical qudits) number of two-qubit gates, while still maintaining the optimal scaling of the decoding error. 
Regarding b), making the encoder and the decoder fault-tolerant amounts to finding a fault-tolerant implementation of the $\grp{SU}(d)$-twirling in Eqs. (\ref{E-cov}) and (\ref{D-cov}). We remark that it is unlikely that unitary designs are the solution: Any unitary $t$-design $\settt{D}$ is an $\epsilon$-net on the parameter space of $\grp{SU}(d)$, with $\epsilon$ vanishing in $t$. For large $t$, the $\epsilon$-net becomes dense enough that, unless the commutant of $\{U\otimes U^\ast~|~U\in\settt{D}\}$ has dimension larger than two, $\settt{D}$ forms a universal gate set (see also \cite[Remark 4]{oszmaniec2020epsilon}). In this case, it follows that to implement the unitary design fault tolerantly we need to implement a universal gate set fault tolerantly, leading to a contradiction.

Nevertheless, unitary designs are not a must for implementing the twirling in Eqs. (\ref{E-cov}) and (\ref{D-cov}). One can, for example, use extrapolation of different noisy channels as an alternative approach. Similar ideas have resulted in error mitigation \cite{PhysRevLett.119.180509,PhysRevX.8.031027}, which effectively deals with errors in near-term devices.

Another possibility, is to partition the logical space into code blocks and apply our encoding to each block individually (call this level 1 encoding). One then implements the gates corresponding to the computation such that any one gate has support on either one or two code blocks. When one wishes to implement a gate across two code blocks, these code blocks are encoded once more using our encoding scheme (call this level 2 encoding); the unitary is then implemented and we decode back to the level 1 encoding. This way, one can increase the logical code size by adding more code blocks (each one of constant size) while only applying gates on encoded states. The quantum reference frames for the level 2 encoding can be re-purposed after ever level 2 decoding for level 2 encodings over other code blocks. Using this scheme, there is a much better overall scaling of the resources involved since we no longer require $\grp{SU}(d)$ twirling over the entire logical space. These possibilities will be explored more in future research. In summary, our covariant QECC construction does not immediately lead to a satisfactory approach of achieving full fault tolerance, but it opens up a new route that is worth further exploration.

\section{Discussion}
\label{sec:discussion 2} 

Starting from an arbitrary error correcting code, we have shown how to use quantum reference frames to construct a new code for which all gates are covariant and transversal, up to a small error in the decoding. In the absence of noise, the error in the decoding is ultimately a consequence of the inability to perfectly determine the orientation of the reference frames via measurement.\Mspace
%Compared to existing proposals of approximate $\grp{SU}(d)$-covariant codes \cite{faist2019continuous}, our protocol achieves the optimal size-accuracy trade-off in two distinct noise models, and it offers a versatile approach that can turn arbitrary existing codes into approximate covariant codes.\Mspace

Our results generalise the setting in \cite{woods2020continuousgroups} from $d$-dimensional reference frame representations of $\grp{U}(1)$ to $d$-dimensional reference frame representations of $\grp{SU}(d)$. This construction is important primarily since once the logical states have been encoded, it allows for a universal set of gates to be applied transversally, which in turn can allow for errors to be more easily detected and corrected during a computation. Furthermore, the construction also allows  any quantum error correcting code to be converted into a covariant code. Previously it was not known if this is possible with a finite dimensional physical space.
In the setting of the weak error model and $\grp{SU}(d)$ symmetry, we establish novel upper bounds on the decoding error which are tight in the scaling with $n$ (number of physical subsystems). Indeed, previous upper bounds in this context, \cite{faist2019continuous,wang2020quasi}, are not tight with $n$.  Furthermore, in this context, we also improved previous results \cite{kubica2020using,zhou2021new} on the lower bounds of the decoding error.
We have also analysed the i.i.d. erasure model, which is more realistic in the context of computation. Here we have produced  upper and lower bounds for the decoding error which are tight in the scaling with $n$.
We argued that our model for a QECC should be able to correct errors beyond those of the local erasure type. Numerical simulations of the performance of our construction for practical error models including the depolarising error and the dephasing error bears this out.
%Our work is also the first to show that the Heisenberg scaling has to be replaced by liner scaling when the i.i.d. error model is considered.

Our results also feature applications beyond correcting errors in quantum computation. 
Two party communication with misaligned reference frames and lossy quantum channels \cite{hayden2017error} is a clear example, but the range of potential applications is much broader.
In particular, we designed and introduced generalised sine states and used them as optimal reference frames in the weak error model. A characteristic novel feature of the generalised sine states is that they can achieve the Heisenberg limit in estimating a completely unknown unitary gate, making it promising to apply them in the emergent field of multiparameter quantum metrology \cite{albarelli2020perspective,gorecki2020optimal,hou2020minimal,PhysRevLett.127.110501,meyer2021variational} and, more generally, in quantum information processing.
Remarkably, these states have recently been used as the key to solving the long-standing optimal programming problem of unitary quantum gates \cite{Yuxiang2}.

From a more fundamental physics perspective, our results provide new insight into the study of quantum gravity.
In particular, there is a conjectured duality between a theory of gravity in Anti-de Sitter (AdS) space and a conformal quantum field theory (CFT) on its boundary, known as the AdS-CFT correspondence. The duality is mediated via a quantum error correcting code, where the encoding channel maps low-energy states in the Anti-de Sitter space (known as the ``bulk'') to states on the  conformal field (known as the ``boundary")~\cite{almheiri2015bulk}. The implication of the Eastin-Knill theorem, discussed in the introduction, to the AdS-CFT duality is that global continuous symmetries in the bulk cannot map onto local symmetries on the boundary (see~\cite{Harlow2018a_short,Harlow2018} for details). This is often interpreted as a problem for the duality since initially such symmetries were believed to exist. It is thus interesting to quantify, as a function of physical quantities such as energy or dimension, how local symmetries of the bulk can approximate local symmetries in the boundary. This program was initiated in \cite{TamaToby,woods2020continuousgroups,faist2019continuous} and is the subject of ongoing investigation (see e.g.~\cite{milekhin2021quantum}). Our results quantify the extent to which the approximate duality for global $\grp{SU}(d)$ symmetries could exist.\Mspace

Our numerical results suggest that our reference frame based QECCs work well for generic non-detectable error models. Understanding and analytically bounding the decoding error, on the other hand, requires further investigation. It would be an interesting future direction of research to rigorously analyse the performance of our QECCs in generic error models and compare to other existing constructions~\cite{wang2020quasi}. The dephasing error model, for example, was studied in the $U(1)$ case in~\cite{woods2020continuousgroups} and shown to be correctable.  This should also be true in the current $\grp{SU}(d)$ case. These studies will also help pave the way to proving robustness of our results when optimised for error models found in experiments. \Mspace
%Many of the approaches to quantum computation in the current noisy intermediate stage we find ourselves in, involve quantum error mitigation rather than full quantum error correction, see e.g. \cite{PhysRevLett.119.180509,PhysRevX.8.031027}\MW{more citations?}. {\color{blue} I pushed the citations forward to the previous paragraph, so we might remove them here.} \MW{I saw. Looks good. Do you want to modify/delete this sentence here then?} Since our encoding is \emph{inherently} random, even when the errors are too abundant to fully correct, the randomness in our code may provide a naturally occurring source of error mitigation.

Finally, since our protocols only involve classical correlations between the computational system and the reference frame system, this setup is naturally suited to implementing these two physical subsystems on different quantum architectures to harness specific advantages of each platform. For example, one may wish to implement the quantum reference frames on a phonic platform using Gaussian quantum computing~\cite{RevModPhys.84.621}, while using superconducting qubits for the computational space.

\begin{acknowledgements}
	This work is supported by the Swiss National Science Foundation via the National Centre for Competence in Research ``QSIT" as well as via project No.\ 200020\_165843, by the National Science Foundation of China through Grant No. 11675136, by Hong Kong Research Grant Council through Grants No. 17300317 and 17300918.  M.W. acknowledges funding from the Swiss National Science Foundation (AMBIZIONE Fellowship, No.~PZ00P2\_179914). 
	M.Y. and G.C. acknowledge support by the John Templeton Foundation through grant 61466, The Quantum Information Structure of Spacetime (qiss.fr). The opinions expressed in this publication are those of the authors and do not necessarily reflect the views of the John Templeton Foundation. Research at the Perimeter Institute is supported by the Government of Canada through the Department of Innovation, Science and Economic Development Canada and by the Province of Ontario through the Ministry of Research, Innovation and Science.
	
\end{acknowledgements}

\appendix 
\iffalse
In this Supplemental Information (SI), we provide all technical details that support our arguments in the main text. The SI is structured as follows: In Section \ref{sec-background}, we provide the background and notations required in later sections.
In Section \ref{sec-general-qec}, we derive both a general upper bound and a general lower bound on the decoding error of our protocols.
In Section \ref{sec-weak-error}, we discuss the performance of our construction in the weak error model, where our protocol, with $n$ qudits subject to weak erasure error, has a decoding error that scales as $1/n^2$.
In Section \ref{sec-strong-error}, we discuss the performance of our construction in the i.i.d.\null{} error model, where our protocol has a decoding error that scales as $1/n$.
We prove that both the $1/n^2$ error scaling in the weak error model (Proposition \ref{prop-converse}) and the $1/n$ error scaling in the i.i.d.\null{} error model (Proposition \ref{prop-converse-strong-error}) are both optimal. 
In Section \ref{sec-dim-prob}, we provide more details on reducing the cost of storing the quantum reference frame states.
In Section \ref{sec:Computational efficiency of implementation}, we discuss the computation complexity of realising our encoder and decoder.
 \fi

%\section{Background}\label{sec-background}

\section{Construction of covariant error correcting codes}\label{sec-general-qec}
\subsection{Preliminary}\label{sec:preliminary}

We  use the notation $\spc H_x$ ($x={\rm R, P, L}$) to denote the Hilbert space of the reference frame register, the computational register, or the logical register. In particular, $\spc{H}$ refers to a $d$-dimensional Hilbert space.
For a generic Hilbert space $\spc K$ and a pure state $|\psi\>  \in  \spc K$, we will use the notation $\psi:  =  |\psi\>\<\psi|$ to denote the projector on the one-dimensional subspace spanned by $|\psi\>$.   The set of quantum states will be denoted by  $\St(\spc K)$. In this work we will focus on finite dimensional quantum systems,  with $\dim  (\spc K)  <\infty$.

 A quantum process transforming an input system into a (possibly different) output system is called a {\em quantum channel}.   A quantum channel transforming an  input system with Hilbert space $\spc K_{\rm in}$ into an output system with Hilbert space $\spc K_{\rm out}$ is a completely positive trace-preserving map $\map N:  \Lin( \spc H_{\rm in})  \to \Lin(\spc H_{\rm out})$, where $\Lin(\spc K)$ denotes the space of linear operators on $\spc K$. For pure channels with a single Kraus operator $U$, i.e.\null{} unitaries or isometries, we use the shorthand $\map{U}(\cdot):=U(\cdot)U^\dag$.
 %The set of all quantum channels with input space $\spc K_{\rm in}$ and output space $\spc K_{\rm out}$ will be denoted by  ${\sf Chan}  (\spc H_{\rm in},\spc H_{\rm out})$.  

We will use the big-$\Omega$ notation, the big-$O$ notation, and the big-$\Theta$ notation to characterise the asymptotic behaviour of functions. For a function $f(n)$, we write $f(n)=\Omega(g(n))$ if there exists a constant $c_1>0$ so that $f(n)\ge c_1\,g(n)$ for large enough $n$, $f(n)=O(g(n))$ if there exists a constant $c_2>0$ so that $|f(n)|\le c_2g(n)$ for large enough $n$, and $f(n)=\Theta(g(n))$ if $f(n)=\Omega(g(n))$ and $f(n)=O(g(n))$. 

We will make frequent use of a couple of basic concepts in representation theory of the special unitary group $\grp{SU}(d)$, which can be found in standard textbooks, e.g., Ref.\  \cite{fulton2013representation}.
We denote the Young diagrams by a vector\footnote{For vectors, we use the bond font.} $\Vec\lambda=(\lambda_1,\lambda_2,\dots)$ with $\lambda_1\ge\lambda_2\ge\cdots$ and by $U_{\Vec\lambda}$ the irreducible representation of $\grp{SU}(d)$ characterised by the Young diagram $\Vec\lambda$. We denote the collection of all Young diagrams with $n$ boxes and at most $d$ rows by $\settt{Y}_n$ (since the dimension $d$ is fixed throughout the paper, we often omit it). In particular, we define $\Vec e_i$ to be the vector whose $i$-th entry is one and other entries are zero. By definition, $\Vec e_1$ corresponds to a legitimate Young diagram whose associated representation is the standard one, and we use the abbreviation $U:=U_{\Vec e_1}$. 
We will use the double-ket notation $|A\kk:=\sum_{n,m}\<n|A|m\>|n\>|m\>$ ($\{|n\>\}$ being an orthonormal basis) with $A$ being  a matrix. With the double-ket notation,  we denote by $|\Phi^+_{U,\Vec\lambda}\>:= |U_{\Vec\lambda}\kk/\sqrt{d_{\Vec\lambda}}$ the maximally entangled state associated to $U_{\Vec\lambda}$, with $d_{\Vec\lambda}=\Tr I_{\Vec\lambda}$ being the dimension of the irreducible subspace. By this notation, $|\Phi^+_U\>$ refers to a maximally entangled state in $\spc{H}\otimes\spc{H}$, rotated by a local unitary $U$.
Unless otherwise specified, $\d U$ or sometimes $\d\hat{U}$ denotes the Haar measure.  

Other commonly used conventions and notations can be found in Table \ref{table-conv}.

\begin{table*}[tb]
\begin{center}
 \begin{tabular}{| p{2cm} | p{9cm} | p{3cm} |} 
 \hline
 Notation & Definition & Where defined\\
  \hline
 $[k]$ & the set $\{1,2,\dots,k\}$ & \\
 \hline
 $|\settt{S}|$ & the cardinality of set $\settt{S}$ & \\
 \hline
 $n$ & total number of qudit systems in the code  & Section \ref{subsec-covqec}\\ 
 \hline
 $n_{\rm R}$ & number of qudits used as the reference frame& Section \ref{subsec-covqec}\\
 \hline
 $n_\Cspace$ & number of qudits in the base (non-covariant) code& Section \ref{subsec-covqec} \\
 \hline
  $s_{\rm R}$ & number of reference frame state copies & Section \ref{subsec-RF} \\
 \hline
  $\epsilon_{\cov}$ & the error of the constructed covariant code, measured by the diamond norm error & Eq.\ (\ref{error-diamond})\\
 \hline
  $O(n)$ & the big-$O$ notation & Section \ref{sec:preliminary}\\
 \hline
  $\Theta(n)$ & the big-$\Theta$ notation & Section \ref{sec:preliminary}\\
 \hline
   $F_{\wc}(\map{A},\map{B})$ & the worst-case input fidelity between channels $\map{A}$ and $\map{B}$ & Eq.\ (\ref{F-wc})\\
 \hline
 $F_{\ent}(\map{A},\map{B})$ & the entanglement fidelity between channels $\map{A}$ and $\map{B}$& Eq.\ (\ref{F-ent}) \\
 \hline
 $\epsilon_{\wc}(\map{A},\map{B})$ & the diamond norm error between channels $\map{A}$ and $\map{B}$ & Eq.\ (\ref{epsilon-wc})\\
 \hline
 $\epsilon_{\ent}(\map{A},\map{B})$ & the entanglement error between channels $\map{A}$ and $\map{B}$  & Eq.\ (\ref{epsilon-ent})\\
 \hline

\end{tabular}
\end{center}

\caption{Table of notations and conventions}\label{table-conv}

\end{table*} 

Next we introduce a couple of measures of error and faithfulness that are used in this work. Consider two channels $\map{A}, \map{B}: \Lin(\spc{K}_{\rm in})\to\Lin(\spc{K}_{\rm out})$. For instance, $\map{A}$ can be a desired quantum gate and $\map{B}$ is an approximation of $\map{A}$. Therefore, we often need measures to quantify their distance (i.e.\null{} measures of simulation error) and similarity (i.e.\null{} measures of simulation faithfulness).

The first measure of faithfulness, defined intuitively as the minimum fidelity between the outputs for the same input, is the \emph{worst-case input (or minimum) fidelity}:
\begin{align}\label{F-wc}
F_{\rm wc}(\map{A},\map{B}):=
\inf_{\rho\in \St(\spc{K}_{\rm in}\otimes\spc{K}')}F\left((\map{A}\otimes\map{I}_{\spc{K}'})(\rho),(\map{B}\otimes\map{I}_{\spc{K}'})(\rho)\right),
\end{align}
where $F(\rho,\sigma):=\left(\Tr\sqrt{\rho^{\frac12}\sigma\rho^{\frac12}}\right)^2$ is the Uhlmann fidelity defined for two arbitrary quantum states $\rho$ and $\sigma$ and $\spc{K}'$ is an arbitrary reference Hilbert space. 
Later we give a semidefinite program to compute ${F_{\wc}(\map{A},\map{B})}$ (see Lemma \ref{lemma:sdp}).
Nonetheless, the worst-case fidelity is often difficult to work with, and thus it is more common to consider another measure of faithfulness, known as the \emph{entanglement fidelity} \cite{raginsky2001fidelity}:
\begin{align}\label{F-ent}
F_{\rm ent}(\map{A},\map{B}):=F\left((\map{A}\otimes\map{I}_{\spc{K}'})(\Phi^+),(\map{B}\otimes\map{I}_{\spc{K}'})(\Phi^+)\right),
\end{align}
where  $\Phi^+:=|\Phi^+\>\<\Phi^+|$ with $|\Phi^+\>\in\spc{K}_{\rm in}\otimes\spc{K}'$ being the maximally entangled state of the system and a reference. It is straightforward by definition that $F_{\rm wc}\le F_{\rm ent}$.

On the other hand, one can also quantify the performance of simulation by error measures. Similarly as for faithfulness, there are also two commonly considered measures of error, defined with respect to the worst-case input and the maximally entangled state. The first is the  worst-case input error, or the \emph{diamond norm error}, defined as:
\begin{align}
\epsilon_{\rm wc}(\map{A},\map{B}):=&\frac12\|\map{A}-\map{B}\|_{\diamond}\\
=&\frac{1}{2}\sup_{\rho\in \St(\spc{K}_{\rm in}\otimes\spc{K}')}\left\|(\map{A}\otimes\map{I}_{\spc{K}'})(\rho)-(\map{B}\otimes\map{I}_{\spc{K}'})(\rho)\right\|_1,\label{epsilon-wc}
\end{align}
where $\|\cdot\|_1$ is the trace norm and $\|\cdot\|_{\diamond}$ denotes the \emph{diamond norm} \cite{kitaev1997quantum}. Correspondingly, we have the \emph{entanglement error}:
\begin{align}\label{epsilon-ent}
\epsilon_{\rm ent}(\map{A},\map{B}):=\frac12\left\|(\map{A}\otimes\map{I}_{\spc{K}'})(\Phi^+)-(\map{B}\otimes\map{I}_{\spc{K}'})(\Phi^+)\right\|_1.
\end{align}
Again, by definition, we have $\epsilon_{\rm ent}\le\epsilon_{\rm wc}$.

The above faithfulness measures and error measures are not independent \cite{gilchrist2005distance}. For instance, the well-known Fuchs-van de Graaf inequalities \cite{fuchs1999cryptographic}, which relate fidelities to distances as
\begin{align}\label{fuchs-van de graaf}
1-\sqrt{F(\rho,\sigma)}\le\frac12\|\rho-\sigma\|_1\le \sqrt{1-F(\rho,\sigma)}\qquad\forall\,\rho,\sigma
\end{align}
can be readily employed. Besides, the worst-case input error can be bounded using the following inequality 
\begin{align}\label{relation_wc_ent}
\epsilon_{\rm wc}(\map{A},\map{B})\le d_{\rm S}\cdot \epsilon_{\rm ent}(\map{A},\map{B}),
\end{align}
where $d_{\rm S}$ is the dimension of the system \cite[Exercise 3.6]{watrous2018theory}.

\subsection{Quantum reference frames}\label{subsec-RF}

In this subsection we introduce quantum reference frames, which will be the key ingredient of our construction of covariant error correcting codes.
We will consider  entangled reference frames, as they offer superior performance \cite{acin2001optimal,chiribella2004efficient,chiribella2005optimal}. 
A quantum reference frame is essentially an entangled state $|\psi'\>_{AB}$, constructed on $2m$ qudits and shared by two parties Alice and Bob. Assume that the reference frames of Alice and Bob are misaligned up to a rotation $U\in\grp{SU}(d)$: If Bob prepares one of his qudits in, say, a state $|k\>$ in the computational basis, it would actually be a rotated version of the corresponding state on Alice's side, i.e., $|k\>_B=U|k\>_A$.

To align their reference frames, Bob sends his part of the state through a distortion-free channel to Bob. Now, Alice has full disposal of the bipartite state $|\psi'\>_{AA'}=( U^{\otimes m}\otimes I^{\otimes m})|\psi\>$, where $|\psi\>$ is the state when there is no reference frame misalignment. By making a suitable measurement, Alice obtains information on the misalignment $U$ to a good precision.

In this work, we will mainly use a quantum reference frame in a single-party but mathematically equivalent scenario. The role of quantum reference frames in our construction of covariant error correcting codes is to store an unknown unitary rotation. 
By the Schur-Weyl duality \cite{fulton2013representation}, we have the following decomposition of the product Hilbert space into the direct sum of irreducible subspaces:
\begin{align}
\spc{H}^{\otimes m}\simeq\bigoplus_{\Vec\lambda\in\settt{Y}_{m}}(\spc{H}_{\Vec\lambda}\otimes \spc{M}_{\Vec\lambda}).
\end{align}
Here $\settt{Y}_m$ is the collection of Young diagrams of $m$ boxes, $\spc{H}_{\Vec\lambda}$ and $\spc{M}_{\Vec\lambda}$ denote the irreducible subspace and the multiplicity subspace associated to the Young diagram $\Vec\lambda$, respectively. Notice that $\spc{M}_{\Vec\lambda}$ also depends on $m$ but this dependence is not important in this work. The notation ``$\simeq$" means that the two spaces are equivalent up to a change of basis, which is called the Schur transform and can be implemented efficiently on a quantum computer \cite{bacon2006efficient,krovi2019efficient}. 

In our work, it is enough to consider (quantum) reference frame states of the form  \cite{chiribella2005optimal}
\begin{align}\label{input-state-form}
|\psi\>=\bigoplus_{\Vec\lambda\in\settt{Y}_{m}}\sqrt{q_{\Vec\lambda}}|\Phi^+_{\Vec\lambda}\>\otimes|\Phi^+_{m_{\Vec\lambda}}\>,
\end{align}
where $|\Phi^+_{\Vec\lambda}\>$ and $|\Phi^+_{m_{\Vec\lambda}}\>$ denote the maximally entangled states of the irreducible subspace and the multiplicity subspace respectively, and $\{q_{\Vec\lambda}\}$ is a probability distribution that uniquely characterises the reference frame state.
Here the state is built on the coupled basis of $2m$ qudit systems, half of which serve as a reference. Explicitly, it lies within the Hilbert space $ \bigoplus_\Vec\lambda \spc{H}_\Vec\lambda\otimes\spc{H}'_{\Vec\lambda}\otimes\spc{M}_\Vec\lambda\otimes\spc{M}'_\Vec\lambda\subset \spc{H}^{\otimes m}\otimes\spc{H}^{\otimes m}$, where $\spc{H}_\Vec\lambda$ ($\simeq\spc{H}'_\Vec\lambda$) is the representation subspace and $\spc{M}_\Vec\lambda$ ($\simeq\spc{M}'_\Vec\lambda$) is the multiplicity subspace.

The reference frame state lies in a subspace of $\spc{H}^{\otimes m}\otimes\spc{H}^{\otimes m}$. To it we apply  $(U\otimes I)^{\otimes m}$ on it,\footnote{A covariant code in the usual sense should have realisation $U^{\otimes 2m}$ on it. Here instead we realise it as $(U\otimes I)^{\otimes m}$, which is still both mathematically rigorous ($U\otimes I$ is still a representation of $U$) and good for the physical implementation (just run fewer operations).\label{footnote1}} resulting in a state of the form
\begin{align}\label{probe-state-form}
|\psi_{U}\>=\bigoplus_{\Vec\lambda}\sqrt{q_\Vec\lambda}|\Phi^+_{U,\Vec\lambda}\>\otimes|\Phi^+_{m_\Vec\lambda}\>.
\end{align}
The optimal measurement to extract the information of $U$ from the above state is the covariant POVM $\{\d\hat{U}, |\eta_{\hat{U}}\>\<\eta_{\hat{U}}|\}$ \cite{chiribella2005optimal}, where $\d\hat{U}$ is the Haar measure and $|\eta_{\hat{U}}\>$ is the following vector
\begin{align}\label{optimal-POVM}
|\eta_{\hat{U}}\>:=\bigoplus_{\Vec\lambda}d_{\Vec\lambda}|\Phi^+_{\hat{U},\Vec\lambda}\>\otimes |\Phi^+_{m_\Vec\lambda}\>,
\end{align}
where $d_{\Vec\lambda}$ denotes the dimension of the irreducible subspace characterised by $\Vec\lambda$.
Denoting by $\chi_{U,\Vec\lambda}:=\Tr[U_{\Vec\lambda}]$ the characters of $\grp{SU}(d)$, the probability density function of getting the outcome $\hat{U}$ when the actual gate is $U$ can be expressed as
\begin{align}\label{est-prob-dist-covariant}
p(\hat{U}|U)=\bigg|\sum_{\Vec\lambda}\sqrt{q_\Vec\lambda}\chi_{U\hat{U}^{-1},\Vec\lambda}\bigg	|^2.
\end{align}
This distribution is nice because it is conjugate invariant, i.e., $p(WUV|WV)=p(U|I)=p(U^\dag|I)$ for any $W, V,$ and $U$.
The above defines a series of estimation strategies, each specified by a choice of $\{q_\Vec\lambda\}$. Later, we will make specific choices that fulfill the goal of constructing covariant codes.

In the error correction scenario, we have to deal with errors occurring on the reference frames. To this purpose, we divide the $n_{\rm R}$ qudits assigned to the reference frame register into $s_{\rm R}>1$ groups of identical systems, i.e., $\spc{H}_{\rm R}=\spc{H}_{{\rm R},1}\otimes\spc{H}_{{\rm R},2}\otimes\cdots\otimes\spc{H}_{{\rm R},s_{\rm R}}$ and construct an individual reference frame state of the form in Eq.~(\ref{input-state-form}) on each of them. We assume that the error occurring on each of the reference frames is detectable, i.e., it takes the state on $\spc{H}_{{\rm R},k}$ ($\forall\,k$) to an error space $\spc{H}_{{\rm err},k}$ that is orthogonal to $\spc{H}_{{\rm R},k}$. This is the case, for instance, for erasure errors.

For detectable errors, we will use the reference frames in the following way:
\begin{enumerate}
 \item We prepare all the reference frame registers in the same state (\ref{input-state-form}) and denote the overall state by $\Psi:=\psi^{\otimes s_{\rm R}}$. Here, since each $\psi$ is constructed on $2m$ qudits, we have $2m\,s_{\rm R} =n_{\rm R}$.
\item For each copy of $\psi$, we perform a binary POVM to detect whether an error has happened. For each label $k$, the POVM is defined by projections on the error space $\spc{H}_{{\rm R, err},k}$ and the original space $\spc{H}_{{\rm R},k}$.
\item We discard those erroneous reference frames and measure \emph{jointly} the remaining ones with the measurement defined in Eq.\ (\ref{optimal-POVM}).  
\end{enumerate}
Denote by $\spc{H}_{{\rm err, R}}:=\spc{H}_{{\rm err, R},1}\otimes\cdots\otimes\spc{H}_{{\rm err, R},s_{\rm R}}$ the total error space. The POVM element of our error-tolerant measurement is $\{|\eta_{\hat{U}}\>\<\eta_{\hat{U}}|_{\settt{s}}\otimes P_{{\rm err},\settt{s}^c}, \d\hat{U}\}_{\settt{s}\subset[s_{\rm R}]}$, where $\{|\eta_{\hat{U}}\>\<\eta_{\hat{U}}|_{\settt{s}},\d\hat{U}\}$ corresponds to a joint measurement on $\bigotimes_{k\in\settt{s}}\spc{H}_{{\rm R},k}$ with $|\eta_{\hat{U}}\>$ defined by Eq.\ (\ref{optimal-POVM}), and $P_{{\rm err},\settt{s}^c}:=\bigotimes_{k\in\settt{s}^c}P_{{\rm err},k}$ is the projection on the error spaces of all $k\not\in\settt{s}$.
Alternatively, the measurement is characterised by a linear map that yields a probability density function $\map{M}_{\hat{U}}:\St(\spc{H}_{\rm R}\otimes\spc{H}_{{\rm err, R}})\to\R$:
\begin{align}\label{error-tolerant-povm}
\map{M}_{\hat{U}}(\cdot):=\sum_{\settt{s}\subset[s_{\rm R}]}\Tr\left[(\cdot)|\eta_{\hat{U}}\>\<\eta_{\hat{U}}|_{\settt{s}}\otimes P_{{\rm err},\settt{s}^c}\right].
\end{align}
When an error $\map{C}_{\rm R}$ occurs, the probability density function of the outcome $\hat{U}$ is just $\map{M}_{\hat{U}}\circ\map{C}_{\rm R}(\Psi)$.

\subsection{Covariant quantum error correction}\label{subsec-covqec}
Given $n$ qudit systems in total, the goal of covariant quantum error correction is to construct a covariant code $(\map{E}_\cov,\map{D}_\cov)$ that encodes one logical qudit, i.e., $d_{\rm L}=d$ into the $n$ qudit systems. This is to say that the encoder $\map{E}_\cov$ satisfies
\begin{align}\label{covariance2}
\map{E}_\cov\circ\map{V}_{\rm L}=\left(\map{V}_1\otimes\cdots\otimes\map{V}_n\right)\circ\map{E}_\cov
\end{align}
for any $V\in\grp{SU}(d)$, where $\map{V}_i$ ($i=1,\dots,n$) is either $\map{V}$ or the identity channel. In this way, any logical gate can be realised \emph{transversally} by implementing the same gate on each of the qudit systems, which is the desideratum of fault-tolerant quantum computing.

Nevertheless, the Eastin-Knill theorem rules out any covariant, finite dimensional code that perfectly corrects local errors. One way to circumvent this restriction is to consider approximate codes that are still covariant but correct an error $\map{C}$ only $\epsilon$-well:
\begin{align}
\map{D}_\cov\circ\map{C}\circ\map{E}_\cov\approx_{\epsilon}\map{I}_{\rm L}.
\end{align}
This was first achieved for $\grp{U}(1)$ covariance using finite dimensional reference frames in \cite{woods2020continuousgroups} and using other techniques in \cite{faist2019continuous}.  
%In this work, we construct $\grp{SU}(d)$ covariant approximate codes using finite dimensional reference frames, evaluate their performances, and show their optimality. 
The main idea of our construction is to divide the $n$ qudit systems into two registers: the reference frame register $\rm R$ (consisting of $n_{\rm R}$ qudits), and the computational register $\Cspace$ (consisting of $n_\Cspace=n-n_{\rm R}$ qudits). We denote by $\map{U}_\Cspace:=\map{U}^{\otimes n_\Cspace}$ and $\map{U}_{\rm L}:=\map{U}$ the actions of an element $U\in\grp{SU}(d)$ on the computational register and the logical register, and by $\map{U}_{\rm R}:=(\map{U}\otimes\map{I})^{\otimes n_{\rm R}/2}$ its action on the reference frame register, were recall  $\map{U}(\cdot):=U(\cdot)U^\dag$.
Our covariant code is built on an arbitrary, possibly non-covariant, code $(\map{E},\map{D})$ on the computational register, i.e.\null{} $\map{E}:\spc{H}_{\rm L}\to\spc{H}_\Cspace$. We make $\map{E}$ and $\map{D}$ covariant via the technique of twirling. The role of the reference frame register is to keep track of the random unitaries applied in the twirling and to keep the subroutine code functioning.

In particular, the covariant encoder is defined by Eq.~\eqref{E-cov} and the decoder is defined by Eq.~\eqref{D-cov}
with $\map{M}_{\hat U}$ the measurement from \eqref{error-tolerant-povm}.

%Here is our covariant protocol of implementing a logical qudit gate:

\begin{algorithm}[H]
  \caption{Covariant implementation of $\map{V}_{\rm L}$ (see also Figure \ref{fig_Protocol1}).} 
   \begin{algorithmic}%[1]
  % \grpcounterref{ALC@line}{0} 
  \State {\em Preparation:} \newline Initiate each reference register in the reference frame state $\Psi$. \newline
   \State {\em Encoder $\map{E}_\cov$:}\newline Apply $\map{U}_\Cspace\circ\map{E}\circ\map{U}_{\rm L}^{-1}$ to the logical system and $\map{U}_{\rm R}$ to the reference, with $U$ following the Haar measure.\newline
   \State {\em Gate implementation:}\newline Apply $\map{V}_\Cspace$ on the computational system and $\map{V}_{\rm R}$ on the reference.\newline
   \State {\em Error $\map{C}$:} \newline An error may occur on the computational register $\Cspace$ and the reference $\rm R$. 
   \newline
   \State {\em Decoder $\map{D}$:} \newline  Measure the reference frame with the error-tolerant measurement $\map{M}_{\hat{U}}$ from \eqref{error-tolerant-povm}. Depending on the measurement outcome $\hat{U}$, apply $\hat{\map{U}}_{\rm L}\circ\map{D}\circ\hat{\map{U}}_\Cspace^{-1}$ on the computational system.
   \end{algorithmic}\label{protocol-cov}
\end{algorithm}
 
%The protocol is illustrated in Figure \ref{fig_Protocol1}.
    \begin{figure*}[tb]
		\centering

\begin{tikzpicture}[auto,thick,text depth=.125ex,text height=1.5ex]
\tikzstyle{box}=[draw,minimum width=9mm,minimum height=6mm,fill=white]
\tikzstyle{state}=[draw,%rounded corners=3mm,minimum width=8mm,minimum height=8mm,fill=white]
semicircle,shape border rotate=90,inner sep=1pt,minimum width=8mm]
\pgfmathsetmacro\MathAxis{height("$\vcenter{}$")}
\matrix[row sep=3mm,column sep=4mm] {
\node (p1) {}; &[2mm]
\node (Ulinv) [box] {$\mathcal{U}^{\scriptscriptstyle {-}1}$}; &
\node (E) [box] {$\mathcal E$}; & 
\node (Up) [box] {$\mathcal{U}$}; & 
\node (V0) [box] {$\mathcal V$}; & 
\node (c0) {}; & 
\node (i0) {}; &[-5mm]
\node (j0) [box] {$\hat{\mathcal U}^{\scriptscriptstyle -1}$}; &
\node (k0) [box] {$\mathcal D$}; &
\node (l0) [box] {$\hat{\mathcal U}$}; &[2mm]
\node (pr) {}; \\%&[-2mm]
%\node (eq) {$\approx$}; &[-2mm]
%\node (ql) {}; & 
%\node [box] (logical) {$\mathcal V$}; &
%\node (qr) {}; \\
& & 
\node (psi1) {}; &
\node (Ur1) [box] {$\mathcal{U}$}; &
\node (V1) [box] {$\mathcal{V}$}; &
\node (c1) [minimum width=8mm] {}; &
\node (j1) [box,white,fill=white] {$\phantom{\mathcal U}$};\\[-5mm]
& & 
\node (psi1b) {}; &
\node (Ur1b) [] {$\phantom{\mathcal{U}}$}; &
\node (V1b) [] {}; &
\node (c1b) [minimum width=8mm] {}; &
\node (j1b) [box,white,fill=white] {$\phantom{\mathcal U}$};\\[-1mm]
& & \node (psi2) {}; &
\node (Ur2) [box] {$\mathcal{U}$}; &
\node (V2) [box] {$\mathcal{V}$}; &
\node (c2) {}; &
\node (j2) [box,white,fill=white] {$\phantom{\mathcal U}$};\\[-5mm]
& & \node (psi2b) {}; &
\node (Ur2b) [] {}; &
\node (V2b) [] {}; &
\node (c2b) {}; &
\node (j2b) [box,white,fill=white] {$\phantom{\mathcal U}$};\\
}; 
\draw (p1) -- node[pos=0.34,above=-2pt] {\tiny \textsf{L}} (Ulinv) -- (E) -- node[midway,above=-2pt] {\tiny \textsf{P}} (Up) -- (V0) -- (j0) -- (k0) -- node[midway,above=-2pt] {\tiny \textsf{L}} (l0) -- (pr);
%\draw (ql) -- (logical) -- (qr);
\draw (psi1) -- node[near end,above=-2pt] {\tiny \textsf{R}} (Ur1) -- (V1) -- (j1);
\draw (psi1b) --  (j1b);
\draw (psi2) -- node[near end,above=-2pt] {\tiny \textsf{R}} (Ur2) -- (V2) -- (j2);
\draw (psi2b) -- (c2b) -- (j2b);
\coordinate (ppsi1c) at ($0.5*(psi1)+0.5*(psi1b)$);
\node (ppsi1) at (E.east |- ppsi1c) [box,rounded rectangle,rounded rectangle right arc=none,anchor=east,minimum height=8mm] {$\psi$};
\coordinate (ppsi2c) at ($0.5*(psi2)+0.5*(psi2b)$);
\node (ppsi2) at (E.east |- ppsi2c) [box,rounded rectangle,rounded rectangle right arc=none,anchor=east,minimum height=8mm] {$\psi$};
\path (V1) -- (V2) node [midway] (fa) {};
\node (b) at (V2.south east -| j0.south west) [inner sep=0,outer sep=0] {};

\coordinate (mh) at ($0.5*(psi2)+0.5*(psi1b)$);
\coordinate (mc) at (mh -| j2.west);
\def\bl{10mm}
\coordinate (t0) at ($(mc) + (0,\bl)$);
\coordinate (t1) at ($(mc) - (0,\bl)$);

\coordinate (t3) at ($(mc -| j2.east) + (-.145mm,0)$);

\coordinate (ctrl1) at (t3 -| j0);
\coordinate (ctrl2) at (t3 -| l0);

\draw[double distance=1pt,semithick,rounded corners] (t3) -- (ctrl2) -- (l0);
\draw[double distance=1pt,semithick] (ctrl1) -- (j0);
\fill[black] (ctrl1) circle(1.5pt);

\draw (t1) to[arc through ccw=(t3)] (t0) -- cycle;
\node at (mc) [anchor=west,minimum height=40pt] {$\mathcal M_{\hat {\mathcal U}}$};

\node [draw=red,densely dashed,rounded corners,fit=(Ulinv) (Up) (Ur2) (ppsi2)] (enc) {};
\node [draw=blue,densely dashed,rounded corners,fit=(j0) (l0) (j2) (t1) (ppsi2.south east -| t1)] (dec) {};
\node at (enc.south west) [above right] {$\mathcal E_{\rm{cov}}$};
\node at (dec.south east) [above left] {$\mathcal D_{\rm{cov}}$};

\coordinate (ccenter0) at ($0.5*(ppsi2.south east)+0.5*(E.north east)$);
\coordinate (ccenter) at (ccenter0 -| c1);

\coordinate (cn) at (ccenter |- E.north);
\coordinate (cs) at (ccenter |- ppsi2.south);
\coordinate (cnw) at ($(cn) - (2.5mm,0) - (0,1mm)$);
\coordinate (cne) at ($(cn) + (2.5mm,0) - (0,1mm)$);
\coordinate (csw) at ($(cs) - (2.5mm,0) + (0,1mm)$);
\coordinate (cse) at ($(cs) + (2.5mm,0) + (0,1mm)$);
\node[draw,fill=white,fit=(cnw) (cne) (cse) (csw)] {};
\node at (ccenter) {$\mathcal C$};
\end{tikzpicture}
		\caption{\footnotesize
			\textbf{Reference frame assisted covariant error correction.} An arbitrary code $(\map{E},\map{D})$ can be converted into a covariant code $(\map{E}_\cov,\map{D}_\cov)$ following the instruction of Protocol \ref{protocol-cov}. The subroutine encoder $\map{E}:\spc{H}_{\rm L}\to\spc{H}_\Cspace$ is made covariant via twirling, where $U$ is a random unitary sampled from the Haar measure. A reference frame register, which is divided into multiple copies of a reference frame state $\psi$, is employed to keep track of the unitary. Since $\map{E}_\cov$ is covariant, any desired logical operation $V$ can be implemented transversally. To correct error (characterised by a channel $\map{C}$) or to decode the information, the reference frame register is first measured, and recovery operations are performed on the computational register depending on the measurement outcome.} 
		\label{fig_Protocol1}
	\end{figure*}   
  
\subsection{General bound on the error}
There are two contributing factors to errors in the covariant scheme: noise occurring between encoding and decoding, and errors in the decoding due to the finiteness of the reference frames. 
Here we show that our protocol is able to recover from both of these. 
We quantify the error by the diamond norm error (\ref{epsilon-wc})
% are interested in 
% In this subsection, we derive an upper bound on the error of our protocol $(\map{E}_\cov,\map{D}_\cov)$, quantified by the diamond norm error:
\begin{align}\label{error-diamond}
\epsilon_\cov:=\max_{V\in\grp{SU}(d)}\epsilon_{\wc}\left(\map{D}_\cov\circ\map{C}\circ(\map{V}_\Cspace\otimes\map{V}_{\rm R})\circ\map{E}_\cov,\map{V}_{\rm L}\right)\,,
\end{align}
where $\map{C}$ is a noisy channel. 
Here we focus on errors that are detectable on the reference frames (see Section \ref{subsec-RF}) and covariant.
We assume that the error can be expressed in the form
\begin{align}\label{general-form-error}
\map{C}=\sum_j p_j \map{C}_j,
\end{align}
where each $\map{C}_j$ is of the product form
\begin{align}
\map{C}_j=\map{C}_{j,\Cspace}\otimes\map{C}_{j,{\rm R}}
\end{align}
acting on the computational register and the reference frame register, respectively.
For the subroutine code pair $(\map{E},\map{D})$, let $\epsilon_{j,{\rm code}}=\epsilon_{\wc}(\map{D}\circ\map{C}_{j,\Cspace}\circ\map{E},\map{I}_{\rm L})$ denote the worst-case error for noise operator $\map{C}_j$. 
Then we can show 
\begin{lemma}\label{lemma-cov-error-bound}
For covariant error of the form (\ref{general-form-error}), the error of our construction (see Protocol \ref{protocol-cov}) is upper bounded as
\begin{align}
\label{eq:ecovbound}
\epsilon_{\cov}\le 9d\cdot\sum_j p_j\epsilon_j,
\end{align}
where 
\begin{align}
\epsilon_j=&\max\bigg\{  \nonumber \\
&\,\epsilon_{j,{\rm code}},1-F_{\wc}\left(\int\d U~p_j(U|I)~\map{U}_\Cspace\otimes\map{U}^\ast_{\rm L},\map{I}_\Cspace\otimes\map{I}_{{\rm L}}\right)\bigg\}\,
\end{align}
for $F_{\wc}\left(\int\d U~p_j(U|I)~\map{U}_\Cspace\otimes\map{U}^\ast_{\rm L},\map{I}_\Cspace\otimes\map{I}_{{\rm L}}\right)\ge 3/4$; otherwise $\epsilon_j= 1$. Here 
\begin{align}\label{pj}
p_j(\hat{U}|U):=\map{M}_{\hat{U}}\circ\map{C}_{j,{\rm R}}\circ\map{U}_{\rm R}(\Psi)
\end{align}
is the probability distribution of the measurement outcome conditioned on $U$ being applied and error $\map{C}_j$ occurring.
\end{lemma}
We remark that, since we mainly focus on the small error regime, the condition $F_\wc\ge3/4$ in the above lemma is usually guaranteed.
\begin{proof}
Since the protocol $(\map{E}_\cov,\map{D}_\cov)$ and the error $\map{C}$ are both covariant, the protocol fares equally well for any $\map{V}_{\rm L}$, and thus we have 
\begin{align}
\epsilon_\cov=\epsilon_{\wc}\left(\map{D}_\cov\circ\map{C}\circ\map{E}_\cov,\map{I}_{\rm L}\right).
\end{align}
Substituting the expressions of $\map{E}_\cov$ [see Eq.\ (\ref{E-cov})] and $\map{D}_\cov$ [see Eq.\ (\ref{D-cov})] as well as $\map{C}=\sum_j p_j \map{C}_j$ into the definition, we obtain
\begin{widetext}
\begin{align}
\epsilon_\cov&\le \sum_j p_j \epsilon_j\\
\epsilon_j&:=\epsilon_{\wc}\left(\int\d U\int\d\hat{U}~\left(\hat{\map{U}}_{\rm L}\circ\map{D}\circ\hat{\map{U}}_\Cspace^{-1}\otimes \map{M}_{\hat{U}}\right)\circ\map{C}_j\circ\left(\map{U}_\Cspace\circ\map{E}\circ\map{U}_{\rm L}^{-1}\otimes\map{U}_{\rm R}(\Psi)\right),\map{I}_{\rm L}\right),
\end{align}
having used the joint convexity of the diamond norm error.
Since each $\map{C}_j$ is decomposed as $\map{C}_{j,\Cspace}\otimes\map{C}_{j,{\rm R}}$, we find
\begin{align}
\epsilon_j&=\epsilon_{\wc}\left(\int\d U\int\d\hat{U}~p_j(\hat{U}|U)\,\hat{\map{U}}_{\rm L}\circ\map{D}\circ\hat{\map{U}}_\Cspace^{-1}\circ\map{C}_{j,\Cspace}\circ\map{U}_\Cspace\circ\map{E}\circ\map{U}_{\rm L}^{-1},\map{I}_{\rm L}\right)\,,
\end{align}
where $p_j(\hat{U}|U):=\map{M}_{\hat{U}}\circ\map{C}_{j,{\rm R}}\circ\map{U}_{\rm R}(\Psi)$, as defined by Eq.~(\ref{pj}),
is the probability of getting the outcome $\hat{U}$ when the embedded unitary is $U$. $p_j(\hat{U}|U)$ is conjugate-invariant, and thus $p_j(\hat{U}|U)=p_j(U'^\dag|I)=p_j(U'|I)$ with $U':=\hat{U}^\dag U$. Substituting into the above expression, we have
\begin{align}
\epsilon_j&=\epsilon_{\wc}\left(\int\d U \int\d U'~p_j(U'|I)\,\map{U}_{\rm L}\circ\map{U}'^{-1}_{\rm L}\circ\map{D}\circ\map{U}'_\Cspace\circ\map{C}_{j,\Cspace}\circ\map{E}\circ\map{U}_{\rm L}^{-1},\map{I}_{\rm L}\right).
\end{align}  
Define the twirled version of the original protocol as:
\begin{align}
\map{P}_{\rm twirl}:=\int\d U~\map{U}_{\rm L}\circ\map{D}\circ\map{C}_{j,\Cspace}\circ\map{E}\circ\map{U}_{\rm L}^{-1}.
\end{align}
Using Lemma \ref{lemma-cov-error-bound-app} in Appendix \ref{app-cov-error}, we bound the error as
\begin{align}\label{general-error-inter0}
\epsilon_j\le 9d\cdot\max\left\{\tilde{\epsilon}_{j,{\rm code}},1-F_{j,{\rm RF}}\right\}.
\end{align}
where 
\begin{align}
F_{j,{\rm RF}}:=F_{\ent}\left(\int\d U\int\d U'~p_j(U'|I)\,\map{U}_{\rm L}\circ\map{U}'^{-1}_{\rm L}\circ\map{D}\circ\map{U}'_\Cspace\circ\map{C}_{j,\Cspace}\circ\map{E}\circ\map{U}_{\rm L}^{-1},\map{P}_{\rm twirl}\right)
\end{align}
is the entanglement fidelity for the reference frame correction and
\begin{align}
\tilde{\epsilon}_{j,{\rm code}}:=\epsilon_{\wc}\left(\map{P}_{\rm twirl},\map{I}_{\rm L}\right)
\end{align}
is the twirled code error. Notice that, by Lemma \ref{lemma-cov-error-bound-app}, Eq.\ (\ref{general-error-inter0}) holds only if 
\begin{align}
\epsilon_{\ent}\left(\int\d U\int\d U'~p_j(U'|I)\,\map{U}_{\rm L}\circ\map{U}'^{-1}_{\rm L}\circ\map{D}\circ\map{U}'_\Cspace\circ\map{C}_{j,\Cspace}\circ\map{E}\circ\map{U}_{\rm L}^{-1},\map{P}_{\rm twirl}\right)\le\frac{1}{2}.
\end{align}
By the Fuchs-van de Graaf inequalities (\ref{fuchs-van de graaf}), this holds as long as $F_{j,{\rm RF}}\ge 3/4$.

On one hand, by the joint convexity of the diamond norm error, the twirled code error is bounded by the code error, i.e.,
\begin{align}\label{general-error-inter1}
\tilde{\epsilon}_{j,{\rm code}}\le\epsilon_{j,{\rm code}}.
\end{align}
On the other hand, by the joint concavity of the square-root fidelity and the unitary invariance of fidelity, we have
\begin{align}
\sqrt{F}_{j,{\rm RF}}\ge&\int\d U\sqrt{F}_{\ent}\left(\int\d U'~p_j(U'|I)\,\map{U}_{\rm L}\circ\map{U}'^{-1}_{\rm L}\circ\map{D}\circ\map{U}_\Cspace'\circ\map{C}_{j,\Cspace}\circ\map{E}\circ\map{U}_{\rm L}^{-1},\map{U}_{\rm L}\circ\map{D}\circ\map{C}_{j,\Cspace}\circ\map{E}\circ\map{U}_{\rm L}^{-1}\right)\nonumber\\
=&\int\d U\sqrt{F}_{\ent}\left(\int\d U'~p_j(U'|I)\,\map{U}'^{-1}_{\rm L}\circ\map{D}\circ\map{U}_\Cspace'\circ\map{C}_{j,\Cspace}\circ\map{E},\map{D}\circ\map{C}_{j,\Cspace}\circ\map{E}\right)\nonumber\\
=&\sqrt{F}_{\ent}\left(\int\d U'~p_j(U'|I)\,\map{U}'^{-1}_{\rm L}\circ\map{D}\circ\map{U}'_\Cspace\circ\map{C}_{j,\Cspace}\circ\map{E},\map{D}\circ\map{C}_{j,\Cspace}\circ\map{E}\right).
\end{align}
Using the maximally entangled state to ``wire around'' the channels, we get
\begin{align}
F_{j,{\rm RF}}&\ge F\left(\int\d U~p_j(U|I)~\left(\map{D}\circ\map{U}_\Cspace\circ\map{C}_{j,\Cspace}\circ\map{E}\otimes\map{U}^\ast_{\rm L}\right)(\Phi^+_{\rm L}),(\map{D}\circ\map{C}_{j,\Cspace}\circ\map{E}\otimes\map{I}_{\rm R})(\Phi^+_{\rm L})\right)\nonumber\\
&\ge F_{\wc}\left(\int\d U~p_j(U|I)~\left(\map{D}\circ\map{U}_\Cspace\circ\map{C}_{j,\Cspace}\circ\map{E}\otimes\map{U}^\ast_{\rm L}\right),\map{D}\circ\map{C}_{j,\Cspace}\circ\map{E}\otimes\map{I}_{\rm L}\right).
\end{align}
\end{widetext}
Here $U^\ast$ denotes the complex conjugation of $U$.
Exploiting the data processing inequality of fidelity, we have 
\begin{align}\label{general-error-inter2}
F_{j,{\rm RF}}&\ge  F_\wc\left(\int\d U~p_j(U|I)~\map{U}_\Cspace\otimes\map{U}^\ast_{\rm L},\map{I}_\Cspace\otimes\map{I}_{{\rm L}}\right).
\end{align}
Finally, combining Eqs.\ (\ref{general-error-inter0}), (\ref{general-error-inter1}), and (\ref{general-error-inter2}) gives \eqref{eq:ecovbound}.
\end{proof}
%Notice that we can actually replace $F_\wc$ in the above inequality by the worst-case input fidelity without reference. This would benefit us at most with a constant reduction of the error. 

The bound in Lemma~\ref{lemma-cov-error-bound} applies to covariant codes constructed from arbitrary reference frames. 
For reference frames as in \eqref{input-state-form}, we can give a more detailed bound. 
Since $\epsilon_{j,{\rm code}}$ is specified by the non-covariant code $(\map{E},\map{D})$ that we use as a subroutine, what we need to do is to bound the reference frame error by determining 
\begin{align}\label{weak-err-Fwc}
F_{\wc}\left(\int\d U~p_j(U|I)~\map{U}_\Cspace\otimes\map{U}^\ast_{\rm L},\map{I}_\Cspace\otimes\map{I}_{{\rm L}}\right),
\end{align}
where $p_j(U|I)$ is the distribution defined by Eq.\ (\ref{pj}).

From now on we focus on erasure errors, which are indeed covariant and detectable. For erasure errors, as the reference frame register is initiated in multiple copies of the state $\psi$ defined by Eq.\ (\ref{input-state-form}), the error map $\map{C}_{j,{\rm R}}$ will ruin some of the copies, and the remaining copies can still be cast in the form (\ref{input-state-form}). 
For simplicity, we still denote the distribution [as in Eq.\ (\ref{input-state-form})] associated to the state of the remaining copies by $\{q_\Vec\lambda\}$, keeping in mind that this distribution may depend on $j$.
Then, by using Eq.\ (\ref{est-prob-dist-covariant}) we can write $p_j(U|I)$ as:
\begin{align}
p(U|I)=\bigg|\sum_{\Vec\lambda\in\settt{S}_{\rm via}}\sqrt{q_\Vec\lambda}\chi_{U^{-1},\Vec\lambda}\bigg|^2,
\end{align}
where $\settt{S}_{\rm via}$ is the ``viable'' set of Young diagrams on which $q_\Vec\lambda>0$.

Define $n':=n_\Cspace+d-1$ and the ``interior'' subset of $\settt{S}_{\rm via}$ as:
\begin{align}\label{S-int}
\settt{S}_{\rm int}:=\left\{\lambda\in\settt{S}_{\rm via}~:~|\lambda_i-\lambda_j|> 3n'\ \forall\,i\not=j\right\}.
\end{align}
Then we can show that:
\begin{lemma}\label{lemma-Fwc}
For reference frames of the form in Eq.~\eqref{input-state-form}, 
\begin{align}
&F_{\wc}\left(\int\d U~p_j(U|I)~\map{U}_\Cspace\otimes\map{U}^\ast_{\rm L},\map{I}_\Cspace\otimes\map{I}_{{\rm L}}\right)\nonumber\\
&\geq \min_{\Vec\Delta\in\settt{S}_{\rm diff}}\sum_{\Vec\lambda\in\settt{S}_{\rm int}}\sqrt{q_\Vec\lambda q_{\Vec\lambda+\Vec\Delta}},\label{general-fid-inter0}
\end{align}
where $\settt{S}_{\rm diff}=\big\{\Vec\Delta\in\Z^{\times d}~:~|\Delta_i|\le n', i=1,\dots,d; \sum_{j=1}^d\Delta_j=0\big\}$.
\end{lemma}

\begin{proof} 
Denote by $\settt{S}_{\rm cost}$ the collection of all Young diagrams appearing in the decomposition of $\map{U}_{\rm L}^\ast\otimes\map{U}_\Cspace$, i.e., the set of Young diagrams corresponding to the ``cost'' function. Notice that, since the computational register consists of $n_\Cspace=n-n_{\rm R}$ qudits and the logical register consists of one qudit, $\settt{S}_{\rm cost}$ does not contain any $\lambda$ with more than $n':=n_\Cspace+d-1$ boxes. We have the decomposition
\begin{align}
\map{U}_{\rm L}^\ast\otimes\map{U}_\Cspace\simeq\bigoplus_{\Vec\lambda\in\settt{S}_{\rm cost}}U_\Vec\lambda\otimes I_{m_\Vec\lambda}
\end{align}
where $m_\Vec\lambda$ are the multiplicities.

We now convert the quantity (\ref{weak-err-Fwc}) into a form that is easier to bound.
For conjugate-invariant $p_j(U|I)$, the channel of interest $\int\d U~p_j(U|I)~\bigoplus_{\Vec\lambda\in\settt{S}_{\rm cost}}(\map{U}_\Vec\lambda\otimes \map{I}_{m_\Vec\lambda})$ is block-covariant  with respect to the symmetry $\bigoplus_{\Vec\lambda\in\settt{S}_{\rm cost}}(\map{V}_\Vec\lambda\otimes\map{I}_{m_\Vec\lambda})$ for $V\in\grp{SU}(d)$. We can now apply the following lemma (see Subsection \ref{app-wcinput} for the proof):
\begin{lemma}\label{lemma-wcinput}
The worst-case input fidelity of a channel $\map{E}$ commuting with unitary channels of the block diagonal form $\bigoplus_\Vec\lambda (\map{U}_\Vec\lambda\otimes\map{I}_{m_\Vec\lambda})$ for any $U\in\grp{SU}(d)$ can be achieved by an input state of the following form:
\begin{align}\label{Psi-form}
|\Psi^\ast\>:=\bigoplus_\Vec\lambda c_\Vec\lambda |\Phi^+_\Vec\lambda\>\otimes|\psi_{m_\Vec\lambda}\>
\end{align}
with $|\Phi^+_\Vec\lambda\>$ being the maximally entangled state in $\spc{H}_\Vec\lambda\otimes\spc{H}'_\Vec\lambda$, $\{c_\Vec\lambda\}$ being an amplitude distribution, and $|\psi_{m_\Vec\lambda}\>$ being a fixed (otherwise arbitrary) state on the multiplicity subspace.
\end{lemma}
Then, the worst-case input fidelity (\ref{weak-err-Fwc}) can be achieved by a state of the form (\ref{Psi-form}).

Notice that, since the error $\map{C}_{j,{\rm R}}$ simply destroys a few copies of the reference state, the remaining copies can still be cast in the form (\ref{probe-state-form}). Therefore, $p_j(U|I)$ is of the form (\ref{est-prob-dist-covariant}).
We now combine the above facts with Lemma~\ref{lemma-wcinput} and express the fidelity (\ref{weak-err-Fwc}) as
\begin{align}
&F_{\wc}\left(\int\d U~p_j(U|I)~\map{U}_\Cspace\otimes\map{U}^\ast_{\rm L},\map{I}_\Cspace\otimes\map{I}_{{\rm L}}\right)\nonumber\\
&=\sum_{\Vec\lambda,\Vec\lambda'\in\settt{S}_{\rm via}}\sqrt{q_\Vec\lambda q_{\Vec\lambda'}}S_{\Vec\lambda,\Vec\lambda'},
\end{align}
where $\settt{S}_{\rm via}$ is the set of Young diagrams on which $q_\Vec\lambda>0$,  $S_{\Vec\lambda,\Vec\lambda'}$ is a correlation function defined as
\begin{align}\label{S-correlation}
S_{\Vec\lambda,\Vec\lambda'}:=\int\d U~ \chi_{U,\Vec\lambda}\chi^\ast_{U,\Vec\lambda'}\left|\sum_{\Vec\mu\in\settt{S}_{\rm cost}}|c_\Vec\mu|^2\,d_\Vec\mu^{-1}\chi_{U,\Vec\mu}\right|^2
\end{align}
and $\{c_\Vec\mu\}$ is an amplitude distribution over $\settt{S}_{\rm cost}$. Notice that here $\{q_\Vec\lambda\}$ corresponds to the reference frame state when the error syndrome $\map{C}_{j,{\rm R}}$ takes place, and we abbreviated the index $j$ for simplicity.

By the orthogonality of the characters, $S_{\Vec\lambda,\Vec\lambda'}\ge0$ depends on the overlap between the irreducible decomposition of $\Vec\lambda\otimes \Vec\mu$ and the irreducible decomposition of $\Vec\lambda'\otimes\Vec\mu'$. Since $\settt{S}_{\rm cost}$ contains only diagrams with no more than $n'$ boxes, $S_{\Vec\lambda,\Vec\lambda'}=0$ unless 
\begin{align}
d_{\rm Young}(\Vec\lambda,\Vec\lambda')\le n',
\end{align}
 where $d_{\rm Young}(\Vec\lambda,\Vec\mu):=\frac12\sum_{i=1}^d|\lambda_i-\mu_i|$ is the distance between Young diagrams. 
 %Still, it seems hard to evaluate $S_{\lambda,\lambda'}$. Nevertheless, we now make a specification of  $\{q_\lambda\}$, which yields a good lower bound of the fidelity without too much efforts.

Now we turn back to the fidelity. 
First, since $S_{\Vec\lambda,\Vec\lambda'}\ge 0$ with equality for $d_{\rm Young}(\Vec\lambda,\Vec\lambda')>n'$, we have
\begin{align}\label{fid-weak-error-inter1}
F_{\wc}\ge\sum_{\Vec\lambda\in\settt{S}_{\rm int}}\sum_{\Vec\Delta\in\settt{S}_{\rm diff}}\sqrt{q_\Vec\lambda q_{\Vec\lambda+\Vec\Delta}}S_{\Vec\lambda,\Vec\lambda+\Vec\Delta}\,,
\end{align}
where $\settt{S}_{\rm int}$ is the interior subset of $\settt{S}_{\rm via}$ defined by Eq.~(\ref{S-int}).
%where $\settt{S}_{\rm diff}=\{\Delta\in\Z^{\times d}~:~|\Delta_i|\le n', i=1,\dots,d; \sum_{j=1}^d\Delta_j=0\}$.
Crucially, we now argue that for $\Vec\lambda\in\settt{S}_{\rm int}$ and $\Vec\Delta\in\settt{S}_{\rm diff}$, the correlation function depends only on their relative distance $\Vec\Delta$, but not explicitly on $\Vec\lambda$, i.e., 
\begin{align}\label{ans}
S_{\Vec\lambda,\Vec\lambda+\Vec\Delta}=\tilde{S}_{\Vec\Delta}.
\end{align}
\begin{widetext}
Invoking Eq.\ (\ref{S-correlation}), the correlation function can be expressed as
\begin{align}
S_{\Vec\lambda,\Vec\lambda+\Vec\Delta}=\sum_{\Vec\mu,\Vec\mu'\in\settt{S}_{\rm cost}}|c_\Vec\mu c_{\Vec\mu'}|^2(d_\Vec\mu d_{\Vec\mu'})^{-1}\int\d U~ \chi_{U,\Vec\lambda}\chi_{U,\Vec\mu}\chi^\ast_{U,\Vec\lambda+\Vec\Delta}\chi_{U,\Vec\mu'}^\ast.
\end{align}
By orthogonality of the characters, the correlation function is just
\begin{align}\label{cor-inter}
S_{\Vec\lambda,\Vec\lambda+\Vec\Delta}&=\sum_{\Vec\mu,\Vec\mu'\in\settt{S}_{\rm cost}}|c_\Vec\mu c_{\Vec\mu'}|^2(d_\Vec\mu d_{\Vec\mu'})^{-1}C_{\Vec\mu,\Vec\mu'}^{\Vec\lambda,\Vec\lambda+\Vec\Delta}\\
C_{\Vec\mu,\Vec\mu'}^{\Vec\lambda,\Vec\lambda+\Vec\Delta}&=\left|\left\{(L_{\Vec\nu/\Vec\lambda}^{\Vec\mu\to\tilde{\Vec\mu}},L_{\Vec\nu/(\Vec\lambda+\Vec\Delta)}^{\Vec\mu'\to\tilde{\Vec\mu}'})~:~\Vec\nu\in\Vec\lambda\otimes\Vec\mu,\,\Vec\nu\in(\Vec\lambda+\Vec\Delta)\otimes\Vec\mu'\right\}\right|.
\end{align}
\end{widetext}
Here $L_{\Vec\nu/\Vec\lambda}^{\Vec\mu\to\tilde{\Vec\mu}}$ denotes a Littlewood-Richardson tableau of shape $\Vec\nu/\Vec\lambda$ with content $\tilde{\Vec\mu}$, obtained by adding the Young diagram $\Vec\mu$ to $\Vec\lambda$ according to the Littlewood-Richardson rule, and $\Vec\nu\in\Vec\lambda\otimes\Vec\mu$ means that $\Vec\nu$ appears at least once in the decomposition of $\Vec\lambda\otimes\Vec\mu$. Now one can see why Eq.\ (\ref{ans}) holds. Indeed, $\settt{S}_{\rm int}$ is so defined [cf.\ Eq.~(\ref{S-int})] that, when $\Vec\lambda\in\settt{S}_{\rm int}$, the lengths of different rows of $\Vec\lambda$ and $\Vec\lambda+\Vec\Delta$  have big enough gaps so that even adding all $n'$ boxes to one row would not make its box number greater than its preceding rows. Therefore,  according to the Littlewood-Richardson rule adding $\Vec\mu'$ and $\Vec\mu$, neither has more than $n'$ boxes, is not constraint by the shape of $\Vec\lambda$ and $\Vec\lambda+\Vec\Delta$.
The sum $\sum_\Vec\Delta C_{\Vec\mu,\Vec\mu'}^{\Vec\lambda,\Vec\lambda+\Vec\Delta}$ is determined by how many different contents can any $\Vec\mu\in\settt{S}_{\rm cost}$ possibly generate.

The following property will be useful: Rectifications of Littlewood-Richardson tableaux and representative Young tableaux (of symmetry tensors) are in one-to-one correspondence. Indeed, every representative Young tableau  is a standard (i.e.\null{} left $\le$ right and top $<$ bottom) tableau. It corresponds to a rectified Littlewood-Richardson tableau, whose $j$th row has a number $x$ of the index $i$ with $x$ being the number of $j$s in the $i$th row of the representative Young tableau. A rectified Littlewood-Richardson tableau of content $\tilde{\Vec\mu}$ thus corresponds to the representative Young tableau $\tilde{\Vec\mu}$. Since the total number of representative Young tableaux is the dimension of the irreducible representation, the total number of contents that $\Vec\mu\in\settt{S}_{\rm cost}$ can generate is $d_{\Vec\mu}$.
Therefore, we have
\begin{align}
\sum_{\Vec\Delta}C_{\Vec\mu,\Vec\mu'}^{\Vec\lambda,\Vec\lambda+\Vec\Delta}=d_\Vec\mu d_{\Vec\mu'}
\end{align}
for any $\Vec\mu,\Vec\mu'\in\settt{S}_{\rm cost}$. Combining with Eqs.\ (\ref{ans}) and (\ref{cor-inter}), we have
\begin{align}
\sum_{\Vec\Delta\in\settt{S}_{\rm diff}}\tilde{S}_{\Vec\Delta}&=\sum_{\Vec\mu,\Vec\mu'\in\settt{S}_{\rm cost}}|c_\Vec\mu c_{\Vec\mu'}|^2=1.\label{sum-SDelta}
\end{align}

\begin{widetext}
Now, the fidelity bound (\ref{fid-weak-error-inter1}) becomes
\begin{align}
F_{\wc}\left(\int\d U~p_j(U|I)~\map{U}_\Cspace\otimes\map{U}^\ast_{\rm L},\map{I}_\Cspace\otimes\map{I}_{{\rm L}}\right)&\ge\sum_{\Vec\lambda\in\settt{S}_{\rm int}}\sum_{\Vec\Delta\in\settt{S}_{\rm diff}}\sqrt{q_\Vec\lambda q_{\Vec\lambda+\Vec\Delta}}\tilde{S}_{\Vec\Delta}\nonumber\\
&\ge \left(\min_{\Vec\Delta\in\settt{S}_{\rm diff}}\sum_{\Vec\lambda\in\settt{S}_{\rm int}}\sqrt{q_\Vec\lambda q_{\Vec\lambda+\Vec\Delta}}\right)\left(\sum_{\Vec\Delta'\in\settt{S}_{\rm diff}}\tilde{S}_{\Vec\Delta'}\right)\nonumber\\
&=\min_{\Vec\Delta\in\settt{S}_{\rm diff}}\sum_{\Vec\lambda\in\settt{S}_{\rm int}}\sqrt{q_\Vec\lambda q_{\Vec\lambda+\Vec\Delta}}\,,\label{general-fid-inter}
\end{align}
\end{widetext}
and the proof is complete.
\end{proof}
%Therefore, the performance of our protocol can be evaluated by combining Lemma \ref{lemma-cov-error-bound} with the bound (\ref{general-fid-inter}).

From the above discussion, we obtain a bound (by combining Lemma \ref{lemma-cov-error-bound} with Lemma \ref{lemma-Fwc}) on the performance of our protocol.
In the following sections, we will use it to evaluate the performance of our protocol for two different erasure error models:
\begin{enumerate}
\item {\bf Weak erasure error} (Section \ref{sec-weak-error}). At most $n_{\rm e}$ out of the $n$ qudits are erased, with $n_{\rm e}$ a constant independent of $n$.  
\item {\bf i.i.d.\null{} erasure error} (Section \ref{sec-strong-error}). Each qudit has a constant probability of being erased. The errors on different qudits are independent.
%leading to a numer of erasures which is a constant fraction of $n$ with high probability. %Every single qudit has a chance of being erased. Essentially, 

\end{enumerate}

\subsection{A lower bound on $\epsilon_\cov$}

In the last part of this section, we derive a lower bound on $\epsilon_\cov$ of any covariant code, which will be used to show the optimality of our protocol.
The lower bound is built upon the main result of Ref.~\cite{kubica2020using}, further strengthened here by us. This bound was derived for $\grp{U(1)}$, but we can always find a single-parameter family of unitaries embedded in $\grp{SU}(d)$, so the bound applies in general. For instance, we can consider the $\grp{U(1)}$ family $U_\theta:=e^{-i\theta H}$ generated by the Hamiltonian 
\begin{align}
H=\sum_{j=1}^{d} h_j|j\>\<j|
\end{align}
for an orthonormal basis $\{|j\>\}_{j=1}^d$ of $\spc{H}$. The $\grp{SU}(d)$ covariance implies the $\grp{U(1)}$ covariance, i.e., $U^{\otimes n}_\theta\circ\map{E}_\cov=\map{E}_\cov\circ\map{U}_\theta$ for any $\theta$.

In Ref.~\cite{kubica2020using}, it was shown for covariant codes the worst-case input fidelity obeys the lower bound:
\begin{align}\label{original-wc-bound}
\sqrt{1-F_{\wc}}\ge \frac{(\Delta H)^2}{3\sqrt{6}I_{\rm Fisher}^\uparrow}.
\end{align}
Here $\Delta H$ is the difference between the maximum eigenvalue and the minimum eigenvalue of $H$, and $I_{\rm Fisher}^\uparrow$ is a Fisher information upper bound of the channel  $\map{C}_{\theta}:=\map{C}\circ\map{U}^{\otimes n}_\theta$ (with $\map{C}$ being the error) that can be evaluated as follows: if there exists Kraus operators $\{K_{l;\theta}\}$ of the channel $\map{C}_\theta$ such that 
$\sum_l \dot{K_{l;\theta}}^\dag K_{l;\theta}=0$ 
(here $\dot{K_{l;\theta}}$ denotes the derivative of $K_{l;\theta}$ with respect to $\theta$), the Fisher information upper bound is again bounded as
\begin{align}\label{I-Fisher}
I_{\rm Fisher}^\uparrow\le 4\bigg\|\sum_{l}\dot{K_{l;\theta}}^\dag \dot{K}_{l;\theta}\bigg\|_{\infty}.
\end{align}
Here $\|\cdot\|_{\infty}$ denotes the operator norm. If there does not exist such a Kraus form, we set $I_{\rm Fisher}^\uparrow=\infty$ and the bound (\ref{original-wc-bound}) is trivial.

The above bound, however, is not enough to show the optimality of our result. Instead, we derive a strengthened version of it:
\begin{lemma}[Theorem 1 of \cite{kubica2020using}; strengthened version]\label{lemma-converse0}
The error of any covariant code is lower bounded by
\begin{align}\label{diamond-converse}
\epsilon_{\cov}\ge \frac{(\Delta H)^2}{16I_{\rm Fisher}^\uparrow}\,.
\end{align}
Here $\Delta H$ is the difference between the maximum eigenvalue and the minimum eigenvalue of $H$, and $I_{\rm Fisher}^\uparrow$ is the Fisher information upper bound (\ref{I-Fisher}).
\end{lemma}
To see Eq.\ (\ref{diamond-converse}) is indeed a strengthening of Eq.\ (\ref{original-wc-bound}), simply notice that $\epsilon_\cov\le\sqrt{1-F_\wc}$ \cite{fuchs1999cryptographic} and thus Eq.\ (\ref{diamond-converse}) implies Eq.\ (\ref{original-wc-bound}). On the other hand, plugging the other part of the Fuchs-van de Graaf inequality $\epsilon_\cov\ge 1-\sqrt{F_{\wc}}$ into Eq.\ (\ref{original-wc-bound}) only yields a bound on $\epsilon_\cov$ that scales as $(I_{\rm Fisher}^\uparrow)^{-2}$ which, as we will soon see, is not enough to prove the desired $1/n^2$-scaling bound on $\epsilon_\cov$. For this reason, we must use the strengthened version.

\medskip

\noindent{\bf Proof of Lemma \ref{lemma-converse0}.}
This result can be derived along the same line of arguments as the proof of the original bound (\ref{original-wc-bound}) in Ref.\ \cite{kubica2020using}. A few improvements need to be made as the following:
First, in Eq.\ (A11) of \cite{kubica2020using}, we directly consider the worst-case error $\epsilon_\cov$ and employ the tighter bound (than the Fuchs-van de Graaf inequality) between $f^2$ an $\epsilon_{\wc}$ when one of the two states is pure; see, e.g., \cite[Eq.\ (9.111)]{nielsen2000quantum}.
Then Eq.\ (A11) becomes
\begin{align}
f^2\ge 1-\epsilon_{\cov},
\end{align}
where $f^2$ is the same quantity as in the original equation.  
With this improved inequality substituted into Eq.\ (A20), we get the counterpart of \cite[Lemma 1]{kubica2020using} for $\epsilon_{\wc}$:
\begin{align}
I_{\rm Fisher}^\uparrow\ge (1-4\epsilon_{\cov})(\Delta H)^2.
\end{align}
Eq.\ (21) of \cite{kubica2020using} then becomes $(m\Delta H)^2(1-4m\epsilon_{\cov})\le m I_{\rm Fisher}^\uparrow$, and optimising over $m$ we get the bound (\ref{diamond-converse}).\hfill\qed

\medskip

\subsection{On the error of covariant channels.}\label{app-cov-error}
Let $\map{A}$ and $\map{B}$ be channels acting on a $d$-dimensional Hilbert space $\spc{H}$. Assume that $\map{A}$ and $\map{B}$ are both covariant with the (full) symmetry group $\grp{SU}(d)$ on $\spc{H}$. 
Define the entanglement error:
\begin{align}
	\nonumber \epsilon_{\rm ent}(\map{X}, \map{Y}) &:=\frac12\|X-Y \|_1 \, ,
\end{align}
where $X$ and $Y$ are the Choi states of two channels $\map{X}$ and $\map{Y}$, respectively.

Here we prove the following lemma:
\begin{lemma}\label{lemma-cov-error-bound-app}
Suppose $\epsilon_{\ent}(\map{A},\map{B})\le 1/2$. The worst-case input error of $\map{B}$ can be bounded as
\begin{align}
\epsilon_{\wc}(\map{B},\map{I})\le 9d\cdot\max\left\{\epsilon_{\ent}(\map{A},\map{I}),1-F_{\ent}\left(\map{A},\map{B}\right)\right\}.
\end{align}
Here $\map{I}$ denotes the identity channel on $\spc{H}$.
\end{lemma}
\proof
In the following, we make frequent use of an elementary relation between the worst-case input error  and the entanglement error:
\begin{align}\label{error_wc_ent}
	\epsilon_{\rm wc}(\map{X}, \map{Y}) \leq d \cdot \epsilon_{\rm ent}(\map{X}, \map{Y}) \, .
\end{align}

For any channel $\map{A}$, define its Choi state as
\begin{align}
	A := \Big(\map{A} \otimes \map{I} \Big) (\Phi^+_d) \, 
\end{align}
with $\Phi^+_d$ being the maximally entangled state in $\spc{H}\otimes\spc{H}$.
When $\map{A}$ is covariant, we have
\begin{align}
	[A, \, U \otimes U^\ast] = 0, \qquad \forall\, U \in {\rm SU}(d) \, .
\end{align}
By Schur's lemma, the Choi states of the covariant channels $\map{A}$ and $\map{B}$ can be decomposed as 
\begin{align}
&A = (1-a) \cdot  \Phi^+_d + a\cdot \rho^{\perp}\qquad\rho^{\perp}:=\frac{1}{d^2-1} \Big(I\otimes I - \Phi^+_d\Big) \\
&B = (1-b) \cdot  \Phi^+_d + b\cdot \rho^{\perp}. 
\end{align}
Therefore, the entanglement error and the entanglement fidelity that we are interested in can be evaluated as:
\begin{align}
\epsilon_{\ent}(\map{A},\map{I})&=a\\
\epsilon_{\ent}(\map{B},\map{I})&=b\\
F_{\ent}(\map{A},\map{B})& = \left(\sqrt{(1-a)(1-b)}+\sqrt{ab}\right)^2 \\
\epsilon_{\ent}(\map{A},\map{B})&=|a-b|. 
\end{align}
If $a\ge b$, we have 
\begin{align}\label{app-error-inter1}
\epsilon_\wc(\map{B},\map{I})&\le d\cdot\epsilon_\ent(\map{B},\map{I})= d\cdot b\le d\cdot a.
\end{align}
If $a< b$, we can write $b=a+\epsilon_\ent(\map{A},\map{B})$. We further distinguish between two cases: If $a>\epsilon_\ent(\map{A},\map{B})/8$, then $b\le 9a$ and thus
\begin{align}\label{app-error-inter2}
\epsilon_\wc(\map{B},\map{I})\le 9d\cdot a. 
\end{align}
\begin{widetext}
Otherwise, if $a\le\epsilon_\ent(\map{A},\map{B})/8$, we have (using the shorthand $\epsilon_\ent:=\epsilon_\ent(\map{A},\map{B})$)
\begin{align*}
1-F_{\ent}(\map{A},\map{B}) 
&=1- \left(\sqrt{(1-a)(1-a-\epsilon_\ent)}+\sqrt{a(a+\epsilon_\ent)}\right)^2 \nonumber\\
&= \epsilon_\ent+2a(1-a-\epsilon_\ent)-2\sqrt{a(1-a)(a(1-a)+\epsilon_\ent(1-2a-\epsilon_\ent))} \nonumber\\
&\ge \epsilon_\ent+2a(1-a-\epsilon_\ent)-2\sqrt{2a(1-a)\epsilon_\ent} \\
&=(\sqrt{\epsilon_\ent}-\sqrt{2a(1-a)})^2-2a\epsilon_\ent \\
&\ge\left(\sqrt{\epsilon_\ent}-\frac{\sqrt{\epsilon_\ent}}2\right)^2-\frac{\epsilon_\ent^2}{4} \\
&\ge \frac{\epsilon_\ent(1-\epsilon_\ent)}{4}\\
&\ge \frac{\epsilon_\ent}{8}.
\end{align*} 
\end{widetext}
Notice that the last inequality holds if $\epsilon_\ent(\map{A},\map{B})\le 1/2$. Then, the worst-case input error of $\map{B}$ can be bounded as
\begin{align}\label{app-error-inter3}
\epsilon_\wc(\map{B},\map{I})&\le d\cdot (a+\epsilon_\ent(\map{A},\map{B}))\le d\cdot \left(a+8(1-F_{\ent}(\map{A},\map{B}))\right).
\end{align}
Finally, we get the desired bound by summarising Eqs.~(\ref{app-error-inter1}), (\ref{app-error-inter2}), and (\ref{app-error-inter3}) into a more compact form. \qed

\subsection{Proof of Lemma \ref{lemma-wcinput}}\label{app-wcinput}
Here we show that the worst-case input fidelity of a channel $\map{E}$ commuting with
\begin{align}
\map{U}_{\rm tot}:=\bigoplus_\Vec\lambda \map{U}_\Vec\lambda\otimes\map{I}_{m_\Vec\lambda}
\end{align}
can be achieved by an input state of the following form:
\begin{align}\label{app-Psi-form}
|\Psi^\ast\>:=\bigoplus_\Vec\lambda c_\Vec\lambda |\Phi^+_\Vec\lambda\>\otimes|\psi_{m_\Vec\lambda}\>.
\end{align} 

The worst-case input fidelity of $\map{E}$ can be written as a SDP (with the Slater's condition always satisfied).
In particular, from Lemma~\ref{lemma:sdp} (see later) the primal problem for the worst-case input (square-root) fidelity is:
\begin{align}\label{app-wc-sq-fid}
\sqrt{F}_{\wc}(\map{E},\map{I})= &\,{\rm min}\ \frac12\left(\Tr (E_{BA}\Gamma_{BA})+\Tr(\Omega_{BA}\Lambda_{BA})\right)\\
&\ {\rm such\ that}\ \left(\begin{matrix}\Gamma_{BA} & - I_B\otimes\rho_A^T \\ -I_B\otimes\rho_A^T & \Lambda_{BA}\end{matrix} \right)\ge 0 \notag\\
&\qquad\qquad\qquad\Tr(\rho_A^T)\ge 1\notag\\
&\qquad\qquad\qquad\rho_A,\Gamma_{BA},\Lambda_{BA}\ge0\notag.
\end{align}
Here $\rho_A:=\Tr_R\Psi_{AR}$ corresponds to the marginal of the input state, $\Omega_{BA}:=|I\kk\bb I|$, and $E_{BA}:=(\map{E}\otimes\map{I}_A)(\Omega)$. 

Suppose that $|\Psi\>$ with marginal $\rho_A$ is an input state achieving the minimum. 
Consider the twirling 
\begin{align}
\map{T}(\cdot):=\int\d U \left(\bigoplus_\Vec\lambda \map{U}_\Vec\lambda\otimes \map{U}_\Vec\lambda^\ast\otimes\map{I}_{m_\Vec\lambda}\otimes\map{I}_{m_\Vec\lambda}\right)(\cdot)
\end{align}
 on the constraints. The constraints become
\begin{align}
&\left(\begin{matrix}\map{T}(\Gamma_{BA}) & - I_B\otimes\tilde{\rho}_A^T \\ -I_B\otimes\tilde{\rho}_A^T & \map{T}(\Lambda_{BA})\end{matrix} \right)\ge 0 \\
&\Tr(\tilde{\rho}_A^T)\ge 1\\
&\tilde{\rho}_A,\map{T}(\Gamma_{BA}),\map{T}(\Lambda_{BA})\ge0.
\end{align}
Here $\tilde{\rho}_A:=\int\d U\,\bigoplus_\Vec\lambda\map{U}^\ast_\Vec\lambda\otimes\map{I}_{m_\Vec\lambda}(\rho_A)$. Noticing that both $E_{BA}$ and $\Omega_{BA}$ are invariant under $\map{T}$ and its dual $\map{T}^\dag$, the objective function remains the same:
\begin{align}
&\frac12\left(\Tr (E_{BA}\map{T}(\Gamma_{BA}))+\Tr(\Omega_{BA}\map{T}(\Lambda_{BA}))\right)\nonumber\\
=&\frac12\left(\Tr (E_{BA}\Gamma_{BA})+\Tr(\Omega_{BA}\Lambda_{BA})\right).
\end{align}
Therefore, an input $|\Psi^*\>$ with marginal $\tilde{\rho}_A$ also achieves the worst-case input fidelity. Since it is invariant under $\bigoplus_\Vec\lambda\map{U}^\ast_\Vec\lambda$, by Schur's lemma it has to be of the diagonal form
\begin{align}
\tilde{\rho}_A=\bigoplus_\Vec\lambda |c_\Vec\lambda|^2 I_\Vec\lambda\otimes\sigma_{m_\Vec\lambda}
\end{align}
with $\{\sigma_{m_\Vec\lambda}\}$ being arbitrary states on the multiplicity subspaces.
Therefore, the corresponding input state $|\Psi^\ast\>$ can be cast into the desired form (\ref{app-Psi-form}).

At last we prove the formulation (\ref{app-wc-sq-fid}) of the worst-case input fidelity as a semidefinite program, adapting the method used for the diamond norm \cite{watrous2013}. First, for a channel $\map{A}:\Lin( \spc K_{\rm in})  \to \Lin(\spc K_{\rm out})$, define its Choi operator $A\in \Lin(\spc K_{\rm out}\otimes \spc K_{\rm in}')$ with $\spc K_{\rm in}'\simeq \spc K_{\rm in}$ as follows. 
For $\ket{\Omega}\in \spc{K}_{\rm in}\otimes \spc{K}_{\rm in}'$ the unnormalized maximally-entangled state, set $A = \big(\map{A} \otimes \map{I}\big) (\Omega)$.

%Here $\rho_A:=\Tr_R\Psi_{AR}$ corresponds to the marginal of the input state, $\Omega_{BA}:=|I\kk\bb I|$, and $E_{BA}:=(\map{E}\otimes\map{I}_A)(\Omega)$. 

\begin{lemma}
\label{lemma:sdp}
Given any two channels $\mathcal A$ and $\mathcal B$ with Choi operators $A$ and $B$, respectively, 
\begin{align}
\sqrt{F_{\rm wc}(\map{A},\map{B})}
=\begin{array}[t]{rl}
\text{min} &
\tfrac12(\Tr[A\Gamma]+\Tr[B\Lambda])\\[1mm]
\text{s.t.} &
\begin{pmatrix} \Gamma & -I\otimes \rho\\ -I\otimes \rho & \Lambda\end{pmatrix}\geq 0\,,\\
&\rho,\Gamma,\Lambda\geq 0\,,\\
&\rho \in \St(\spc K_{\rm in})\,,\\
& \Gamma,\Lambda\in \Lin (\spc K_{\rm out}\otimes \spc K_{\rm in})\,.
\end{array}
\end{align}
\end{lemma}

\begin{proof}
The first step in the proof is to simplify the dependence on the input state. 
%express the worst-case input fidelity as an optimization involving the Choi operators of the channels. 
Using the dual form of the fidelity function from \cite{watrous2013,killoran2012} in \eqref{F-wc}, we have 
\begin{align}
\sqrt{F_{\wc}(\map A,\map B)} = \begin{array}[t]{rl}
\text{min} &
\tfrac12(\Tr[Y \big(\map{A}\otimes \map {I}_{\rm R}\big)(\rho)]+\Tr[Z \big(\map{B}\otimes \map {I}_{\rm R}\big)(\rho)])\\
\text{s.t.} 
&\begin{pmatrix} Y & -\id\\ -\id & Z\end{pmatrix}\geq 0\,,\\
&\rho\in \St(\spc K_{\rm in}\otimes \spc K_{\rm R})\,, \\
&Y,Z \in \Lin (\spc K_{\rm out}\otimes \spc K_{\rm R}) \,.
\end{array}
\end{align}
Observe that we can assume the optimal input state $\rho$ is pure because the objective function is linear in $\rho$.  
Therefore, we can write it as $\rho=K\Omega K^\dagger$, where $\Omega\in \Lin(\spc{K}_{\rm in}\otimes \spc{K}_{\rm in}')$ and $K:\spc{K}_{\rm in}'\to \spc{K}_{{\rm R}}$ is the operator defined by $\ket{\psi}\mapsto \bra{\Omega}(\ket{\rho}\otimes \ket{\psi})$.
Now let $K$ be the variable in the optimization.
By construction, $\rho\geq 0$, and the trace constraint becomes $\Tr[K^\dagger K]=1$. 
Note that $K^\dagger K\in \Lin(\spc{H}_{\rm in}')$. 

In the first term of the objective function we recognize the Choi operator of $\map{A}$:
\begin{align}
\Tr[Y \big(\map{A}\otimes \map {I}_{\rm R}\big)(\rho)]
&= \Tr[Y K \big(\map{A}\otimes \map{I})(\Omega) K^\dagger ]\\
&=\Tr[K^\dagger Y K A]\,,
\end{align}
and this works similarly for the second term. 
Thus, the optimization takes the form 
\begin{align}
\label{eq:simplifiedform}
\sqrt{F_{\wc}(\map A,\map B)} = \begin{array}[t]{rl}
\text{min} &
\tfrac12(\Tr[K^\dagger Y K A]+\Tr[K^\dagger Z K B])\\
\text{s.t.} 
&\begin{pmatrix} Y & -\id\\ -\id & Z\end{pmatrix}\geq 0\,,\\
&\Tr[K^\dagger K]=1\,,\\
& K\in \Lin(\spc{K}_{\rm in}'\to \spc K_{\rm R})\,, \\
&Y,Z \in \Lin (\spc K_{\rm out}\otimes \spc K_{\rm R}) \,.
\end{array}
\end{align}

Now we want to move to different variables, but without changing the optimal value. 
Defining $\Gamma$ and $\Lambda$ in $\Lin(\spc{K}_{\rm out}\otimes \spc{K}_{\rm in}')$ by $\Gamma=K^\dagger Y K$ and $\Lambda=K^\dagger Z K$, the objective function becomes $\tfrac12(\Tr[\Gamma A]+\Tr[\Lambda B])$.  
To deal with the constraints, let $M=\text{diag}(K,K)$ and conjugate the block matrix in the constraint by $M$, multiplying from the left by $M^\dagger$ and the right by $M$.
Doing so puts $\Gamma$ and $\Lambda$ on the diagonal, and $-\id\otimes K^\dagger K$ on the off-diagonal. 
Note that the block matrix constraint implies $Y\geq 0$ and $Z\geq 0$, which by construction then implies $\Gamma\geq 0$ and $\Lambda \geq 0$. 
Defining $\rho\in \Lin(\spc{K}_{\in}')$ as $\rho=K^\dagger K$, we have an optimization in the variables  $\Gamma$, $\Lambda$ and $\rho$.  
However, conjugation generally relaxes the constraints, which could lead to a smaller minimum value than the original optimization. 
Ignoring the difference between $\spc{K}_{\rm in}'$ and $\spc{K}_{\rm in}$, we have established
\begin{align}
\label{eq:nearlythere}
\sqrt{F(\map{A},\map{B})}\leq \begin{array}[t]{rl}
\text{min} &
\tfrac12(\Tr[\Gamma A]+\Tr[\Lambda {B}])\\[1mm]
\text{s.t.} &
\begin{pmatrix} \Gamma & -I\otimes \rho\\ -\id\otimes \rho & \Lambda\end{pmatrix}\geq 0\,,\\
&\rho,\Gamma,\Lambda\geq 0\,,\\
&\rho \in \St(\spc K_{\rm in})\,,\\
& \Gamma,\Lambda\in \Lin (\spc K_{\rm out}\otimes \spc K_{\rm in})\,.
\end{array}
\end{align}

To establish equality in \eqref{eq:nearlythere} and complete the proof, we show that any feasible variables in \eqref{eq:nearlythere} can be converted into feasible variables in \eqref{eq:simplifiedform} having the same value of the objective function.
Nominally, the following choice will work. 
Pick an arbitrary isometry $V:\spc{K}_{\rm in}\to \spc{K}_{\rm R}$ and define 
\begin{align}
K&=V\rho^{\nicefrac 12}\,,\\
Y&=V\rho^{-\nicefrac 12}\Gamma \rho^{-\nicefrac 12} V^\dagger\,,\\
Z&=V\rho^{-\nicefrac 12}\Lambda \rho^{-\nicefrac 12} V^\dagger\,.
 \end{align}
However, the inverse of $\rho$ is potentially problematic, as we do not know that $\rho$ is full rank. 
We can avoid this problem as follows (which could presumably also be done by continuity). 
Suppose that $P\in \Lin(\spc{K}_{\rm in})$ is the projection onto the support of $\rho$.
Then $\Gamma'= P \Gamma P$ and $\Lambda'= P\Lambda P$ are also feasible in \eqref{eq:nearlythere}, as we can conjugate the constraints by $P$. 
Moreover, these variables will not have a larger value of the objective function, so we may as well begin the argument with feasible variables of this form.  

Now we make a slight modification of the above choice, using the inverse on the support of $\rho$:
\begin{align}
K&=V \rho^{\nicefrac 12}\,,\\
Y&=V\left( \rho^{-\nicefrac 12}\Gamma \rho^{-\nicefrac 12}+I\otimes (I- P)\right) V^\dagger \,,\\
Z&=V\left( \rho^{-\nicefrac 12}\Lambda \rho^{-\nicefrac 12}+I\otimes (I- P)\right) V^\dagger \,.
 \end{align}

Given that the support of $\Gamma$ is contained in that of $I\otimes \rho$, it follows that $\Tr[K^\dagger Y  K A ]=\Tr[\Gamma A]$, and similarly for the other term. Hence this choice of variables leads to the same value of the objective function.  
Feasibility in \eqref{eq:simplifiedform} also holds. 
The positivity and trace conditions hold immediately, and only the block matrix constraint is a little more involved. 
Define $P'\in \Lin(\spc{K}_{\rm R})$ by $P'=VPV^\dagger$. 
Then conjugating the block matrix in \eqref{eq:nearlythere} by $M=\text{diag}(L,L)$ for $L=I\otimes V \rho^{-\nicefrac 12}$ gives
\begin{align}
\begin{pmatrix}
V \rho^{-\nicefrac 12}\Gamma \rho^{-\nicefrac 12}V^\dagger & -I\otimes P' \\ -I\otimes P'  & V \rho^{-\nicefrac 12}\Lambda \rho^{-\nicefrac 12} V^\dagger
\end{pmatrix}
\geq 0\,.
\end{align}
To this inequality we can add 
\begin{align}
\begin{pmatrix}
I \otimes (I-P') & -I \otimes (I-P') \\ -I \otimes (I-P')  & I \otimes (I-P') 
\end{pmatrix}\geq 0\,,
\end{align}
and the result is the block matrix constraint in \eqref{eq:simplifiedform}. 
\end{proof}
Note that a different SDP for $F_{\wc}(\map{A},\map{I})$ (not the square root) was given in \cite[Equation A12]{faist2019practical}. After completion of this work, we discovered that the SDP of Lemma~\ref{lemma:sdp} also appears as Proposition 50 of \cite{katariya_geometric_2020}.

\section{Heisenberg-limited error correction in the weak error model}\label{sec-weak-error}
\subsection{Setting}
In this section, we consider a relatively weak type of errors (compared to the other type we will consider). Specifically, the error is that at most $n_{\rm e}$ qudits among the $n$ qudit systems composing the computational register and the reference frame register, are randomly lost:  
\begin{align}\label{weak-error-model}
\map{C}=p_0\map{I}+\sum_{\settt{s}\subset[n]:|\settt{s}|\le n_{\rm e}} p_{\settt{s}}\left(\map{C}_{\rm e}\right)_{\settt{s}},
\end{align}
where $\left(\map{C}_{\rm e}\right)_{\settt{s}}$ denotes the erasure of qudits whose labels are in the set $\settt{s}$ and $\{p_{\settt{s}}\}$ is a probability distribution.

Any quantum error-correcting code over qudits of distance at least $k+1$ can perfectly correct $k$ erasures~\cite{Gottesman2006}. For instance, the polynomial codes of Aharonov and Ben-Or are $[[2k+1,1,k+1]]_d$ (one logical qudit encoded into $2k+1$ computational qudits; code distance $k+1$) stabilizer codes with this property~\cite{aharonov_fault-tolerant_1997}. 
In these cases, we can employ one of the perfect codes as the non-covariant subroutine $(\map{E},\map{D})$ of our code. Since we only need to encode one logical qudit, the perfect code requires only $n_\Cspace=O(1)$ computational qudits. 

This is to say: the (non-covariant) code $(\map{E},\map{D})$ we use requires only $O(1)$ computational qudits and satisfies
\begin{align}
\epsilon_{\settt{s}_\Cspace,{\rm code}}=0\qquad\forall\,|\settt{s}_\Cspace|\le n_{\rm e},
\end{align}
where $\settt{s}_\Cspace\subset[n_\Cspace]$ is the set of indices for error locations on the computational register.
Now we arrange the reference frame register against this type of noise. To this purpose, we divide the $n_{\rm R}$ qudit there into $n_{\rm e}+1$ groups, each consisting of $2m$ qudits, and we have the relation:
\begin{align}\label{weak-err-nR}
n_{\rm R}=2m(n_{\rm e}+1).
\end{align}
On each group we construct a (highly coupled) reference frame state $\psi$ of the form (\ref{input-state-form}). 
Then, at least one of the reference frame states will survive the erasure, and we can measure it to obtain the embedded rotation. In the following, we evaluate this (worst) case. 
Now we fix the form of the reference frame state by specifying the distribution $\{q_\Vec\lambda\}$ in Eq.\ (\ref{input-state-form}). We first specify $\settt{S}_{\rm via}\subset\settt{Y}_m$ on which $q_\Vec\lambda>0$.
To this purpose, we first define a parameter $M$ that depends on $m$ as
\begin{align}\label{M}
M=\left\lfloor\frac{1}{3}\left(\frac{2m}{d(d-1)}-1\right)\right\rfloor
\end{align}
and $m_0:=m-d(d-1)(3M+1)/2$. 
\begin{widetext}
Define $\tilde{\Vec\mu}\in\settt{Y}_{m_0}$ as the following Young diagram with $m_0$ boxes:
\begin{align}\label{flat-Young}
\tilde{\Vec\mu}:=(\tilde{\mu}_{1},\dots,\tilde{\mu}_{d})\quad{\rm s.t.}\ \sum_i|\tilde{\mu}_{i}|=m_0\quad{\rm and}\quad\tilde{\mu}_{j}+1\ge\mu_{i}\ge\tilde{\mu}_{j}\quad\forall\,j>i.
\end{align}
Now we define the following viable subset of Young diagrams with $d$ rows and $m$ boxes, on which our probe state has support:  
\begin{align}\label{viable-Young}
\settt{S}_{\rm via}:=\left\{\Vec\lambda\in\settt{Y}_m~:~\exists\,\tilde{\Vec\lambda}\in[M]^{\times(d-1)}\,{\rm s.t.}\,\lambda_i=\tilde{\mu}_{i}+(2d-i-2)M+d-i+\tilde{\lambda}_i,\, i=1,\dots,d-1\right\}.
\end{align} 
\end{widetext}
The reference frame state we use in the weak error model is defined by
\begin{align}\label{def-q}
q_{\tilde{\Vec\lambda}}:=\prod_{i=1}^{d-1}g_{\tilde{\lambda}_i},
\end{align}
where $g$ is the following distribution over $[M]$:
\begin{align}\label{g}
g_{\tilde{\lambda}_i}:=\frac{2}{M+1}\sin^2\left(\frac{\pi(2\tilde{\lambda}_i+1)}{2(M+1)}\right).
\end{align}
A similar construction has recently been used to achieve optimal programming of unitary gates \cite{Yuxiang2}.

\subsection{Bounding the error}

In the current error model, the code $(\map{E},\map{D})$ can be made exact (i.e.\null{} error-free).
Therefore, to use Lemma \ref{lemma-cov-error-bound}, we simply need to bound $\sum_{\Vec\lambda\in\settt{S}_{\rm int}}\sqrt{q_\Vec\lambda q_{\Vec\lambda+\Vec\Delta}}$ and appeal to Lemma~\ref{lemma-Fwc}.
In the following, we show that our choice of $\{q_\Vec\lambda\}$ [see Eq.\ (\ref{def-q})] satisfies
\begin{align}\label{epsilon}
\min_{\Vec\Delta\in\settt{S}_{\rm diff}}\sum_{\Vec\lambda\in\settt{S}_{\rm int}}\sqrt{q_\Vec\lambda q_{\Vec\lambda+\Vec\Delta}}\ge 1-\epsilon
\end{align}
for $\epsilon=\frac{d}{2}\left(\frac{\pi n'}{M+1}\right)^2+O(M^{-3})$. In this case, we can define $\settt{S}_{\rm int}$ as given by Eq.\ (\ref{viable-Young}) with the additional constraint that $2n'\le \tilde{\lambda}_i\le M-2n'$ for every $i$.
First, denote by $\epsilon_g$ the quantity 
\begin{align}\label{scaling-epsilon-g}
\epsilon_g(\delta):=1-\sum_{k=2n'}^{M-2n'}\sqrt{g_k g_{k+\delta}}\le\frac12\left(\frac{\pi\delta}{M+1}\right)^2+O(M^{-3}).
\end{align}
The inequality can be shown by straightforward calculation (see Subsection \ref{app-proof-Fwc-weak-error}).
For this distribution, it is straightforward that
\begin{align}
\sum_{\Vec\lambda\in\settt{S}_{\rm int}}\sqrt{q_\Vec\lambda q_{\Vec\lambda+\Vec\Delta}}\ge 1-\frac{d}{2}\left(\frac{\pi n'}{M+1}\right)^2-O(M^{-3}).
\end{align}
for any $\Vec\Delta\in\settt{S}_{\rm diff}$. 
Summarising, we reach the bound 
\begin{align}
&F_{\rm wc}\left(\int\d U~p(U|I)~\map{U}_\Cspace\otimes\map{U}^\ast_{\rm L},\map{I}_\Cspace\otimes\map{I}_{{\rm L}}\right)\nonumber\\
\ge& 1-\frac{d}{2}\left(\frac{\pi n'}{M+1}\right)^2-O(M^{-3}).\label{Fwc-weak-error}
\end{align}
Substituting Eqs.\ (\ref{weak-err-nR}), (\ref{M}), and $n'=n_\Cspace+d-1$ into the above bound, we get:
\begin{align}
F_{\rm wc}\ge 1-\frac{9\pi^2d^3(d-1)^2(n_{\rm e}+1)^2(n_\Cspace+d-1)^2}{2n_{\rm R}^2}-O(n_{\rm R}^{-3}).
\end{align}
Applying Lemma \ref{lemma-cov-error-bound} and recalling that the code error is always zero, we obtain the performance of our protocol:
\begin{theo}[Heisenberg-limited covariant error correction]\label{thm-weak-err}
For the weak error model, defined by Eq.\ (\ref{weak-error-model}), the diamond norm error of Protocol \ref{protocol-cov} is upper bounded as 
\begin{align}
\epsilon_{\cov}\le \frac{81\pi^2d^4(d-1)^2(n_{\rm e}+1)^2(n_\Cspace+d-1)^2}{2n_{\rm R}^2}+O(n_{\rm R}^{-3}).
\end{align}
The reference frame register in Protocol \ref{protocol-cov} should be initiated in the state $\psi^{\otimes (n_{\rm e}+1)}$, where the state $\psi$ is prepared in the form (\ref{input-state-form}) with coefficients given by Eq.\ (\ref{def-q}).
\end{theo}
Since $n_{\rm R}=n-n_\Cspace$ and $n_\Cspace$ can be chosen to be $O(1)$, our protocol achieves the Heisenberg limit $1/n^2$ with respect to the total number of qudit systems.

\subsection{Optimality of Protocol \ref{protocol-cov} under the weak error model}

Here we prove the optimality of our protocol under  the weak error model (\ref{weak-error-model}). In particular, we consider any code  constructed on $n$ qudit systems, denoted as $(\map{E}_\cov,\map{D}_\cov)$, that is covariant under the $\grp{SU}(d)$ action.
For the weak error model, we show that the Heisenberg limit $1/n^2$ is the ultimate limit for any covariant code, when each qudit has an equal probability $1/n$ of being erased. 
We stress the full generality of this result, in the sense that it does not assume any specific structure of the code.

The optimality can be shown by applying the lower bound on $\epsilon_\cov$, as given in Lemma \ref{lemma-converse0}, to the weak error model (\ref{weak-error-model}). 
We focus on the case when either exactly $n_{\rm e}$ qudits are erased or no qudit is erased at all, i.e., $p_{\settt{s}}=0$ for $0<|\settt{s}|<n_{\rm e}$.
Notice that Ref.\ \cite{kubica2020using}, where the original lower bound was derived, considered only independent local errors, which is not the case here.
However, as long as we can show $I_{\rm Fisher}^\uparrow<\infty$, the bound will work for our model. 
In the following we identify a Kraus form as in Eq.\ (\ref{I-Fisher}).
Consider the following Kraus form of $\left(\map{C}_{\rm e}\right)_{\settt{s}}$, the erasure of qudits with labels in $\settt{s}\subset[n]$, that depends on $\theta$:
\begin{align*}
\left(\map{C}_{\rm e}\right)_{\settt{s}}(\cdot)&=\sum_{\vec{n}}C_{\settt{s},\vec{n};\theta}(\cdot)C_{\settt{s},\vec{n};\theta}^\dag\\
C_{\vec{n},\settt{s};\theta}&:=\prod_{j=1}^{n_{\rm e}}\left(\exp\left\{\frac{i\theta h_{n_j} }{{n-1\choose n_{\rm e}-1}p_{\settt{s}}}\right\}|d+1\>\<j|_{\settt{s}_j}\right).
\end{align*}
Here, for convenience, we add a state $|d+1\>$ as the state after erasure, $\vec{n}=(n_1,\dots,n_{n_{\rm e}})\in(\Z_{d+1})^{\otimes n_{\rm e}}$, $H=\sum_{j=1}^{d+1}h_j|j\>\<j|$ with $h_{d+1}:=0$, $A_l:= I_{1}\otimes\cdots\otimes I_{l-1}\otimes A\otimes I_{l+1}\otimes\cdots\otimes I_n$ for any operator $A$, and $\settt{s}_j$ refers to the $j$-th largest element of $\settt{s}$.

The entire channel is $\map{C}_\theta:=\map{C}\circ\map{U}_\theta^{\otimes n}$ with $\map{C}$ defined by Eq.\ (\ref{weak-error-model}), which has Kraus operators
\begin{align*}
K_{0;\theta}&=\sqrt{p_0}\bigotimes_{l=1}^n U_{l;\theta}\\
K_{\vec{n},\settt{s};\theta}&=\sqrt{p_{\settt{s}}}\prod_{j=1}^{n_{\rm e}}\left(\exp\left\{\frac{i\theta h_{n_j} }{{n-1 \choose n_{\rm e}-1}p_{\settt{s}}}\right\}|d+1\>\<n_j|_{\settt{s}_j}\right)\left(\bigotimes_{l=1}^n U_{l;\theta}\right)\\
&\settt{s}\subset[n],\ |\settt{s}|=n_{\rm e}\quad \vec{n}\in(\Z_{d+1})^{\otimes n_{\rm e}}.
\end{align*}
\begin{widetext}
 Their derivatives are
\begin{align*}
\dot{K}_{0;\theta}&=i\sqrt{p_0}\left(\bigotimes_{l} U_{l;\theta}\right)\left(-\sum_{l}H_{l}\right)\\
\dot{K}_{\vec{n},\settt{s};\theta}&=i\sqrt{p_{\settt{s}}}\prod_{j=1}^{n_{\rm e}}\left(\exp\left\{\frac{i\theta h_{n_j} }{{n-1\choose n_{\rm e}-1}p_{\settt{s}}}\right\}|d+1\>\<n_j|_{\settt{s}_j}\right)\left(\bigotimes_{l}U_{l;\theta}\right)\left(\sum_{j=1}^{n_{\rm e}} \frac{h_{n_j} }{{n-1\choose n_{\rm e}-1}p_{\settt{s}}}|n_j\>\<n_j|_{\settt{s}_j}-\sum_{l}H_{l}\right).
\end{align*}
where $l=1,\dots,n$.
One can verify that
\begin{align}
&\dot{K}_{0;\theta}^\dag K_{0;\theta}+\sum_{\settt{s}}\sum_{\vec{n}}\dot{K}_{\vec{n},\settt{s};\theta}^\dag K_{\vec{n},\settt{s};\theta}\nonumber\\
=&ip_0\left(\sum_{l}H_{l}\right)-i\sum_{\settt{s}} p_{\settt{s}}\left(\sum_{\vec{n}}\left(\sum_j\frac{ h_{n_j} }{{n-1\choose n_{\rm e}-1}p_{\settt{s}}}|n_j\>\<n_j|_{\settt{s}_j}-\sum_{l}H_{l}\right)\left(\prod_{j}|n_j\>\<n_j|_{\settt{s}_j}\right)\right)
\nonumber\\
=&ip_0\left(\sum_{l}H_{l}\right)-i\left(\sum_{\settt{s}} p_{\settt{s}}\sum_j \frac{ H_{\settt{s}_j} }{{n-1\choose n_{\rm e}-1}p_{\settt{s}}}-(1-p_0)\sum_{l}H_{l}\right)\nonumber\\
=&0.
\end{align}
In the meantime
\begin{align}
\dot{K}_{0;\theta}^\dag \dot{K}_{0;\theta}+\sum_{\settt{s}}\sum_{\vec{n}}\dot{K}_{\vec{n},\settt{s};\theta}^\dag \dot{K}_{\vec{n},\settt{s};\theta}
&=\sum_{\settt{s}}\sum_{j=1}^{n_{\rm e}}\frac{\left(H^2\right)_{\settt{s}_j}}{{n-1\choose n_{\rm e}-1}^2p_{\settt{s}}}-\left(\sum_{l}H_{l}\right)^2.
\end{align}
\end{widetext}
Therefore, the Fisher information bound satisfies
\begin{align}
I_{\rm Fisher}^\uparrow&\le4\left\|\dot{K}_{0;\theta}^\dag \dot{K}_{0;\theta}+\sum_{\settt{s}}\sum_{\vec{n}}\dot{K}_{\vec{n},\settt{s};\theta}^\dag \dot{K}_{\vec{n},\settt{s};\theta}\right\|_{\infty}\nonumber\\
&\le 4\sum_{\settt{s},j}\frac{1}{{n-1\choose n_{\rm e}-1}^2p_{\settt{s}}}\|H^2\|_{\infty}+4n^2\|H\|_{\infty}^2\nonumber\\
&= 4\sum_{\settt{s}}\frac{n_{\rm e}}{{n-1\choose n_{\rm e}-1}^2p_{\settt{s}}}\|H^2\|_{\infty}+4n^2\|H\|_{\infty}^2.
\end{align}

Substituting into Eq.\ (\ref{diamond-converse}), we get
\begin{align}
\epsilon_{\wc}\ge \frac{(\Delta H)^2}{64(\sum_{\settt{s}}n_{\rm e}\|H^2\|_{\infty}/(p_{\settt{s}}{n-1\choose n_{\rm e}-1}^2)+n^2\|H\|_{\infty}^2)}.
\end{align}

Now, we assume that each qudit has an equal probability $p_{\settt{s}}=1/{n\choose n_{\rm e}}$ of being erased and $p_0=0$. Since $\|H^2\|_{\infty}=\|H\|_{\infty}^2$, the above bound implies 
\begin{align}
\epsilon_{\cov}\ge \frac{(\Delta H)^2}{64n^2(1+1/n_{\rm e})\|H\|_{\infty}^2}.
\end{align}
We can choose the minimum eigenvalue and the maximum eigenvalue of $H$ to sum up to zero. Then, $\Delta H=2\|H\|_{\infty}$, and we reach the following proposition:
\begin{prop}\label{prop-converse}
For the weak erasure error model (\ref{weak-error-model}) with $p_{\settt{s}}=1/{n\choose n_{\rm e}}$ for any $\settt{s}:|\settt{s}|=n_{\rm e}$,  
the error of any covariant code is lower bounded as
\begin{align}
\epsilon_{\cov}\ge \frac{1}{16n^2(1+1/n_{\rm e})}.
\end{align} 
\end{prop}

Since the above bound matches the performance of our protocol in scaling, we conclude that the optimal error scaling of covariant codes is identified as $1/n^2$.
 
\subsection{Proof of Eq.\ (\ref{scaling-epsilon-g})}\label{app-proof-Fwc-weak-error}
Invoking the definition of $\{g_k\}$ from Eq.\ (\ref{def-q}), we have the following chain of (in)equalities:	
\begin{widetext}
\begin{align*}
\sum_{k=n}^{M-n}\sqrt{g_kg_{k+\delta}}&=\frac{2}{M+1}\sum_{k=n}^{M-n}\sin\left(\frac{\pi(2k+1)}{2(M+1)}\right)\sin\left(\frac{\pi(2k+2\delta+1)}{2(M+1)}\right)\\
&=\frac{1}{M+1}\sum_{k=n}^{M-n}\left(\cos\left(\frac{\pi\delta}{M+1}\right)-\cos\left(\frac{\pi(2k+\delta+1)}{M+1}\right)\right)\\
&=\frac{1}{M+1}\left((M-2n+1)\cos\left(\frac{\pi\delta}{M+1}\right)-\sum_{k=n}^{M-n}\cos\left(\frac{\pi(2k+\delta+1)}{M+1}\right)\right)\\
&=\frac{1}{M+1}\left(\left(M-2n+1+\frac{\sin\left(\frac{2\pi n}{M+1}\right)}{\sin\left(\frac{\pi}{M+1}\right)}\right)\cos\left(\frac{\pi\delta}{M+1}\right)\right)\\
&=\frac{\cos\left(\frac{\pi\delta}{M+1}\right)}{M+1}\left(M-2n+1+\cos\left(\frac{\pi}{M+1}\right)+\sum_{k=1}^{2n-1}\cos\left(\frac{\pi\,k}{M+1}\right)\right)\\
&\ge \frac{1}{M+1}\left(M+1-\frac12\left(\frac{\pi}{M+1}\right)^2-\sum_{k=1}^{2n-1}\frac12\left(\frac{\pi\,k}{M+1}\right)^2\right)\left(1-\frac12\left(\frac{\pi\delta}{M+1}\right)^2\right)\\
&=1-\frac12\left(\frac{\pi\delta}{M+1}\right)^2-O\left(M^{-3}\right).
\end{align*}
\end{widetext}

\section{The i.i.d.\null{} error model}\label{sec-strong-error}

\subsection{Setting}
In this section, we deal with another type of erasure errors which are stronger than the one considered in the previous section. 
The error affects each qudit independently, erasing it with a probability:
\begin{align}\label{strong-error-model 4}
\map{C}=\bigotimes_{j=1}^{n}\left((1-p_{\rm e})\map{I}+p_{\rm e}\map{C}_{\rm e}\right)\qquad    p_{\rm e}\in\left(0,\frac12\right),
\end{align}
where $\map{C}_{\rm e}$ denotes the single-qudit erasure channel. Since the erasure channel is degradable, if $p_{\rm e}\ge1/2$ the information leaked to the environment would not be retrievable. Note that, in general, the error model on each qudit does not have to be identical, and each qudit $j$ can have distinct probability $p_{{\rm e},j}$ of being erased. In that case, however, we can simply set $p_{\rm e}$ to be the worst case over $\{p_{{\rm e},j}\}$ and consider this more stringent model instead. We can cast (\ref{strong-error-model 4}) in the form of Eq.\ (\ref{general-form-error}) as
\begin{align}
\map{C}=\sum_{\settt{s}\subset[n]}p_{{\rm e},\settt{s}}\left(\bigotimes_{k\in\settt{s}}\map{C}_{\rm e}\right)\otimes\left(\bigotimes_{k'\in\settt{s}^c}\map{I}\right)
\end{align}
where $p_{{\rm e},\settt{s}}:=p_{\rm e}^{|\settt{s}|}(1-p_{\rm e})^{n-|\settt{s}|}$.

In contrast to the previous model, there does not exist any code that perfectly corrects independent erasure errors. Instead, there exist pretty good codes that correct the error unless too many qudits are erased.
The quantum capacity of the erasure channel, with erasure probability $p_{\rm e}$, has been determined to be $1-2p_{\rm e}$ \cite{bennett1997capacities}. We can choose $n_\Cspace$, the number of computational qudits, to grow with $n$.
When $n$ is large, since the number of qudits we want to encode is only one qudit and is much smaller than that allowed by the capacity (which is $(1-2p_{\rm e})n_\Cspace$), the error probability would vanish exponentially in $n_\Cspace$.  

We choose $n_\Cspace= n^\gamma$, where $\gamma\in(0,1)$ does not depend on $n$ and can be chosen to be very small. We use any code $(\map{E},\map{D})$ in Protocol \ref{protocol-cov} that encodes one qudit into $n^\gamma$ computational qudits, with the property that it has a decoding error 
$O\left(e^{-x_{d}\cdot n^\gamma}\right)$
for some $x_{d}>0$ that may depend on $d$. By a random coding argument, one can show that there exists a stabilizer code satisfying our requirement (see, e.g., \cite{gottesman1997stabilizer}), although its explicit form is not given. Recently progress in error correcting codes also showed that quantum polar codes~\cite{renes_efficient_2012,renes_polar_2014} and Reed-Muller codes~\cite{kumar_reed-muller_2016} have the desired property. Notice that the requirement of the $O\left(e^{-x_{d}\cdot n^\gamma}\right)$ scaling is chosen for convenience of analysing the error, and it can be further relaxed in practice.

Meanwhile, the model is now too noisy for the highly coupled reference frame state used in the weak error model to be effective. Instead, we prepare on the reference frame register the following $s_{\rm R}$-copy state:
\begin{align}
\Psi=(\Phi^+)^{\otimes s_{\rm R}}
\end{align}
with $s_{\rm R}=n_{\rm R}/2=(1-n^{-1+\gamma})n/2$ and $|\Phi^+\>$ being the maximally entangled state on $\spc{H}\otimes\spc{H}$. Intuitively, this choice is to distribute the eggs in different baskets. One erasure error destroys at most one of the $s_{\rm R}$ reference frame states. As long as there are still $\Theta(n)$ copies left we can achieve high performance, which happens with very high probability since $p_{{\rm e},\settt{s}}$ is a binomial distribution.

\subsection{Bounding the error}

When $s_{\rm e}$ reference frames are erased, the number of remaining reference frames is
\begin{align}
s:=s_{\rm R}-s_{\rm e}.
\end{align}
The remaining reference frames can be decomposed in the form (\ref{input-state-form}) as
\begin{align}
(\Phi^+)^{\otimes s}=\bigoplus_{\Vec\lambda\in\settt{Y}_{s}}\sqrt{p_{\Vec\lambda,s}}|\Phi_{\Vec\lambda}^+\>\otimes|\Phi^+_{m_\Vec\lambda}\>
\end{align}
where $p_{\Vec\lambda,s}$ is known as the Schur-Weyl distribution. It has the following precise form \cite[Eq.\ (3.28)]{alicki1988symmetry}:
\begin{align}\label{schur-weyl-dist}
p_{\Vec\lambda,s}=\frac{s!\prod_{i<j}(\tilde{\lambda}_i-\tilde{\lambda}_j)^2}{d^s(\tilde{\lambda}_d)!\prod_{k=1}^{d-1}k!(\tilde{\lambda}_k)!}\qquad\tilde{\lambda}_j:=\lambda_j+d-j.
\end{align}

We now combine it with Eq.\ (\ref{est-prob-dist-covariant}) and Lemma~\ref{lemma-wcinput} to express the fidelity (\ref{weak-err-Fwc}) as
\begin{align}\label{F-m-inter1}
F_{s}=\sum_{\Vec\lambda,\Vec\lambda'\in\settt{Y}_{s}}\sqrt{p_{\Vec\lambda,s} p_{\Vec\lambda',s}}S_{\Vec\lambda,\Vec\lambda'},
\end{align}
where $S_{\Vec\lambda,\Vec\lambda'}$ is the correlation function defined by Eq.\ (\ref{S-correlation}).  For any $\alpha>0$, if $s=\beta\cdot n$ for some fixed $\beta>0$, then there exists $n_{0,\alpha}$ so that 
\begin{align}\label{F_m}
1-F_{s}\le \left(\frac{(d^2-d+32)\Gamma}{\beta^2}\right)\cdot\left(\frac{1}{2s}\right)^{1-\alpha-2\gamma}
\end{align}
holds for arbitrary $n\ge n_{0,\alpha}$. Here $\Gamma:=d^{\frac{d(d-1)}{2}}/(\prod_{k=1}^{d-1}k!)$.

Eq.~(\ref{F_m}) implies that $1-F_{s}$ scales almost as $1/s$ for large $s$.
Its  proof is delayed to Subsection \ref{app-F_m}.
Roughly speaking, in the large $s$ limit, $p_{\Vec\lambda,s}$ converges to a tilted multi-variate Gaussian. We can focus on a region near the peak of $p_{\Vec\lambda,s}$ where the distribution is concentrated and where the correlation function $S_{\Vec\lambda,\Vec\lambda'}$ is maximised as well. This would introduce at most an error vanishing as $1/s$. In addition, since a Gaussian is relatively flat around its peak, we can show that  $\sqrt{p_{\Vec\lambda+\Vec\Delta,s} p_{\Vec\lambda,s}}\approx p_{\Vec\lambda+\Vec\Delta/2,s}$ up to an error that scales as $1/s$.  This would be the main reason of getting this error scaling.

With Eq.\ (\ref{F_m}) it is rather straightforward to bound the error of our protocol.  First we define a threshold value
\begin{align}\label{m-ast}
s^\ast:=(1-2p_{\rm e})\cdot s_{\rm R}=(1/2-p_{\rm e})(n-n^\gamma).
\end{align}
This definition ensures that the probability that $s<s^\ast$ goes to zero exponentially fast.
Next, notice that the error probabilities on the computational register and the reference frame register are independent:
\begin{align}\label{indep-prob}
p_{\rm e}(\settt{s})=p_{{\rm e, P}}(\settt{s}_\Cspace)p_{{\rm e, R}}(\settt{s}_{\rm R})\qquad \settt{s}=\settt{s}_\Cspace\cup\settt{s}_{\rm R},
\end{align}
where $p_{{\rm e, P}}$ and $p_{{\rm e, R}}$ are the probability distributions of erasure errors in the computational register and in the reference frame register, respectively.
\begin{widetext}
Applying Lemma \ref{lemma-cov-error-bound} and Eq.\ (\ref{indep-prob}), we can split the error of our protocol as
\begin{align}
\epsilon_\cov&\le9d\left(\sum_{\settt{s}\subset[n],F_{s}\ge\frac34}p_{\rm e}(\settt{s})\max\left\{\epsilon_{\settt{s}_\Cspace,{\rm code}},1-F_{s}\right\}+\sum_{\settt{s}\subset[n],F_{s}<\frac34}p_{\rm e}(\settt{s})\right).
\end{align}
From Eqs.\ (\ref{F_m}) and (\ref{m-ast}), we know that, for large enough $n$, $F_{s}>\frac34$ for $s\ge s^\ast$. That is, when no more than $2p_{\rm e}s_{\rm R}$ erasure errors occur on the reference frame register. We can then express the bound as
\begin{align}
\epsilon_\cov&\le 9d\left(\sum_{\settt{s}_\Cspace}p_{{\rm e, P}}(\settt{s}_\Cspace)\epsilon_{\settt{s}_\Cspace,{\rm code}}+\sum_{\settt{s}_{\rm R}}p_{{\rm e, R}}(\settt{s}_{\rm R})\left(1-F_{s}\right)+\sum_{\settt{s}_{\rm R}:|\settt{s}_{\rm R}|>2p_{\rm e}s_{\rm R}}p_{{\rm e, R}}(\settt{s}_{\rm R})\right).
\end{align}
Notice that, by our assumption on $(\map{E},\map{D})$, we should have 
\begin{align}\label{err-code-assump}
\sum_{\settt{s}_\Cspace}p_{{\rm e, P}}(\settt{s}_\Cspace)\epsilon_{\settt{s}_\Cspace,{\rm code}}=O\left(e^{- x_{d}\cdot n^\gamma}\right). 
\end{align}
Therefore, the reference frame error constitutes the main contribution to the overall error:
\begin{align}
\epsilon_\cov&\le 9d\left(\sum_{\settt{s}_{\rm R}}p_{{\rm e, R}}(\settt{s}_{\rm R})\left(1-F_{s}\right)+O\left(e^{-x_{d}\cdot n^\gamma}\right)+\sum_{\settt{s}_{\rm R}:|\settt{s}_{\rm R}|>2p_{\rm e}s_{\rm R}}p_{{\rm e, R}}(\settt{s}_{\rm R})\right)\qquad s=\frac{n_{\rm R}}{2}-|\settt{s}_{\rm R}|\nonumber\\
&\le 9d\left(1-F_{s^\ast}+2\sum_{\settt{s}_{\rm R}:|\settt{s}_{\rm R}|>2p_{\rm e}s_{\rm R}}p_{{\rm e, R}}(\settt{s}_{\rm R})\right)+O\left(e^{-x_{d}\cdot n^\gamma}\right)\nonumber\\
&\le  \left(\frac{36d(d^2-d+32)\Gamma}{(1-2p_{\rm e})^2}\right)\cdot\left(\frac{1}{(1-2p_{\rm e})n}\right)^{1-\alpha-2\gamma}+O\left(e^{-p^2_{\rm e}(1-\gamma)n}\right)+O\left(e^{-x_{d}\cdot n^\gamma}\right).
\end{align}
\end{widetext}
The second inequality comes from dividing the summation into the term for $s\ge s^\ast$ and the term for $s<s^\ast$. The third inequality comes from Eq.\ (\ref{F_m}) (with $\beta=1/2-p_{\rm e}$) and the Hoeffding bound. From the above bound, we can see that $\gamma$ can be chosen to be arbitrarily close to zero.
Absorbing the other error terms into the major term, we have the following bound on the error:
\begin{theo}\label{theo-strong-error}
Consider the i.i.d.\null{} error model defined by Eq.\ (\ref{strong-error-model 4}). For any $\alpha>0$, there exists $n_{\alpha}>0$ so that the diamond norm error of Protocol \ref{protocol-cov} is upper bounded by
\begin{align}
\epsilon_{\cov}\le \left(\frac{36(d^2-d+32)d^{\frac{d^2-d+2}{2}}}{(1-2p_{\rm e})^2\prod_{j=1}^{d}(j-1)!}\right)\left(\frac{1}{(1-2p_{\rm e})n}\right)^{1-\alpha}
\end{align}
for any $n\ge n_\alpha$. The reference frame register in Protocol \ref{protocol-cov} should be initiated in multiple copies of the maximally entangled qudit state.
\end{theo}

The error of our protocol scales  almost as $1/n$, instead of $1/n^2$ in the previous case. This is a result of the (stronger) i.i.d.\null{} noise. In fact, as we show in the next subsection, the error scaling of our protocol is still optimal for this error model.

\subsection{Optimality of Protocol \ref{protocol-cov} under the i.i.d.\null{} error model}

Here we prove the optimality of our protocol under the i.i.d.\null{} error model (\ref{strong-error-model 4}). In particular, we consider any code  constructed on $n$ qudit systems, denoted as $(\map{E}_\cov,\map{D}_\cov)$, that is covariant under the $\grp{SU}(d)$ action.
For the i.i.d.\null{} error model, we show that the $1/n$ scaling is optimal.

Under the i.i.d.\null{} error model (\ref{strong-error-model 4}), the optimality of our protocol can, again, be shown by applying Lemma \ref{lemma-converse0}. Here the Fisher information upper bound $I_{\rm Fisher}^\uparrow$ should be replaced by the one for the local, independent erasure model. For the i.i.d.\null{} (erasure) error model, the Fisher information upper bound has already been given by \cite[Eq.\ (B19)]{kubica2020using} as
\begin{align}
I_{\rm Fisher}^\uparrow=4n(\Delta H)^2\left(\frac{1-p_{\rm e}}{p_{\rm e}}\right).
\end{align}
Substituting into Lemma \ref{lemma-converse0}, we obtain:
\begin{prop}\label{prop-converse-strong-error}
For the i.i.d.\null{} (erasure) error model, defined by Eq.\ (\ref{strong-error-model 4}),
the error of any covariant code is lower bounded by  
\begin{align}
\epsilon_{\cov}\ge \frac{p_{\rm e}}{64n(1-p_{\rm e})}.
\end{align} 
\end{prop}
Therefore, since Protocol \ref{protocol-cov} achieves the $(1/n)$-scaling, it is optimal in the asymptotic limit of large $n$.

\subsection{Proof of Eq.\ (\ref{F_m})}\label{app-F_m}

Here we prove Eq.~(\ref{F_m}) under the assumption that $s=\beta\cdot n$ for some fixed $\beta>0$.
Since $S_{\Vec\lambda,\Vec\lambda'}\ge0$ we can lower bound the fidelity (\ref{F-m-inter1}) as
\begin{align}
F_{s}\ge\sum_{\Vec\lambda,\Vec\lambda'\in\settt{S}_{{\rm cent}}}\sqrt{p_{\Vec\lambda,s} p_{\Vec\lambda',s}}S_{\Vec\lambda,\Vec\lambda'},
\end{align}
where $\settt{S}_{{\rm cent}}\subset\settt{Y}_{s}$ is defined as
\begin{align}\label{app-S-cent}
\settt{S}_{{\rm cent}}&:=\left\{\Vec\lambda\in\settt{Y}_{s}~:~\left|\lambda_i-\frac{s}{d}\right|\le \frac{s^{\frac{1+x}{2}}}{2}, |\lambda_i-\lambda_j|>3n',\forall\,i,j\right\}\\
&n'=n^\gamma+d-1.
\end{align}
Here $x>0$ is a parameter to be specified at the end of the proof.
We can express the fidelity as
\begin{align}
F_{s}\ge\sum_{\Vec\lambda\in\settt{S}_{\rm cent}}\sum_{\Vec\Delta\in\settt{S}_{\rm diff}}\sqrt{p_{\Vec\lambda,s} p_{\Vec\lambda+\Vec\Delta,s}}S_{\Vec\lambda,\Vec\lambda+\Vec\Delta}.
\end{align}
where 
\begin{align}\label{app-S-Delta}
\settt{S}_{\rm diff}=\{\Vec\Delta\in\Z^{\times d}~:~|\Delta_i|\le n', i=1,\dots,d; \sum_{j=1}^d\Delta_j=0\}.
\end{align}
Similar as in the proof of Lemma \ref{lemma-Fwc}, we first show that:
For $\Vec\lambda\in\settt{S}_{{\rm cent}}$ and $\Vec\Delta\in\settt{S}_{\rm diff}$, the correlation function depends only on their relative distance $\Vec\Delta$, but not explicitly on $\Vec\lambda$, i.e., $S_{\Vec\lambda,\Vec\lambda+\Vec\Delta}=\tilde{S}_{\Vec\Delta}$.
Indeed, when $\Vec\lambda$ is in the viable set $S_{{\rm cent}}$, the lengths of different rows of both $\Vec\lambda$ and $\Vec\lambda+\Vec\Delta$  have big enough gaps  so that adding $n'$ boxes to one row would not make its box number greater than its preceding rows. Therefore, adding $\Vec\mu'$ and $\Vec\mu$ in $\settt{S}_{\rm cost}$ according to the Littlewood-Richardson rule is not constraint by the shape of $\Vec\lambda$.
Now, the fidelity bound becomes
\begin{align}
F_{s}&\ge\sum_{\Vec\lambda\in\settt{S}_{\rm cent}}\sum_{\Vec\Delta\in\settt{S}_{\rm diff}}\sqrt{p_{\Vec\lambda,s} p_{\Vec\lambda+\Vec\Delta,s}}\tilde{S}_{\Vec\Delta}\\
&\ge \left(\min_{\Vec\Delta\in\settt{S}_{\rm diff}}\sum_{\Vec\lambda\in\settt{S}_{\rm cent}}\sqrt{p_{\Vec\lambda,s} p_{\Vec\lambda+\Vec\Delta,s}}\right)\left(\sum_{\Vec\Delta'\in\settt{S}_{\rm diff}}\tilde{S}_{\Vec\Delta'}\right)\\
&\ge \min_{\Vec\Delta\in\settt{S}_{\rm diff}}\sum_{\Vec\lambda\in\settt{S}_{\rm cent}}\sqrt{p_{\Vec\lambda,s}p_{\Vec\lambda+\Vec\Delta,s}},
\end{align}
having used Eq.\ (\ref{sum-SDelta}) in the last step.

What remains is to bound $\sum_{\Vec\lambda\in\settt{S}_{\rm cent}}\sqrt{p_{\Vec\lambda,s} p_{\Vec\lambda+\Vec\Delta,s}}$ for every $\Vec\Delta\in\settt{S}_{\rm diff}$, where $p_{\Vec\lambda,s}$ is the Schur-Weyl distribution:
\begin{align}\label{Schur-weyl}
p_{\Vec\lambda,s}=\frac{s!\prod_{i<j}(\tilde{\lambda}_i-\tilde{\lambda}_j)^2}{d^s(\tilde{\lambda}_d)!\prod_{k=1}^{d-1}k!(\tilde{\lambda}_k)!}\qquad\tilde{\lambda}_j:=\lambda_j+d-j.
\end{align}
First, we show that for the Schur-Weyl distribution $p_{\Vec\lambda,s}$, $\sqrt{p_{\Vec\lambda+\Vec\Delta,s}p_{\Vec\lambda,s}}\approx p_{\Vec\lambda+\Vec\Delta/2,s}$ for $\Vec\lambda\in\settt{S}_{\rm cent}$.
Define the following multinomial distribution of $\tilde{\Vec\lambda}$:
\begin{align}\label{multinomial}
b_{\tilde{\Vec\lambda}}:=\frac{(s+d(d-1)/2)!}{(\tilde{\lambda}_1)!\cdots(\tilde{\lambda}_d)!}d^{-(s+d(d-1)/2)},
\end{align} 
and we can bound the Schur-Weyl distribution as
\begin{align}
p_{\Vec\lambda,s}&=b_{\tilde{\Vec\lambda}}\,\frac{\Gamma}{[s+d(d-1)/2]_{d(d-1)/2}}\prod_{i<j}\left(\tilde{\lambda}_i-\tilde{\lambda}_j\right)^2\\
&\le b_{\tilde{\Vec\lambda}}\,\Gamma\prod_{i<j}\left(\frac{\tilde{\lambda}_i-\tilde{\lambda}_j}{\sqrt{s}}\right)^2
\label{Schur-weyl-alt},
\end{align}
where $[l]_k:=l!/(l-k)!$ and $\Gamma:=d^{\frac{d(d-1)}{2}}/(\prod_{k=1}^{d-1}k!)$.

\begin{widetext}
We have
\begin{align}\label{Schur-weyl1}
\frac{\sqrt{p_{\Vec\lambda+\Vec\Delta,s}p_{\Vec\lambda,s}}}{p_{\Vec\lambda+\Vec\Delta/2,s}}=\frac{\sqrt{b_{\tilde{\Vec\lambda}+\Vec\Delta}b_{\tilde{\Vec\lambda}}}}{b_{\tilde{\Vec\lambda}+\Vec\Delta/2}}\prod_{i<j}\left(1-\left(\frac{\Delta_i-\Delta_j}{2\tilde{\lambda}_i-2\tilde{\lambda}_j+\Delta_i-\Delta_j}\right)^2\right).
\end{align}

Using Stirling's approximation $n!=\sqrt{2\pi n}(n/e)^n(1+1/(12n)+O(n^2))$ and $\sum_{i=1}^{d}\Delta_i=0$, we have
\begin{align}
\frac{\sqrt{b_{\tilde{\Vec\lambda}+\Vec\Delta}b_{\tilde{\Vec\lambda}}}}{b_{\tilde{\Vec\lambda}+\Vec\Delta/2}}\le\exp\left\{-\frac{f_{\tilde{\Vec\lambda}+\Vec\Delta}+f_{\tilde{\Vec\lambda}}}{2}+f_{\tilde{\Vec\lambda}+\Vec\Delta/2}\right\}\left(1+\frac{d^2}{6 s}+O\left(s^{-2}\right)\right)
\end{align}
for $\Vec\lambda\in\settt{S}_{\rm cent}$, where $f_{\tilde{\Vec\lambda}}:=\sum_i (\tilde{\lambda}_i+1/2)\ln  \tilde{\lambda}_i$.
\end{widetext}
 \iffalse
\red{Intermediate step:
\begin{align}
\left(\tilde{\lambda}_i+\Delta_i+\frac12\right)\ln(\tilde{\lambda}_i+\Delta_i)&=\left(\tilde{\lambda}_i+\Delta_i+\frac12\right)\left(\ln\left(\tilde{\lambda}_i+\frac{\Delta_i}{2}\right)+\frac{\Delta_i}{2\tilde{\lambda}_i}+\frac{1}{2}\left(\frac{\tilde{\lambda}_i}{2\tilde{\lambda}_i}\right)^2+O\left(\left(\frac{\Delta_i}{\tilde{\lambda}_i}\right)^3\right)\right).
\end{align}
}
 \fi
Straightforward calculation shows that
\begin{align}
\frac{f_{\tilde{\Vec\lambda}}+f_{\tilde{\Vec\lambda}+\Vec\Delta}}{2}-f_{\tilde{\Vec\lambda}+\Vec\Delta/2}&=\sum_i \frac{(\Delta_i)^2}{8\tilde{\lambda}_i}+O\left(\frac{(n')^3}{s^2}\right).
\end{align}
Since $n'=O(n^\gamma)$ and $s=\Theta(n)$, we have
\begin{align}
\frac{f_{\tilde{\Vec\lambda}}+f_{\tilde{\Vec\lambda}+\Vec\Delta}}{2}-f_{\tilde{\Vec\lambda}+\Vec\Delta/2}&=\sum_i \frac{(\Delta_i)^2}{8\tilde{\lambda}_i}+O\left(n^{-2+3\gamma}\right),
\end{align}
which, plus $\sum_i(\Delta_i)^2\le (n')^2$ and $|\lambda_i- s/d|=O\left(s^{\frac{1+x}{2}}\right)$, implies that
\begin{align}
\sqrt{b_{\tilde{\Vec\lambda}+\Vec\Delta}b_{\tilde{\Vec\lambda}}}&\ge b_{\tilde{\Vec\lambda}+\Vec\Delta/2}\left(1-\frac{d(n')^2}{8s}-\frac{d^2}{6 s}-O\left(n^{-\frac{3-x-4\gamma}{2}}\right)\right)\\
&= b_{\tilde{\Vec\lambda}+\Vec\Delta/2}\left(1-O(n^{-1+2\gamma})-O\left(n^{-\frac{3-x-4\gamma}{2}}\right)\right).
\end{align}
Combining the above bound with Eqs.\ (\ref{app-S-Delta}) and (\ref{Schur-weyl1}), we get
\begin{widetext}
\begin{align}
\frac{\sqrt{p_{\Vec\lambda+\Vec\Delta,s}p_{\Vec\lambda,s}}}{p_{\Vec\lambda+\Vec\Delta/2,s}}&\ge \left(1-O(n^{-1+2\gamma})-O\left(n^{-\frac{3-x-4\gamma}{2}}\right)\right)\cdot\prod_{i<j}\left(1-\frac{(n')^2}{(\tilde{\lambda}_i-\tilde{\lambda}_j+\Delta_i/2-\Delta_j/2)^2}\right)\nonumber\\
&\ge 1-\sum_{i<j}\frac{(n')^2}{(\tilde{\lambda}_i-\tilde{\lambda}_j+\Delta_i/2-\Delta_j/2)^2}-O(n^{-1+2\gamma})-O\left(n^{-\frac{3-x-4\gamma}{2}}\right)
\label{strong-error-inter1}
\end{align}
and we have $\sqrt{p_{\Vec\lambda+\Vec\Delta,s}p_{\Vec\lambda,s}}\approx p_{\Vec\lambda+\Vec\Delta/2,s}$ with the approximation error given by Eq.~(\ref{strong-error-inter1}).

What remains is to evaluate the summation $\sum_{\Vec\lambda\in\settt{S}_{\rm cent}}p_{\Vec\lambda+\Vec\Delta/2,s}$.
Using the bounds (\ref{app-S-cent}) and (\ref{Schur-weyl-alt}), the error terms can be bounded as follows:
\begin{align}
\frac{\sum_{\Vec\lambda\in\settt{S}_{\rm cent}}p_{\Vec\lambda+\Vec\Delta/2,s}}{(\tilde{\lambda}_i-\tilde{\lambda}_j+\Delta_i/2-\Delta_j/2)^2}&\le\sum_{\Vec\lambda\in\settt{S}_{\rm cent}}\frac{\Gamma\,b_{\tilde{\Vec\lambda}+\Vec\Delta/2}}{s}\cdot\prod_{k<l,k\not=i,l\not=j}\left(\frac{\tilde{\lambda}_k-\tilde{\lambda}_l+\Delta_k/2-\Delta_l/2}{\sqrt{s}}\right)^2\nonumber\\
&\le \frac{\Gamma}{s}\cdot\prod_{k<l,k\not=i,l\not=j}\left(\frac{s^{\frac{1+x}{2}}+l-k+n'}{\sqrt{s}}\right)^2\sum_{\Vec\lambda\in\settt{S}_{\rm cent}}b_{\tilde{\Vec\lambda}+\Vec\Delta/2}\nonumber\\
&\le \frac{\Gamma}{s}\cdot\prod_{k<l,k\not=i,l\not=j}\left(\frac{s^{\frac{1+x}{2}}+l-k+n'}{\sqrt{s}}\right)^2\cdot 1\nonumber\\
\end{align}
\end{widetext}
Assuming $(\sqrt{2}-1)s^{\frac{1+x}{2}}\ge d-1+n'\ge l-k+n'$ for any $l,k$ (which always holds for large enough $n$), we can simplify the above bound to
\begin{align}
\frac{\sum_{\Vec\lambda\in\settt{S}_{\rm cent}}p_{\Vec\lambda+\Vec\Delta/2,s}}{(\tilde{\lambda}_i-\tilde{\lambda}_j+\Delta_i/2-\Delta_j/2)^2}&\le 2\Gamma \left(\frac{1}{2s}\right)^{1-x\cdot\frac{d^2-d-2}{2}}.\label{strong-error-inter2}
\end{align}
In the same manner, we can show that 
\begin{align}\label{strong-error-inter3}
\sum_{\Vec\lambda\in\settt{S}'_{\rm cent}\setminus\settt{S}_{\rm cent}}p_{\Vec\lambda+\Vec\Delta/2,s}\le 2\Gamma(4n'+d-1)^2\,\left(\frac{1}{2s}\right)^{1-x\cdot\frac{d^2-d-2}{2}}
\end{align}
where
\begin{align}\label{S-cent-prime}
\settt{S}'_{{\rm cent}}:=\left\{\Vec\lambda\in\settt{S}_{s}~:~\left|\lambda_i-\frac{s}{d}\right|\le \frac{s^{\frac{1+x}{2}}}{2},\ \forall\,i\right\}
\end{align}
is just $\settt{S}_{\rm cent}$ with the restriction $|\lambda_i-\lambda_j|>n'$ lifted. Indeed, since for $\Vec\lambda\in\settt{S}'_{\rm cent}\setminus\settt{S}_{\rm cent}$ there exists a pair $(i,j)$ such that $|\lambda_i-\lambda_j|\le 3n'$ and $|\lambda_k-\lambda_l|\le s^{\frac{1+x}2}$ for the rest, we have
\begin{align}
&\sum_{\Vec\lambda\in\settt{S}'_{\rm cent}\setminus\settt{S}_{\rm cent}}p_{\Vec\lambda+\Vec\Delta/2,s}\nonumber\\
\le&\sum_{\Vec\lambda\in\settt{S}'_{\rm cent}\setminus\settt{S}_{\rm cent}}\Gamma\,b_{\tilde{\Vec\lambda}+\Vec\Delta/2}\cdot\prod_{k<l}\left(\frac{\tilde{\lambda}_k-\tilde{\lambda}_l+\Delta_k/2-\Delta_l/2}{\sqrt{s}}\right)^2\nonumber\\
\le&\sum_{\Vec\lambda\in\settt{S}'_{\rm cent}\setminus\settt{S}_{\rm cent}}\Gamma\,b_{\tilde{\Vec\lambda}+\Delta/2}\cdot\left(\frac{3n'+d-1+n'}{\sqrt{s}}\right)^2\cdot\left(2s\right)^{x\cdot\frac{d^2-d-2}{2}} \nonumber\\
\le &2\Gamma\left(4n'+d-1\right)^2\cdot\left(\frac{1}{2s}\right)^{1-x\cdot\frac{d^2-d-2}{2}}. \nonumber
\end{align}

The probability that $\Vec\lambda\not\in\settt{S}'_{\rm cent}$ vanishes exponentially due to large deviation theory. Precisely, the following bound holds \cite{christandl2006spectra,o2015quantum}:
\begin{align}
&\sum_{\Vec\lambda\in\settt{S}'_{\rm cent}}p_{\Vec\lambda+\Vec\Delta/2,s}\nonumber\\
\ge& 1-\max_{\Vec\lambda'-\Vec\Delta/2\not\in\settt{S}_{\rm cent}}(s+1)^{\frac{d(d-1)}{2}}\exp\left\{-\frac{2d_{\rm Young}(\Vec\lambda', \Vec{s/d})^2}{s}\right\}
\end{align}
where $\Vec{s/d}$ here refers to the Young diagram $(s/d,\dots,s/d)$. Substituting in Eq.\ (\ref{S-cent-prime}) and using $s=\Theta(n)$ we get
\begin{align}\label{strong-error-inter4}
\sum_{\Vec\lambda\in\settt{S}'_{\rm cent}}p_{\Vec\lambda+\Vec\Delta/2,s}\ge 1-O\left(e^{-n^x}\right).
\end{align}
Finally, summarising Eqs.\ (\ref{strong-error-inter1}), (\ref{strong-error-inter2}), (\ref{strong-error-inter3}), and (\ref{strong-error-inter4}), we get 
\begin{align}
&1-F_{s}\nonumber\\
\le&1-\min_{\Vec\Delta\in\settt{S}_{\rm diff}}\sum_{\Vec\lambda\in\settt{S}_{\rm cent}}\sqrt{p_{\Vec\lambda,s} p_{\Vec\lambda+\Vec\Delta,s}}\\
\le& \left(\frac{d(d-1)(n')^2}{2}+(4n'+d-1)^2\right)2\Gamma\cdot\left(\frac{1}{2s}\right)^{1-x\cdot\frac{d^2-d-2}{2}}\nonumber\\
&\qquad+O(n^{-1+2\gamma})+O\left(n^{-\frac{3-x-4\gamma}{2}}\right)+O\left(e^{-n^x}\right).
\end{align}
Picking any $x\in(0,\min\{2\alpha/(d^2-d-2),1\})$ we get 
\begin{align}
1-F_{s}&\le \Gamma\left(d^2-d+32\right)\cdot(n')^2\left(\frac{1}{2s}\right)^{1-\alpha}+O(n^{-1+2\gamma}).
\end{align}

Recalling that $n'= n^\gamma+d-1$ and $s=\beta\cdot n$,  we get Eq.\ (\ref{F_m}).
Since the above bound holds for any $x>0$, the scaling of $1-F_{s}$ approaches $s^{-1+2\gamma}$ in the asymptotic limit of large $s$.

\section{Resource requirements and implementation}

\subsection{Compression of quantum reference frames}\label{sec-dim-prob}
In both models, the reference frame states are constructed on $n_{\rm R}$ qudits, which is a $(d^{n_{\rm R}})$-dimensional system. However, here we show that the reference frame states can be compressed into a much smaller system. Indeed, the dimension of the state is given by the effective dimension 
\begin{align}
d_{\rm R}:=\dim\left(\mathsf{Span}\{\map{U}_{\rm R}(\Psi)\}_{U\in\grp{SU}(d)}\right),
\end{align}
which, as we show in the following, grows only polynomially in $n_{\rm R}$.

Since we always use entangled reference frames, $U_{\rm R}=(U\otimes I)^{\otimes\frac{n_{\rm R}}{2}}$.
By Schur-Weyl duality, we have the decomposition
\begin{align}
U_{\rm R}\simeq\bigoplus_{\Vec\lambda\in\settt{Y}_{n_{\rm R}/2}}(U_\Vec\lambda\otimes I_{m_\Vec\lambda})\otimes I^{\otimes\frac{n_{\rm R}}{2}}.
\end{align}
Therefore, $d_{\rm R}$, which is equal to the rank of the twirled state $\int\d U~\map{U}_{\rm R}(\Psi)$, satisfies the upper bound
\begin{align}
d_{\rm R}&\le \sum_{\Vec\lambda\in\settt{Y}_{n_{\rm R}/2}}d_\lambda^2\nonumber\\
&\le |\settt{Y}_{n_{\rm R}/2}|\max_{\Vec\lambda\in\settt{Y}_{n_{\rm R}/2}}d_\Vec\lambda^2\nonumber\\
&\le \left(\frac{n_{\rm R}}{2}+1\right)^{(d^2-1)}.
\end{align}
The last bound is quite straightforward to obtain; see, for instance, Ref.\ \cite[Eqs.\ (6.16) and (6.18)]{hayashi2017group}. Therefore, we can compress the reference frame state to a system of much smaller dimension, which reduces exponentially the cost of quantum memory during idle time. The cost can be further reduced at the price of a small recovery error using the compression protocols in Refs.\ \cite{yang2016efficient,yang2016optimal,yang2018compression}.

\subsection{Computational efficiency of implementation}\label{sec:Computational efficiency of implementation}
%\red{Maybe Mischa's note could fit here?}
At first sight, our encoding and decoding may appear to necessitate sampling from the Haar measure on $\grp{SU}(d)$ in order to be implemented [see Eqs.~\eqref{E-cov} and \eqref{D-cov}]. From a computational complexity standpoint, this would be a problem, since to implement a random Haar unitary one needs an exponential number of two-qubit gates and random bits \cite{Knill1995Haar}. Luckily, our encoding and decoding do not require sampling from the Haar measure, but rather only from a distribution which agrees with it upto the $t^\textup{th}$ moment, for appropriate $t$. Such equivalent distributions are known as unitary $t$-designs~\cite{PhysRevLett.98.130502,Dahlsten_2007,PhysRevA.80.012304,CMPHarrow,Diniz2011,PhysRevA.78.062329,CMPBrandao,PhysRevX.7.021006}, and are more efficient to implement. 

Specifically, let $U$ be a unitary representation on $\C^{d\times d}$ and $P_{t,t}(U)$ be a matrix whose entries are polynomials of order $t$ in the coefficients of $U$ and of order $t$ in the coefficients of $U^*$, and let $\mathbb{E}_{U\sim \nu}[f(U)]$ be the expectation value of a function $f$ according to measure $\nu$. We then say that $ \mathbb{E}_{U\sim \nu_\textup{Haar}}[P_{t,t}(U)]$, where $\nu_\textup{Haar}$ denotes the Haar measure, \emph{admits a unitary $t$-design}.

%Can delete this: For the following proposition, recall the notation $\map{U}_{\rm L}(\cdot) = {U}_{\rm L} (\cdot) {U}_{\rm L}^\dag$, $\map{U}_\Cspace(\cdot) = {U}_\Cspace^{\otimes n_\Cspace} (\cdot) {U}_\Cspace^\dag{\vphantom{U_\Cspace}^{}}^{\otimes n_\Cspace}$. 
\begin{prop}[Encoder and decoder as designs]\label{prop:encoder decode designs}
	%Suppose there exists a linear isomorphism $V_{\rm P \mapsto L}$ s.t. $ U_{\rm L}= V_{\rm P \mapsto L} U_\Cspace V_{\rm P \mapsto L}^\dag$ for all $U_{\rm L}\in \grp{SU}(d)$ and $U_\Cspace\in \grp{SU}(d)$. Then,
	The encoder $\map{E}_\cov(X_{\rm L})$ and decoder $\map{D}_\cov(X_\Cspace)$ admit unitary $(n_\Cspace+n_{\rm R}/2+1)$-designs for all $X_{\rm L}\in \Lin(\mathcal{H}_{\rm L})$, $X_\Cspace\in \Lin(\mathcal{H}_\Cspace^{\otimes n})$.
\end{prop}

\proof
	We start with the encoder. Expanding the definition Eq. \eqref{E-cov}, one can write 
	\begin{widetext}
	\begin{align}\label{encoder 2 proof}
	\map{E}_\cov(\cdot) &= \mathbb{E}_{U_\ph \sim {\rm Haar}} \!\! \left[\left( U_\ph^{\otimes n_\Cspace} {\map{E}}\left( U_{\rm L}^\dag  (\cdot) U_{\rm L}\right) {U_\ph^\dag}^{\otimes n_\Cspace}\right) \otimes\left( \Big({\id_\ph\otimes U_\ph}\Big) ^{\!\otimes n_{\rm R}/2} \psi \left({\id_\ph\otimes U_\ph^\dag}\right) ^{\!\otimes n_{\rm R}/2}\right) \right].
	\end{align}	
	\end{widetext}
	Recall from section \ref{sec:preliminary} that the representation of $U$ used on all the individual qudits is the same faithful representation of $\grp{SU}(d)$. As such, the logical and computational representations are related via a linear isomorphism, $V_{\Cspace \mapsto \rm L}$, satisfying  $ U_{\rm L}= V_{\Cspace \mapsto \rm L} U_{\ph \,} V_{\Cspace \mapsto \rm L}^\dag$ for all logical and computational unitary representations $U_{\rm L}$ and $U_{\ph\,} $. 
	Since the Hilbert spaces involved are finite dimensional, and $V_{\Cspace \mapsto \rm L}$, $\map{E}$ are linear maps, it follows that the matrix entries of the terms in square brackets in Eq. \eqref{encoder 2 proof}, are polynomials of order $(n_\Cspace+n_{\rm R}/2+1)$ in the coefficients of $U_\ph$ and of the same order in the matrix entries of $U_\ph^*$. Hence it is a unitary $(n_\Cspace+n_{\rm R}/2+1)$-design.
	
	In the case of the decoder, we have
	\begin{align}\label{decoder 2 proof}
	\map{D}_\cov(\cdot)=\mathbb{E}_{U_\ph \sim {\rm Haar}} \!\! \left[{\map{U}}_{\rm L}\circ\map{D}\circ{\map{U}}_\Cspace^{-1}\otimes\map{M}_{{U}_{\ph}}\right] (\cdot),
	\end{align}
	where 
	\begin{align}\label{measure map for proof}
	\map{M}_{{U}_{\ph\,}}(\cdot):=\sum_{\settt{s}\subset\{1,\dots, s_{\rm R}\}}\Tr\left[(\cdot)|\eta_{{U}}\>\<\eta_{{U}}|_{\settt{s}}\otimes P_{{\rm err},\settt{s}^c}\right],
	\end{align}
	with $|{\eta_{{U}}}\>=(U_\ph\otimes\id_\ph)^m |{\eta_{0}}\>$, and $|{\eta_{0}}\>:=\bigoplus_{\Vec\lambda}d_{\Vec\lambda}|{\Phi^+_{\Vec\lambda}}\>\otimes |{\Phi^+_{m_\Vec\lambda}}\>$,\quad  $s_{\rm R}\cdot 2m=n_{\rm R}$.
	Similarly to as in Eq.~\eqref{encoder 2 proof}, the term ${\map{U}}_{\rm L}\circ\map{D}\circ{\map{U}}_\Cspace^{-1}(\cdot)$ in Eq.~\eqref{decoder 2 proof}, when evaluated on any input, has matrix entries which are polynomials of order $n_{\ph\,}+1$ in the coefficients of $U_{\ph\,}$ and also of the same order in the coefficients of $U_{\ph\,}^*$. Hence observing the form of Eq.~\eqref{measure map for proof}, we conclude the proof.
\qed

Ref.\ \cite{CMPBrandao} devises and quantifies a method to approximate unitary $t$-designs. We start with a strategy to construct a random unitary $U$  over $N$ qubits on sites labelled 1 to $N$: Pick an index $l$ uniformly from $[N-1]$ and a unitary denoted $U_{l,l+1}$, drawn from the Haar measure on $\grp{SU}(4)$, which acts on the two neighbouring qubits, $l$ and $l+1$. Repeat the above $k$ times and multiply the unitaries together. The resultant unitary is $U$ and it is sampled from a distribution which we denote $\nu(k)$.

In our case if the qudits each consist of $N$ qubits, we can use this procedure to construct an approximate covariant encoding as follows: We record the random sequence of $k$ nearest neighbour unitaries $U_{l,l+1}$. We first apply it %to $\rho_{\rm L}$
to the logical qubits and encode via $\map{E}$, and then apply the recorded random sequence $n_\Cspace+ n_{\rm R}/2$ times to the computational and reference frame qubits. The resultant encoder is

% two copies of $U_{l,l+1}$ | one on the logical space and the other on the computational space | after $k$ iterations, we will have generated random unitaries $U_{\rm L}$ and $U_\Cspace$ respectively. If we then apply these analogously to how we did in the Haar random case when we constructed encoder $\map{E}_\cov$ and decoder $\map{D}_\cov$ in protocol \ref{protocol-cov}, we will have generated
%To start with, consider encoder and decoder channels $\map{E}^\nu_\cov$ and $\map{D}^\nu_\cov$ which are constructed analogously to $\map{E}_\cov$ and $\map{D}_\cov$ other than that they sample from a distribution $\nu$ over unitary representation of $\grp{SU}(d)$ rather than the Haar measure:
\begin{align}
\map{E}_\cov^{\nu(k)}(\cdot)&:=\mathbb{E}_{U\sim \nu(k)}\left[\map{U}_\Cspace\circ\map{E}\circ\map{U}_{\rm L}^{-1}(\cdot)\otimes\map{U}_{\rm R}(\Psi)\right].\label{E-cov approx}
\end{align}
Similarly, one can use the procedure to produce the approximate decoder,
\begin{align}
\map{D}_\cov^{\nu(k)}(\cdot)&:= \mathbb{E}_{\hat{U}\sim \nu(k)}\left[\left(\hat{\map{U}}_{\rm L}\circ\map{D}\circ\hat{\map{U}}_\Cspace^{-1}\otimes\map{M}_{\hat{U}}\right) (\cdot)\right].\label{D-cov approx}
\end{align}
From~\cite{CMPBrandao} it follows that
\begin{align}
\| \map{E}^{\nu(k)}_\cov - \map{E}_\cov \|_\diamond &\leq \epsilon_{\cov},\label{approx code}\\
\| \map{D}^{\nu(k)}_\cov - \map{D}_\cov \|_\diamond &\leq \epsilon_{\cov},\label{approx decode}
\end{align} 
if 
\begin{align}
k=& 170,000 \, N \lceil \log(4 (n_\Cspace+n_{\rm R}/2+1))\rceil^2 (n_\Cspace+n_{\rm R}/2+1) ^{8.1} \\
&\times \Big(  2 N (n_\Cspace+n_{\rm R}/2+1)+1 +\log(1/\epsilon_\cov) \Big)\notag,
\end{align}
which scales polynomially in both $n_\Cspace$ and $n_{\rm R}$ (recall lower bounds on $\epsilon_{\cov}$, in Propositions~\ref{prop-converse} and \ref{prop-converse-strong-error}%Lemma \ref{lemma-converse0}
).
Since sampling a polynomial number of times from $\grp{SU}(4)$ can be performed efficiently,  both the approximate encoder $\map{E}^{\nu(k)}_\cov$ and decoder $\map{D}^{\nu(k)}_\cov$ can be efficiently implemented in $n_\Cspace$ and $n_{\rm R}$ so long as the reference frame state $\Psi$ and measurement $\map{M}_{\hat{U}}$ can be efficiently constructed. It is also important to note that, under the application of an arbitrary number of transversal gates, the errors in Eqs. \eqref{approx code}, \eqref{approx decode} do not grow \cite{PhysRevX.7.021006}.

\section{Details of the numerical calculations}\label{sec-numerical}
In this section, we introduce the numerical experiments we implemented.

We use the ``5-qubit code" \cite{laflamme1996perfect} (5 computational qubits; one logical qubit) as the subroutine encoding and decoding pair, and calculate the performance of our $\grp{SU}(2)$ covariant code. 
The ``5-qubit code" can correct arbitrary single qubit errors on its code space. It is realized by the following encoding:
\begin{align*}
	|0_{\rm L}\> \rightarrow 1/4 \, &(
	|00000\> + |10010\> + |01001\> + |10100\> \\
	&+ |01010\> - |11011\> - |00110\> - |11000\> \\
	&- |11101\> - |00011\> - |11110\> - |01111\> \\
	&- |10001\> - |01100\> - |10111\> + |00101\>), \\
	|1_{\rm L}\> \rightarrow 1/4 \, &(
	|11111\> + |01101\> + |10110\> + |01011\> \\
	&+ |10101\> - |00100\> - |11001\> - |00111\> \\
	&- |00010\> - |11100\> - |00001\> - |10000\> \\
	&- |01110\> - |10011\> - |01000\> + |11010\>) \, .
\end{align*}
After this encoding, arbitrary single qubit errors can be detected with stabilizer measurements with 4 ancilla qubits. Depending on the measurement outcome, a correction will be performed, and then the original logical qubit will be obtained by doing the inverse of the encoding map.

In our experiment, not only erasure errors but also other common error types like dephasing errors and depolarising errors are considered.
For erasure errors, we use the generalised sine states as the reference frame state. This is given by %$|\phi\> \approx \sqrt{2/J} \sum_{j=j_{\rm min}}^{J} \frac{\sin (\pi j/J)}{\sqrt{2j+1}}|I^{(j)}\kk$. 
$|\phi\> \propto \sum_{j=j_{\rm min}}^{J} \frac{\sin (\pi j/J)}{\sqrt{2j+1}}|I^{(j)}\kk$. 
After encoding with $U_g$, it becomes 
%$|\phi_g\> \approx \sqrt{2/J} \sum_{j=j_{\rm min}}^{J} \frac{\sin (\pi j/J)}{\sqrt{2j+1}}|U_g^{(j)}\kk$. 
$|\phi_g\> \propto \sum_{j=j_{\rm min}}^{J} \frac{\sin (\pi j/J)}{\sqrt{2j+1}}|U_g^{(j)}\kk$.
The probability density function of the outcome $h$ with covariant POVM $\big\{d U_h, |\eta_h\>\<\eta_h|  \big\}$, $|\eta_h\> = \oplus_j |U_h^{(j)}\kk$ is then given by:
\begin{align}
p(h|g) = \dfrac{\Big|\sum_j \frac{ \sin(\pi j/J)}{\sqrt{2j+1}} \Tr\Big[U_{gh^{-1}}^{(j)}\Big] \Big|^2}{\Big|\sum_j \sin(\pi j/J) \cdot \sqrt{2j+1}  \Big|^2} \, .
\end{align}
\iffalse
\begin{align}
	p(h|g) = \Big|\sum_j \frac{\sqrt{2/J} \cdot \sin(\pi j/J)}{\sqrt{2j+1}} \Tr\Big[U_{gh^{-1}}^{(j)}\Big] \Big|^2 \, .
\end{align}
\fi

For i.i.d.\null{} depolarizing and dephasing errors, we use $n_{\rm R}/2$ Bell states $|\phi\> = |\Phi^{+}\>^{\otimes n_{\rm R}/2}$ as the reference frame state. 
After encoding with $U_g$, it becomes $|\phi_g\> = |\Phi^{+}_g\>^{\otimes n_{\rm R}/2}$. 
When measuring these states, we will get a set of $U_h$ with $n_{\rm R}/2$ elements: $\{U_{h,1}, ...,  U_{h, n_{\rm R}/2}\}$. 
From this set we will get a final $U_h$, which maximises $\<\Phi^{+}_h| \Big( \sum_{i=1}^{n_{\rm R}/2} |\Phi^{+}_{h,i}\>\<\Phi^{+}_{h,i}| \Big) |\Phi^{+}_h\> $, and can be regarded as the best representation of these measurement outcomes.

We also tested our protocol in the situation where one fifth of the reference frame qubits go through the completely depolarising/dephasing error channel. 
For this case, an algorithm similar as majority voting is first used. 
For each $U_{h,i}$, we calculate $\<\Phi^{+}_{h,i}| \Big( \sum_{i=1}^{n_{\rm R}/2} |\Phi^{+}_{h,i}\>\<\Phi^{+}_{h,i}| \Big) |\Phi^{+}_{h,i}\> $. The smallest one fifth outcomes will be regarded as affected by noise. 
The remaining set will be used to get the final $U_h$ which is the best representation of these outcomes.

After getting $U_h$, we perform the decoding according to it and then calculate the performance. We analyse how the performance changes as a function of the size of the reference frame state.

To calculate the performance, we first calculate the entanglement fidelity as defined in Eq. (\ref{F-ent}). For the 5-qubit code, we get:
\begin{align}
	F_{\rm ent, 5-code} = \int dg \, p(h|g) \, F_{\rm ent, 5-code}(g,h) \, ,
\end{align}
where $p(h|g)$ is the probability to get measurement outcome $U_h$ with encoding $U_g$, and $F_{\rm ent, 5-code}(g,h)$ is the entanglement fidelity using $U_g$ in encoding and $U_h$ in decoding. 
For every reference frame size considered ($2m$ for the sine state and $n_{\rm R}$ for the Bell states), we randomly generated $400$ unitaries $U_g$ from the Haar measure, and calculated their average entanglement fidelity, giving us the final $F_{\rm ent, 5-code}$. 

To calculate the worst-case error $\epsilon_{\rm cov, 5-code}$, we use the relationship between the worst-case error and the entanglement error in Eq. (\ref{relation_wc_ent}). 
Notice that the whole process of encoding, noisy evolution, and decoding in the considered case is a qubit covariant channel. Therefore, its Choi state can always be decomposed as $\lambda |\Phi^{+}\>\<\Phi^{+}| + \frac{(1-\lambda)}{3} (I - |\Phi^{+}\>\<\Phi^{+}|)$, for a $\lambda \in [0,1]$. Thus its worst-case error is equal to one minus the entanglement fidelity. 
Thus we can bound the worst-case error by  $\epsilon_{\rm cov, 5-code} \leq 2 \cdot (1-F_{\rm ent, 5-code})$.

%\bibliographystyle{apsrev4-1}%{plain}%{unsrt}   
%\bibliography{covqec}   
 
%\newpage 
\bibliographystyle{unsrt}
\bibliography{CovQEC}%.bib} 
\end{document}